\documentclass[a4paper,fleqn]{cas-dc}

\usepackage{xpatch}

\ExplSyntaxOn

\cs_set:Npn \__first_footerline: {}

\xpatchcmd{\author}
{\l_stm_augroup_color_tl !50}
{\l_stm_augroup_color_tl}
{}{}

\xpatchcmd{\author}
{\l_stm_augroup_color_tl !50}
{\l_stm_augroup_color_tl}
{}{}

\ExplSyntaxOff

\usepackage[numbers]{natbib}

\usepackage{graphicx}
\usepackage{multirow}
\usepackage[normalem]{ulem}
\usepackage{amsmath,amssymb,amsfonts}
\usepackage{amsthm}
\usepackage{mathrsfs}
\usepackage[title]{appendix}
\usepackage{xcolor}
\usepackage{textcomp}
\usepackage{manyfoot}
\usepackage{booktabs}
\usepackage{algorithm}
\usepackage{algorithmicx}
\usepackage{algpseudocode}
\usepackage{listings}
\usepackage{bm}
\usepackage{multirow}
\usepackage{moresize}
\usepackage{algorithm}
\usepackage{algpseudocode}
\usepackage{glossaries-extra}
\usepackage{float}
\usepackage{cancel}
\usepackage{subfig}
\usepackage{framed}
\usepackage{dashbox}
\usepackage{url}
\usepackage[
    n,
    operators,
    advantage,
    sets,
    adversary,
    landau,
    probability,
    notions,    
    logic,
    ff,
    mm,
    primitives,
    events,
    complexity,
    asymptotics,
    keys]{cryptocode}

\usepackage{moresize}
\usepackage{algorithm}
\usepackage{algpseudocode}
\usepackage[dvipsnames]{xcolor}
\usepackage{threeparttable}

\newtheorem{theorem}{Theorem}%  meant for continuous numbers
\newtheorem{proposition}{Proposition}

\newtheorem{definition}{Definition}
\newtheorem{lemma}{Lemma}

\newtheorem{claim}{Claim}
\newcommand{\Adv}{\mathsf{Adv}}
\newcommand{\ngl}{\mathsf{negl}}

\def\tsc#1{\csdef{#1}{\textsc{\lowercase{#1}}\xspace}}
\tsc{WGM}
\tsc{QE}

\begin{document}
\let\WriteBookmarks\relax
\def\floatpagepagefraction{1}
\def\textpagefraction{.001}

\shorttitle{}    

\shortauthors{}  

\title [mode = title]{Pushing Forward Multi-Secret-Key Homomorphic Encryption for Private Average Aggregation}  

\author[1]{Miguel Morona-Mínguez}[orcid=0009-0000-3759-1420]%[<options>]

%\cormark[1]


\ead{mmorona@gts.uvigo.es}


\credit{Formal analysis, Investigation, Software, Validation, Visualization, Writing - Original draft preparation}

\affiliation[1]{organization={atlanTTic, Universidade de Vigo},
            addressline={Escola de Enxeñaría de Telecomunicación}, 
            city={Vigo},
%          citysep={}, % Uncomment if no comma needed between city and postcode
            postcode={36310}, 
            state={Galicia},
            country={Spain}}

\author[1]{Fernando Pérez-González}[orcid=0000-0002-0568-1373]%[]

% Footnote of the second author

% Email id of the second author
\ead{fperez@gts.uvigo.es}


\credit{Formal analysis, Funding acquisition, Methodology, Resources, Supervision, Validation, Writing - review and editing}

\author[1]{Alberto Pedrouzo-Ulloa}[orcid=0000-0002-7283-6275]%[]
\cormark[1]
% Footnote of the third author
%\fnmark[3]

% Email id of the third author
\ead{apedrouzo@gts.uvigo.es}

% URL of the third author

\credit{Conceptualization of this study, Formal analysis, Investigation, Methodology, Supervision, Validation, Software, Writing - Original draft preparation, Writing - review and editing}

\cortext[1]{Corresponding author}



\begin{abstract}
Federated Learning enables multiple clients to train a shared model while keeping their local datasets isolated. However, the exchanged model updates may still leak sensitive information, making private aggregation a central building block in practical deployments, especially in the cross-silo setting. Homomorphic Encryption naturally fits the client--aggregator communication pattern of Federated Learning, but conventional single-key deployments rely on strong non-collusion assumptions. Multiparty Homomorphic Encryption removes this limitation, although recent attacks under restricted decryption access require large-variance smudging noise during collaborative decryption, which significantly increases ciphertext size and implementation complexity. In this work, we propose lightweight multi-secret-key protocols for private average aggregation based on RLWE-based Homomorphic Encryption. Our construction departs from the usual multiparty blueprint by avoiding the generation of a collective public key. Instead, each client encrypts its update under its own secret key, while the resulting ciphertexts remain compatible with homomorphic aggregation and collaborative decryption. By explicitly tracking and cancelling the ciphertext noise during decryption, the protocol removes the need for large $\lambda$-dependent smudging noise. We instantiate the construction with both exact BFV-based and approximate CKKS-based variants, prove its security in the semi-honest model against an adversary corrupting the aggregator and up to $L-1$ clients, and compare its communication and runtime performance with state-of-the-art MHE-based aggregation. Our results show that the proposed approach substantially reduces ciphertext expansion and online cost, while preserving practical homomorphic aggregation performance.
\end{abstract}


% Keywords
\begin{keywords}
Federated Learning (FL) \sep Private Aggregation \sep Homomorphic Encryption (HE) \sep Ring Learning with Errors (RLWE) \sep Multi-Key Encryption
\end{keywords}

\maketitle

\section{Introduction}
\label{sec:intro}
Federated Learning (FL), first introduced in~\cite{MMRHA17}, is a Machine Learning (ML) framework in which multiple clients collaboratively address a learning task under the coordination of a central server, typically referred to as the \emph{aggregator}. A distinctive property of FL is that the training data are distributed across different clients, such that each dataset remains stored locally and never leaves the corresponding client’s premises. While FL was initially proposed in~\cite{MMRHA17} as a decentralized collaborative solution to avoid the prohibitive communication cost of collecting training data from many clients, it also provides a natural mechanism to reinforce data control and privacy protection by keeping local datasets isolated within each client’s premises~\cite{Kairouzetal21}.

\textbf{Federated Learning process:} Starting from an initial common ML model parameterized by $\vec{w}^*$, an FL protocol typically proceeds through successive communication rounds.  In each round, $L$ participating clients, indexed by $i \in \{1,\ldots,L\}$, perform the following steps: (1) clients locally update the common model by partially training it using only their individual datasets, thereby producing local updates $\vec{w}_i$, and (2) the central server aggregates all received local updates to compute a new global model. This aggregation step is usually performed by evaluating an aggregation function of the form $\vec{w}^* = f_{\Sigma}(\vec{w}_1,\ldots,\vec{w}_L)$. The resulting global model is then distributed to the clients for the next communication round. The process is repeated until a predefined convergence criterion is satisfied.

\textbf{Privacy issues:} Unfortunately, even though clients’ raw data are not explicitly shared during training, privacy risks in FL still persist. Several works have shown that the exchanged local model updates may leak information about the data used during training, making them vulnerable to different types of inference attacks, particularly when adversaries have white-box access to both individual and aggregated updates~\cite{NSH19,LF20}. This issue is especially relevant for the aggregator, which, unlike other participants, has access to all individual updates sent by the clients.

In view of this situation, additional protection mechanisms are required to mitigate residual privacy leakage, highlighting the importance of incorporating additional Privacy-Enhancing Technologies (PETs)~\cite{PRPLDCGUHMPLLMJ23} into the FL pipeline. In particular, private aggregation solutions~\cite{Bonawitzetal17} prevent any single participant from accessing individual local updates, thereby reducing the level of trust required from honest participants, especially with respect to the aggregator.

As surveyed in~\cite{MOJC23}, a wide range of PETs can be employed for this purpose. However, no single candidate suffices in all cases, since different PETs exhibit distinct trade-offs in terms of communication, computation, and robustness. Moreover, FL deployments differ in their requirements and are commonly categorized into two high-level settings: the \emph{cross-device} setting, in which the model is trained using data from many intermittently available and resource-constrained devices, and the \emph{cross-silo} setting, in which the model is trained using data from a smaller number of always-available and computationally powerful servers.

\textbf{Scope of our work:} In this paper, we focus on the cross-silo setting, in which a set of stateful clients collaboratively train and improve a shared model using their collective data, while the aggregator coordinates the process but is not necessarily allowed to learn either intermediate or final model parameters. First, we analyze in detail the privacy and functional requirements that private aggregation mechanisms must satisfy in the cross-silo FL context; a review of existing definitions in the literature is provided in Section~\ref{sec:background1}. Second, we concentrate on a specific PET, namely Homomorphic Encryption (HE)~\cite{MTBH21,AHSL22} and, after studying its security properties, focus on optimizing its use for private aggregation in cross-silo FL scenarios.

\subsection{FL Privacy via Homomorphic Encryption}

HE-based private aggregation naturally aligns with FL's client–aggregator communication pattern. First, clients encrypt their local updates. Then, the aggregator homomorphically evaluates the aggregation function $f_{\Sigma}(\vec{w}_1,\ldots,\vec{w}_L)$. Finally, the clients decrypt the aggregated update.

\textbf{Multiple keys and HE-based aggregation:} Several works, such as~\cite{ZLXWYL20,FQ21}, consider a simple deployment scenario in which all clients share the same secret and public keys. This implicitly imposes a strong non-collusion assumption between the aggregator and any party with access to the secret key (see Section~\ref{sec:background2}). To relax this assumption, HE schemes can be extended to support ciphertexts encrypted under multiple keys~\cite{AHSL22}. While several approaches already allow the homomorphic aggregation of ciphertexts encrypted under different keys~\cite{MNSL22,PBCSSZ23}, the most versatile class of solutions corresponds to threshold variants of single-key HE. A representative example is Multiparty Homomorphic Encryption (MHE), as proposed in~\cite{MTBH21}, where a set of secret keys and a collective public key are generated during an offline setup phase prior to the execution of the FL protocol~\cite{MPP24}. These multi-key constructions require decryption to be performed collaboratively by all participating clients (see Section~\ref{sec:mhe}).

\subsection{Contributions and improvements over~\cite{MPP24}}

In our preliminary work~\cite{MPP24}, we surveyed the MHE state of the art and analyzed its privacy vulnerabilities in the FL setting. We showed how MHE can be securely employed for private aggregation. In particular, removing the non-collusion assumption among participants when using MHE, whether based on CKKS or BFV, requires adding a smudging noise with variance proportional to $2^\lambda$ during decryption. Its central result was a detailed theoretical and experimental comparison of two MHE instantiations based on approximate CKKS and exact BFV single-key HE schemes, showing that CKKS-based MHE can match or even exceed the performance of BFV-based alternatives. This contrasts with previous observations in works such as~\cite{CSBBC24}, particularly when a large smudging noise is required.

For a statistical security parameter $\lambda$, the use of a smudging noise with variance $2^\lambda$ times larger introduces an overhead of $\mathcal{O}(\lambda)$ bits in the size of the encrypted updates exchanged during each aggregation round. For practical and cryptographically secure choices, such as $\lambda \geq 128$, this overhead can easily double the transmitted data size. Beyond its direct impact on efficiency, this also introduces additional implementation challenges, since noise generation must handle values exceeding the typical machine word size.

\textbf{Main contributions:} The efficiency degradation caused by the large smudging noise motivates the search for alternative HE-based solutions that preserve security while avoiding this additional cost. In this work, we propose more efficient and lightweight private aggregation protocols based on two families of HE schemes supporting multiple secret keys, corresponding to approximate and exact HE. Following~\cite{PBCSSZ23}, our approach relies on a lightweight setup phase using additive secret sharing. More importantly, it departs from the blueprint of~\cite{MTBH21,MPP24} by eliminating the need to generate a collective public key. The main ideas and contributions of this work can be summarized as follows:
\begin{itemize}
\item We introduce a new protocol tailored to average aggregation that extends single-key HE schemes, either approximate CKKS-based or exact BFV-based, to support multiple secret keys. This design removes the $\mathcal{O}(\lambda)$ overhead by allowing the use of a smudging noise whose variance is independent of $\lambda$. As in~\cite{PBCSSZ23}, the protocol is limited to homomorphic averaging, yet provides important advantages over the state of the art. First, it does not require a public key. Second, each client encrypts updates using its own secret key, which allows more compact ciphertexts and slightly improves performance. Third, the smudging noise has the same variance as the cryptosystem error distribution and is independent of $\lambda$. This property enables a significant reduction in ciphertext size compared to previous approaches~\cite{MPP24,PBCSSZ23}.
\item We formally prove the security and privacy of the proposed multi-secret-key protocol for private average aggregation using a simulation-based argument, assuming a semi-honest aggregator and a majority of semi-honest clients.
\item We provide concrete estimates of the overhead introduced by the proposed HE-based protocol and compare them with state-of-the-art MHE approaches, which incur an asymptotic overhead of $\mathcal{O}(\lambda)$ bits. We extend the performance evaluation of~\cite{MPP24} by presenting more detailed experimental results comparing approximate and exact MHE aggregation. In addition, we compare both approaches with our proposed secret-key HE-based aggregation in terms of bit precision and ciphertext size. We also analyze the parameter regimes in which the approximate or exact variants are preferable and report concrete implementation runtimes across practical parameter ranges.
\end{itemize}

Although our solutions are primarily designed for the cross-silo setting, they naturally extend to cross-device scenarios by re-running a lightweight setup phase per aggregation round. This allows the protocol to accommodate intermittent client participation. The additional setup phase incurs minimal overhead, provided that clients have sufficient computational and memory resources to perform additive share generation, encryption, and partial decryption.

\subsection{Notation and structure}
Polynomials are denoted by lowercase letters, omitting the polynomial variable when there is no ambiguity, for example, writing $a$ instead of $a(x)$. We also represent polynomials as column vectors of coefficients $\vec{a}$.  Scalar vectors are denoted with arrows (e.g., $\vec{v}$), while vectors whose entries are polynomials are denoted in boldface (e.g., $\mathbf{v}$). Let $R_q = \mathbb{Z}_q[x]/(1 + x^n)$, where $q \in \mathbb{N}$ and $n$ is a power of two, denote the polynomial ring modulo $1 + x^n$ with coefficients in $\mathbb{Z}_q$, represented in the interval $(-q/2, q/2]$. Let $\chi$ denote the error distribution over $R_q$, whose coefficients are independently sampled from a discrete Gaussian distribution with standard deviation $\sigma$ and truncated support $[-B, B]$. For a finite set $\mathcal{S}$, the notation $x \leftarrow \mathcal{S}$ denotes uniform sampling from $\mathcal{S}$, while $x \leftarrow \chi$ denotes sampling from the distribution $\chi$. Finally, both $||\vec{a}||$ and $||a||$ denote the infinity norm of $\vec{a}$.

The rest of the document is organized as follows: Section~\ref{sec:background1} surveys existing security definitions for private aggregation. Section~\ref{sec:background2} reviews approximate and exact single-key HE schemes and their application to private aggregation under simpler threat models, including a discussion of privacy vulnerabilities arising from restricted decryption access. Section~\ref{sec:mhe} introduces Multiparty HE as an extension of HE supporting multiple secret keys in the FL setting and analyzes current countermeasures against the attacks described in~\cite{CSBBC24,CCPSS24}, together with analytical comparisons between Multiparty BFV and Multiparty CKKS. Section~\ref{sec:prop-opt} presents our proposed solution, which does not rely on a large-variance smudging error. Section~\ref{sec:perf-eval} evaluates and compares the performance of approximate and exact MHE with our proposed approach. Finally, Section~\ref{sec:conc-fut} concludes the paper and outlines future research directions.

\section{Background on Aggregation Threat Models}
\label{sec:background1}
This section surveys the definitions typically considered in the literature for private aggregation and examines how they align with the actual requirements of FL. Since this primitive was not originally introduced in the FL context, we first recall, in Section~\ref{sec:reqaggfunc}, the commonly used threat-model definitions, and then revisit them in Section~\ref{sec:suitabilitypriaggdef} to identify which assumptions should be refined for the FL setting considered in this work.

As introduced in Section~\ref{sec:intro}, the private computation of the aggregation functionality $f_\Sigma(\vec{w}_1, \ldots, \vec{w}_L)$ has become a key tool for mitigating privacy leakage from the exchange of ML model updates in FL protocols~\cite{Kairouzetal21, NSH19}. While numerous private aggregation algorithms have been proposed in the literature~\cite{MOJC23}, covering various properties and formal security definitions~\cite{SCRCS11, MOJC23}, the appropriate definition of a private aggregation primitive depends on the \emph{specific threat model} and the \emph{particular requirements} of the target FL training pipeline. In light of this, we aim to clarify next not only which functions $f_\Sigma$ are typically considered, but also what is commonly understood by computing them privately.

This discussion also motivates the threat models and security limitations analyzed in Section~\ref{sec:background2} for a first, straightforward HE-based implementation of private aggregation in FL. More advanced threat models, involving stronger semi-honest adversaries and requiring the use of multiple secret keys, are considered in Sections~\ref{sec:mhe} and~\ref{sec:prop-opt}.

\subsection{Threat Models for Private Aggregation}
\label{sec:reqaggfunc}

In the FL literature, the aggregation of the local updates sent by the clients to the aggregator is commonly based on linear functions~\cite{Kairouzetal21}. We denote by $f_\Sigma(\vec{w}_1,\ldots,\vec{w}_L) = \sum_{i=1}^{L}\lambda_i\vec{w}_i$ a general linear aggregation function, where the $\lambda_i$, $1=1, \ldots, L$, are arbitrary public coefficients. This formulation includes, as common special cases, the sum, obtained for $\lambda_i=1$ for all $i$, and the average, obtained for $\lambda_i=1/L$ for all $i$. Most private aggregation proposals aim to securely compute a linear function of this form~\cite{MOJC23}.

The term \textit{secure aggregation} first appeared outside FL, in the context of Wireless Sensor Networks (WSNs), for the private aggregation of data generated by multiple nodes~\cite{HE03,MOJC23,SCRCS11}. It was later applied to smart grids and smart meters~\cite{OM07, CMT05, KDK11}, leading to a growing body of related solutions, including privacy-preserving aggregation~\cite{SCRCS11}, privacy-friendly aggregation~\cite{KDK11}, and private stream aggregation~\cite{CSS12}. The first work discussing secure aggregation in the FL setting appeared in~\cite{Bonawitzetal17}. For more on the history of secure aggregation and its applications beyond FL, we refer the reader to~\cite{MOJC23}. As that work highlights, \textbf{current secure aggregation protocols for FL still follow the blueprint established by early WSN solutions}~\cite{MOJC23}, which continues to shape the requirements considered today.

In particular,~\cite{MOJC23} outlines two main security definitions that an ideal secure aggregation protocol should satisfy under a strict threat model involving dishonest adversaries. While our focus here is on the semi-honest model,\footnote{A participant is semi-honest if they follow the protocol but try to infer additional information from the messages received.} we include for completeness the corresponding definition for malicious behaviors. These notions serve as a baseline for the FL-specific discussion in Section~\ref{sec:suitabilitypriaggdef}. We begin with their requirement for semi-honest adversaries:
\begin{definition}[Aggregator Obliviousness~\cite{SCRCS11}. Adapted with major changes from~\cite{MOJC23} and~\cite{SCRCS11}]\label{def:aggobliviousness} This definition captures the following security notions:
\begin{itemize}
\item An honest-but-curious aggregator cannot learn more than what can be inferred from the sum of the clients’ inputs (``aggregator capability''). 
\item Clients alone (even when colluding among them) cannot learn anything about the aggregation result or the updates of honest clients (as they lack the ``aggregator capability'').
\item If clients collude with the aggregator by sharing their private inputs, the adversary learns nothing beyond what can be inferred from the aggregated value and the inputs of the colluding clients.
\end{itemize}
This definition is formalized in~\cite{SCRCS11} via a security game (Aggregator Obliviousness or AO security game), which can be extended to cover more general statistics. We refer the reader to~\cite{SCRCS11} for further details. 
\end{definition}

Beyond the semi-honest case,~\cite{MOJC23} introduces the following notion for entities that may deviate from the protocol:
\begin{definition}[Aggregate Integrity~\cite{MOJC23}. Adapted with minor changes from~\cite{MOJC23}]\label{def:aggintegrity}
This security notion ensures that:
\begin{itemize}
\item A malicious aggregator cannot forge an incorrect aggregation result without being detected by the clients. A forgery is defined as an aggregation outcome that differs from the true sum of the clients’ inputs.
\item A group of malicious clients cannot compromise the aggregation result as long as a sufficient number of honest clients (above a certain threshold) participate. A compromise is defined as a significant deviation of the aggregation outcome from the sum of the honest clients’ inputs.
\end{itemize}
\end{definition}

Implementing secure aggregation primitives satisfying Definition~\ref{def:aggintegrity} typically requires mechanisms for clients to verify the aggregator's behavior~\cite{GLLGHDB21} and ensure that dishonest clients' inputs meet validity conditions~\cite{LBVKH23}. While such verification helps limit malicious-client drift, linear aggregation functions are not well-suited to Byzantine clients, since even one Byzantine participant can manipulate the result and prevent convergence of the trained FL model~\cite{BEGS17}.

A common alternative is to replace the average with a more robust aggregation function, such as the median, trimmed mean, or multi-Krum~\cite{BEGS17}, among others~\cite{FGGPS22}. However, it remains unclear whether these robust functions are compatible with existing secure aggregation primitives~\cite{MOJC23}, which are primarily designed for linear aggregation and, as discussed, remain vulnerable to Byzantine clients.

Due to these limitations, \emph{we consider a passive threat model where clients behave semi-honestly. Moreover, since average functions are widely adopted in practical FL deployments, this work focuses on the private implementation of conventional linear aggregation functions}.

\subsection{FL-tailored Private Aggregation Definitions}
\label{sec:suitabilitypriaggdef}

We now revisit Definitions~\ref{def:aggobliviousness} and~\ref{def:aggintegrity} from the perspective of FL. Although these definitions provide clear guidelines regarding the properties a secure aggregation functionality should satisfy, they are strongly influenced by the context of wireless sensor networks, where a third party computes the aggregation of data from multiple nodes with no further interaction between the nodes and the aggregator. While this setting does not exactly align with the specific requirements of FL, relevant surveys such as~\cite{Kairouzetal21} already assume secure aggregation as a foundational component in FL, typically requiring that the server \emph{learn no more than an aggregated function} of the clients’ inputs. We therefore examine whether all the properties of this building block are truly suitable for FL, by revisiting the general definition provided in~\cite{Kairouzetal21}:
\begin{definition}[FL definition taken literally from~\cite{Kairouzetal21}]\label{def:FL}
``Federated learning is a machine learning setting where multiple entities (clients) collaborate in solving a machine learning problem, under the coordination of a central server or service provider. Each client's raw data is stored locally and not exchanged or transferred; instead, focused updates intended for immediate aggregation are used to achieve the learning objective.''
\end{definition}

We note that Definition~\ref{def:FL} does not explicitly require the aggregator to learn the intermediate aggregation result during training. Rather, its role is to coordinate the aggregation process. This distinction is important for determining which properties a private aggregation functionality should satisfy in FL. First, allowing the aggregator to inspect aggregated updates at each round increases the privacy leakage already associated with the exchange of model updates~\cite{Kairouzetal21, NSH19}. From a privacy perspective, this exposure should therefore be minimized. Second, in FL, the aggregator coordinates the aggregation rounds and forwards the result to the clients, who are the actual recipients of the intermediate aggregation outputs. In contrast, Definition~\ref{def:aggobliviousness} assumes a setting where clients do not receive the aggregation result. Early secure aggregation protocols thus appear to follow a one-way communication model from clients to the aggregator, always revealing the aggregated output to the server.

As a result of this modeling choice, several potentially valid approaches are often excluded, even though their guarantees may be appropriate in practical FL applications. A clear example is the cross-silo FL setting~\cite{ZLXWYL20}, where clients typically have more computational resources and influence over training than in cross-device settings. Such deployments often involve institutions, such as hospitals or financial organizations, collaborating over distributed private datasets~\cite{CVC+23,FTR+21,IGM22}. In this setting, clients are often considered less adversarial, while the aggregator remains a high-risk entity due to its access to intermediate updates during training. Accordingly, recent works, including~\cite{CVC+23,UNGuideOnPETs23}, frequently adopt more relaxed threat models for clients. In hospital collaborations, for instance, institutions may be reasonably trusted because they are known to each other and can establish agreements before training; legal safeguards and shared interests further support semi-honest behavior.

In addition to private aggregation, complementary countermeasures such as Differential Privacy (DP), or metrics to evaluate privacy risks in the aggregated result (output privacy)~\cite{UNHandbookPETs23}, can be applied beforehand by the clients to decide whether to participate in training~\cite{PRPLDCGUHMPLLMJ23}. We treat this issue as independent from privately computing the aggregation itself (input privacy)~\cite{UNHandbookPETs23}. Therefore, we focus on the private computation of $\vec{w}^* = f_\Sigma(\vec{w}_1, \ldots, \vec{w}_L)$, regardless of how much information the aggregated result may reveal about the training data (output privacy)~\cite{MOJC23,UNGuideOnPETs23}. In particular, in our preliminary work~\cite[Definition~2]{MPP24}, we considered the following general multiparty aggregation problem~\cite{MTBH21}, where multiple clients, each owning a locally stored and isolated dataset, collaborate using their private local updates:
\begin{definition}[Adapted from Def.~$1$ in~\cite{MTH19}]\label{def:mpcagg}
Let $\mathcal{C} = \{C_1, C_2, \ldots, C_L\}$ be a set of $L$ clients, where each client $C_i$ holds an input $\vec{w}_i$ and receives an output $\vec{w}^*$. Let $f_\Sigma(\vec{w}_1, \vec{w}_2, \ldots, \vec{w}_L) = \vec{w}^*$ be the average aggregation function ($f_\Sigma$ is referred to as the ideal functionality) over the input parties. Let $\mathcal{A}$ be a static semi-honest adversary that can corrupt up to $L - 1$ clients in $\mathcal{C}$, and let $\mathcal{C_{A}}$ be the subset of clients corrupted by $\mathcal{A}$. Then, the \textbf{secure multiparty aggregation problem} consists in providing $\mathcal{C}$ with $\vec{w}^*$, ensuring that $\mathcal{A}$ learns nothing more about $\{ \vec{w}_i \}_{C_i \notin \mathcal{C_{A}}}$ than what can be inferred from the inputs $\{ \vec{w}_i \}_{C_i \in \mathcal{C_{A}}}$ and the output $\vec{w}^*$ it controls. This condition ensures input privacy.
\end{definition}

The formulation in Definition~\ref{def:mpcagg} deliberately abstracts away from the aggregator, even though the aggregator is part of the FL protocol and may participate in the computation for computational or communication efficiency. This omission in~\cite{MPP24} emphasized that, in the cross-silo setting~\cite{ZLXWYL20}, clients are the sole owners of the data whose privacy must be protected and the recipients of the intermediate aggregation results during FL training. In this setting, the aggregator acts neither as an input party nor, in general, as a receiver of the intermediate outputs. However, in certain aggregation executions, such as at the end of training, \textbf{the aggregator may also receive the aggregated result.} For this reason, we refine the formulation here to reflect that possibility:
\begin{definition}[Adapted from Def.~\ref{def:mpcagg} to our FL setting]\label{def:mpcaggfl}
Let $Agg$ be the aggregator, which may act as a receiver party when applicable, and let $\mathcal{C} = \{C_1, C_2, \ldots, C_L\}$ be a set of $L$ clients, where each client $C_i$ holds an input $\vec{w_i}$ and acts as an input and receiver party. Let $f_\Sigma(\vec{w}_1, \vec{w}_2, \ldots, \vec{w}_L) = \vec{w}^*$ be the average aggregation function, referred to as the ideal functionality, over the input parties. Let $\mathcal{A}$ be a static semi-honest adversary that may corrupt up to $L - 1$ clients in $\mathcal{C}$, and potentially $Agg$ as well. Let $\mathcal{C_{A}}$ denote the subset of clients corrupted by $\mathcal{A}$. Then, the \textbf{secure multiparty aggregation problem in FL} consists of providing $\vec{w}^*$ to the clients in $\mathcal{C}$, and to $Agg$ if required in that execution, while ensuring that $\mathcal{A}$ learns nothing more about $\{ \vec{w}_i \}_{C_i \notin \mathcal{C_{A}}}$ than what can be inferred from the corrupted inputs $\{ \vec{w}_i \}_{C_i \in \mathcal{C_{A}}}$ and the output $\vec{w}^*$ it observes. In this problem, $Agg$ coordinates and assists the clients in the computation of $\vec{w}^*$. This condition ensures input privacy.
\end{definition}

A solution to the \textbf{secure multiparty aggregation problem in FL} consists of a protocol $\pi_{f_\Sigma}$ that realizes the ideal functionality $f_\Sigma$ while preserving input privacy~\cite{Lindell17}. This definition makes clear that, while the aggregator may be involved in the execution of $\pi_{f_\Sigma}$ throughout the training process, it does not necessarily need access to the clients’ inputs. In some cases, output access can be allowed, e.g., access to the final trained model. Consequently, Definition~\ref{def:mpcaggfl} does not categorically prevent the aggregator from learning the output. Rather, it separates that aspect explicitly from input privacy. As we will show in the following sections, this distinction can significantly impact the implementation of $\pi_{f_\Sigma}$ when relying on modern HE schemes. Note also that some HE-based instantiations compute an approximation of $f_\Sigma$. For such instantiations, we keep $f_\Sigma$ as the exact mathematical aggregation functionality and distinguish input privacy from approximate correctness: the protocol output may be a bounded, randomized approximation of $f_\Sigma$, while the privacy guarantee is still defined with respect to the exact functionality $f_\Sigma$.

\section{Background on HE for FL}
\label{sec:background2}
The role of the aggregator discussed in Definitions~\ref{def:mpcagg} and~\ref{def:mpcaggfl} resembles that of a third party in outsourced computation, to whom clients delegate computational tasks. Moreover, as discussed above for Definition~\ref{def:aggobliviousness}, previous research~\cite{ZLXWYL20} has pointed out how certain FL settings, such as cross-silo, impose more stringent privacy requirements on the aggregator. A natural way to mitigate privacy leakage in such outsourcing scenarios is HE, which has already been considered for average aggregation in several works~\cite{SCRCS11,BJL16,ZLXWYL20,MSMSGS21}. Notably, HE-based aggregation is particularly well suited to these settings, as it satisfies these privacy constraints while imposing fewer constraints on the underlying FL system architecture than alternative approaches.

\textbf{Lattice-based HE:} Although HE has previously been considered in the context of FL, the adoption of modern lattice-based HE schemes, such as BFV, BGV, and CKKS~\cite{ACCDGGHH19,BCCCCDGHKKLLMPPLSYY24}, remains relatively recent~\cite{SPTFBSH21,SCSSG23}. Beyond their efficiency in implementing linear functions over encrypted long vectors via SIMD (Single Instruction Multiple Data), lattice-based HE schemes offer additional advantages. First, they are widely believed to be resilient against quantum adversaries; indeed, most proposals selected in the NIST Post-Quantum Cryptography Standardization Process are lattice-based~\cite{NIST-PQC}. Second, they provide sufficient expressiveness to homomorphically evaluate more complex functions than earlier schemes such as Paillier~\cite{Paillier99} or ElGamal~\cite{CGS97,Elgamal85}. As a result, while most prior work on private aggregation focuses on average computations~\cite{MOJC23}, recent studies have begun leveraging lattice-based HE to support more complex, Byzantine-resilient aggregation functions~\cite{CGPSSZ24}.

However, lattice-based HE schemes introduce challenges under multiple keys. We discuss them for public-key schemes in Section~\ref{sec:mhe}, and then introduce our proposed symmetric-key protocol in Section~\ref{sec:prop-opt}. Before that, as argued in Section~\ref{sec:suitabilitypriaggdef}, relaxed threat models can be addressed using single-key HE. In particular, when private channels exist between the aggregator and the clients and no collusions are assumed, single-key HE remains the simplest, most efficient. It is in this setting that our distinction in Definition~\ref{def:mpcaggfl}, on whether the aggregator learns the output, becomes crucial for protocol efficiency. We first recall HE building blocks in Section~\ref{sec:heblocks}. We then describe their use for private aggregation under a relaxed threat model with private channels in Section~\ref{sec:singlekeyhepriagg}. Finally, Section~\ref{sec:ind-cpa-d} analyzes the security and efficiency implications of aggregator access to the output.

\subsection{HE building blocks}
\label{sec:heblocks}

Since we focus on linear aggregation functions $\vec{w}^{*}=f_{\Sigma}(\vec{w}_1,\ldots,\vec{w}_L)=\sum_{i=1}^{L}\lambda_i\vec{w}_i$, with averaging being the most common in FL, we restrict ourselves to additive HE schemes.
\begin{definition}[Additive Homomorphic Encryption]\label{def:addhe}
    Let $\mathsf{E} = \{\mathsf{Setup}, \mathsf{SecKeyGen}, \mathsf{PubKeyGen}, \mathsf{Enc}, \mathsf{Dec}\}$ be an asymmetric encryption scheme, whose security is parameterized by $\lambda$. An additive homomorphic encryption scheme extends $\mathsf{E}$ with the $\mathsf{Add}$ procedure, having $\mathsf{E_{ahe}} = \{\mathsf{E}.\mathsf{Setup},$ $ \mathsf{E}.\mathsf{SecKeyGen}, \mathsf{E}.\mathsf{PubKeyGen}, \mathsf{E}.\mathsf{Enc}, \mathsf{E}.\mathsf{Dec}, \mathsf{Add}\}$, where
    \begin{itemize}
        \item Let $c_{\vec{v}}$ and $c_{\vec{w}}$ be two ciphertexts encrypting, respectively, $\vec{v}$ and $\vec{w}$. The procedure $\mathsf{Add}$ satisfies:
        \begin{equation} \label{eq:correctness-add}
            \mathsf{E_{ahe}}.\mathsf{Dec}(\mathsf{E_{ahe}}.\mathsf{Add}(c_{\vec{v}}, c_{\vec{w}})) = \vec{v} + \vec{w}.
        \end{equation}
    \end{itemize}
\end{definition}

We choose two \textit{state-of-the-art} RLWE (Ring Learning with Errors)-based cryptosystems to instantiate the $\mathsf{E}_\mathsf{ahe}$ scheme of Definition~\ref{def:addhe}. Specifically, they correspond to BFV~\cite{Brakerski12,FV12} and CKKS~\cite{CKKS17}, representing exact and approximate HE schemes, respectively. We refer the reader to Table~\ref{tab:bfv-ckks} for details on their primitives and parameters.

\begin{table*}[!ht]  
    \centering
    \renewcommand{\arraystretch}{2}
    \caption{Summary of $\mathsf{BFV}$ \& $\mathsf{CKKS}$ HE Schemes.}
    \label{tab:bfv-ckks}
    \scriptsize 
    \begin{tabular}{|m{1.85in}|m{0.65in}|m{3.70in}|}
        \hline
        \multicolumn{3}{|c|}{Parameters remarks}\\
        \hline
        \multicolumn{3}{|p{6.30in}|}{
        $\mathsf{BFV}$: $R_q$ is the ciphertext ring. $\Delta = \lfloor \frac{q}{t} \rfloor$. $R_t$ is the plaintext ring, with $q > t$.
        \newline $\mathsf{CKKS}$: $R_q$ is the ciphertext ring. Given an arbitrary target margin error $\epsilon$ and $B$, we can set the scale $\Delta$.
        }\\
        \hline
        
        \multicolumn{3}{|c|}{Cryptographic Primitives}\\
        \hline
        \multirow{2}{*}{$\mathsf{\mathsf{E}_\mathsf{ahe}.Setup()} = pp$}
        & \textsf{BFV} & Select $pp \leftarrow (t, n, q, \sigma, B, \Delta)$ \\
        \cline{2-3}
        & \textsf{CKKS} & Select $pp \leftarrow (\epsilon, n, q, \sigma, B, \Delta)$ \\
        \hline
        $\mathsf{\mathsf{E}_\mathsf{ahe}.SecKeyGen}(1^\lambda) = \mathsf{sk}$
        & \multicolumn{2}{p{4.35in}|}{Given the security parameter $\lambda$, sample $s \leftarrow R_3$. Output $\mathsf{sk} = s$} \\
        \hline
        $\mathsf{\mathsf{E}_\mathsf{ahe}.PubKeyGen(sk)} = \mathsf{pk}$
        & \multicolumn{2}{p{4.35in}|}{Sample $p_1 \leftarrow R_q$, and $e_{\pk} \leftarrow \chi$. $\mathsf{pk} = (p_0, p_1) = (-s\cdot p_1 + e_{\pk}, p_1)$} \\		
        \hline
        $\mathsf{\mathsf{E}_\mathsf{ahe}.Enc}(\mathsf{pk}, m) = \mathsf{ct}$
        & \multicolumn{2}{p{4.35in}|}{Sample $u \leftarrow R_3$ and $e_0, e_1 \leftarrow \chi$. $\mathsf{ct} = (c_0,c_1) = (\Delta m + u \cdot p_0 + e_0, u \cdot p_1 + e_1)$} \\
        \hline
        $\mathsf{\mathsf{E}_\mathsf{ahe}.Add}(\mathsf{ct}, \mathsf{ct}') = \mathsf{ct}_{\mathsf{add}}$
        & \multicolumn{2}{p{4.35in}|}{Given $\mathsf{ct} = (c_0, c_1)$ and $\mathsf{ct}' = (c_0', c_1')$, $\mathsf{ct}_{\mathsf{add}} = (c_0 + c_0', c_1 + c_1')$} \\	
        \hline
        \multirow{2}{*}{$\mathsf{\mathsf{E}_\mathsf{ahe}.Dec}(\mathsf{sk}, \mathsf{ct}) = m$}
        & \textsf{BFV} & $m = {[\lfloor \frac{t}{q} {[c_0 + c_1 \cdot s]}_q \rceil]}_t$ \\
        \cline{2-3}
        & \textsf{CKKS} & $m \approx \frac{{[c_0 + c_1\cdot s]}_q}{\Delta}$ \\
        \hline
        \multicolumn{3}{|c|}{Decryption remarks}\\
        \hline
        \multicolumn{3}{|p{6.30in}|}{
        $\mathsf{BFV}$: Decryption correctness holds if $(2n +1)B < \frac{q}{2t} - \frac{t}{2}$ (see~Lemma~\ref{lemma:correct-decryption-single-key}), being $n > \delta_R$ an upper bound for the expansion factor $\delta_R$ of the ring $R$~\cite{deCJV21}.
        \newline $\mathsf{CKKS}$: After decryption, we obtain $\tilde{m} = m + e_{\mathsf{ct}} \approx m$, where $e_{\mathsf{ct}}$ is the internal error of the ciphertext, with $||e_{\mathsf{ct}}|| < \frac{(2n + 1)\cdot B}{\Delta}$.
        }\\
        \hline
    \end{tabular}
    \vspace{-0.3cm}
\end{table*}
Both schemes allow us to perform SIMD-style additions by packing a total of $n$ values.\footnote{For the more general case of SIMD-style additions and multiplications, we can pack in each separate ciphertext a total of, respectively, $n$ integer values and $n/2$ approximate complex values for BFV and CKKS.} However, they still present significant differences regarding correctness~\cite{ABMP24}. While for BFV we can choose adequate parameters which satisfy Eq.~\eqref{eq:correctness-add} for all the homomorphic additions needed to compute $f_\Sigma(\cdot)$, capturing the behavior of approximate HE requires relaxing the exact correctness of $\mathsf{E}$. This can be done by considering $\mathsf{E_{ahe}}.\mathsf{Dec}(\mathsf{E_{ahe}}.\mathsf{Add}(c_{\vec{v}}, c_{\vec{w}})) \approx \vec{v} + \vec{w}$ instead. Now, let $c_{\vec{w}_i}=\mathsf{E_{ahe}}.\mathsf{Enc}(\mathsf{pk},\vec{w}_i)$ denote an encryption of $\vec{w}_i$. Given an error tolerance $\epsilon > 0$, we say that the CKKS scheme is approximately correct, if the following holds:
\begin{equation*}
\left\lVert f_\Sigma(\vec{w}_1,\ldots,\vec{w}_L) - \mathsf{E_{ahe}}.\mathsf{Dec}(f_\Sigma(c_{\vec{w}_1},\ldots,c_{\vec{w}_L})) \right \rVert < \epsilon.
\end{equation*}

\textbf{Further details on the lattice-based HE:} To ensure a fair comparison between BFV and CKKS, we made minor modifications to the original CKKS scheme described in~\cite{CKKS17}. These changes primarily involve sampling the randomness in both BFV and CKKS from the same distributions. Additionally, CKKS and BFV use different slot packing methods~\cite{SV14,CKKS17}, leading to some variations in operations before and after encryption and decryption. However, since our focus in this work is on homomorphic additions, for which SIMD-style operations hold without packing, we remove this step in both schemes. Consequently, some of the procedures described in Table~\ref{tab:bfv-ckks} coincide for both schemes.

\subsection{Single-Key HE under a Relaxed Threat Model}
\label{sec:singlekeyhepriagg}

The incorporation of HE for securely computing the aggregation function in FL follows two main principles: (1) unprotected local data is isolated in silos and remains on the clients' premises, and (2) all local model updates leaving a silo are encrypted using HE. Figure~\ref{fig:priflprotocol} sketches the process followed to train a model. The protocol proceeds as follows:
\begin{enumerate}
	\item \textit{Setup step}: All involved clients generate the common protocol cryptographic keys. An initial common ML model is agreed upon by clients and the aggregator.
	\item \textit{Local training step}: Each client trains locally with its data to obtain a local model update.
	\item \textit{Input step}: Each client encrypts its local model update with the public key $\pk$ and sends it to the aggregator over a private channel.
	\item \textit{Evaluation step}: The aggregator homomorphically computes the functionality $f_\Sigma$.
	\item \textit{Output step}: Finally, all involved parties receive the aggregated model update and locally decrypt it using their shared secret key.
	\item \textit{Output (extra) step}: If the aggregator receives the output, the clients send the decryption to the aggregator.
\end{enumerate}

Steps $2$ to $5$ are repeated several times until convergence.

\begin{figure}
    \centering
    \includegraphics[width = \columnwidth]{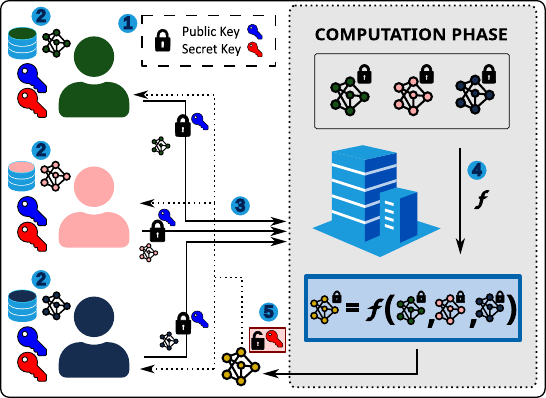}
    \caption{Protocol for private aggregation in FL using single-key HE. Private channels exist between aggregator and clients.}
    \label{fig:priflprotocol}
    \vspace{-0.4cm}
\end{figure}

We remind the reader that in this setting, collusions between the aggregator and any client break security, since the aggregator has access to the ciphertexts of the honest clients while a dishonest client holds the common secret key.\footnote{We refer the reader to Section~\ref{sec:mhe} for a detailed description of multiparty HE schemes that address collusions between aggregator and clients.} Consequently, assuming private channels between each client and the aggregator and no such collusions, both HE schemes described in Section~\ref{sec:heblocks} are a convenient fit for private aggregation. Under these assumptions, and \emph{when the aggregator does not learn the output}, the protocol above preserves input privacy and solves the multiparty aggregation problem of Definition~\ref{def:mpcaggfl}. First, as the HE scheme satisfies $\mathsf{IND}$-$\mathsf{CPA}$ (indistinguishability under chosen-plaintext attack; see Appendix~\ref{app:security-definitions}), the aggregator is unable to infer any information regarding the local updates except their length. Second, since the clients only receive the encrypted aggregation, several colluding semi-honest clients (up to $L - 1$ out of a total of $L$) cannot learn more about the honest clients' inputs than what can be deduced from their shared inputs.

However, RLWE-based schemes also present important limitations if applied straightforwardly. In particular, if the \textit{Output (extra) step} is executed and the aggregator learns the output, the protocol may exhibit unexpected vulnerabilities. If the HE scheme is not adequately instantiated, the decryptions could leak information regarding the private key. We elaborate next on $\mathsf{IND}$-$\mathsf{CPA}^\mathsf{D}$, which is the security definition that best captures this output-access setting.

\subsection{More advanced Threat Models: $\mathsf{IND}$-$\mathsf{CPA}^\mathsf{D}$}
\label{sec:ind-cpa-d}

The notion $\mathsf{IND}$-$\mathsf{CPA}^\mathsf{D}$ extends $\mathsf{IND}$-$\mathsf{CPA}$ by allowing the adversary access to a restricted decryption oracle that can only be used for genuine ciphertexts, or ciphertexts obtained from genuine ciphertexts through valid homomorphic operations.\footnote{Formal security games are detailed in the Appendix~\ref{app:security-definitions}.} Specialized to our aggregation function $f_\Sigma$, the idea is that an adversary knowing ${\vec{w}_1, \ldots, \vec{w}_L }$ and $f_\Sigma$ can also obtain $f_{\Sigma}(\vec{w}_1,\ldots,\vec{w}_L)$. Then, if this adversary has access to ${\mathsf{Enc}(\vec{w}_1), \ldots, \mathsf{Enc}(\vec{w}_L)}$ and homomorphically evaluates the aggregation function $f_\Sigma$ as $f_\Sigma(\mathsf{Enc}(\vec{w}_1), \ldots, \mathsf{Enc}(\vec{w}_L))$, she should not gain more information from the decryption $\mathsf{Dec}\left(f_\Sigma \left( \mathsf{Enc}(\vec{w}_1), \ldots, \mathsf{Enc}(\vec{w}_L) \right)\right)$ than what she can already obtain from $f_{\Sigma}(\vec{w}_1,\ldots,\vec{w}_L)$. Unfortunately, in~\cite{LM21}, the authors showed that some HE schemes, which are $\mathsf{IND}$-$\mathsf{CPA}$ secure, leak information through the difference between $\mathsf{Dec}(\mathsf{Eval}(f, \mathsf{Enc}(m)))$ and $f(m)$ and hence are not $\mathsf{IND}$-$\mathsf{CPA}^\mathsf{D}$ secure. An example is the case of approximate HE schemes like CKKS~\cite{CKKS17}, for which even the difference $\mathsf{Dec}(\mathsf{Enc}(\vec{w}_i)) - \vec{w}_i$ directly leaks the internal noise of the underlying RLWE sample, allowing the adversary to break the RLWE indistinguishability assumption.

Until very recently, the cryptographic community believed that only approximate HE presented this vulnerability and that exact HE schemes like BFV~\cite{Brakerski12,FV12}, BGV~\cite{BGV14} or CGGI~\cite{CGGI16} were invulnerable to such attacks, thus being naturally safe under $\mathsf{IND}$-$\mathsf{CPA}^\mathsf{D}$. However, recent works~\cite{CSBBC24,CCPSS24} have shown that this vulnerability can also be present in practice for exact schemes, whenever they have a non-negligible probability of incorrect decryption. Both works exemplify how this vulnerability can be exploited to implement effective key-recovery attacks against several mainstream HE libraries. These attacks are especially relevant when dealing with threshold HE schemes~\cite{CSBBC24,CCPSS24,BSBK25}. Fortunately, the authors also propose several countermeasures to achieve $\mathsf{IND}$-$\mathsf{CPA}^\mathsf{D}$ security. In particular, the problem can be mitigated by adding an independent smudging noise, whose variance grows exponentially in $\lambda$, during decryption. Notably,~\cite{CSBBC24} demonstrates how this countermeasure significantly impacts the cryptosystem's parameters, subsequently reducing its efficiency. According to the authors, this effect is even more pronounced for CKKS, where using a large-variance smudging noise is likely to severely reduce the precision of the decrypted result.

\textbf{Impact of the smudging noise:} Given a ciphertext $\mathsf{ct}$, the key-recovery attacks mentioned above aim to extract its noise component $e_\mathsf{ct}$. If this extraction is feasible, the secret key can be recovered using linear algebra techniques. For approximate HE schemes, this process is relatively straightforward once decryption is accessible. However, for exact HE schemes, these attacks attempt to induce decryption failures as a way to estimate the corresponding $e_\mathsf{ct}$ component. Therefore, even in extreme cases where the post-processed decryption becomes unusable, adding sufficiently large smudging noise after decryption is necessary to hide any information leakage about the original error $e_\mathsf{ct}$.

Regarding this countermeasure, recent work by~\cite{ABMP24} discusses how the \textit{application-agnostic} nature of $\mathsf{IND}$-$\mathsf{CPA}^\mathsf{D}$ often leads to impractically large parameters when adding large smudging noise. This occurs because $e_\mathsf{smg}$ must statistically hide $e_\mathsf{ct}$ for all possible homomorphic circuits that satisfy correctness for the initially chosen cryptosystem parameters with the $\mathsf{E_{ahe}}.\mathsf{Setup}()$ procedure. The authors propose a relaxation of this definition, termed \textit{application-aware} $\mathsf{IND}$-$\mathsf{CPA}^\mathsf{D}$, which is better suited for the FL setting. In our work, we follow their guidelines and assume that both the clients and the aggregator behave semi-honestly, meaning that they are \textit{expected to follow exactly the prescribed instructions for the aggregator to compute} $c_{\vec{w}^*} = f_\Sigma(c_{\vec{w}_1}, \ldots, c_{\vec{w}_L})$ \textit{and for the clients to encrypt their local updates as} $c_{\vec{w}_1}, \ldots, c_{\vec{w}_L}$.

\section{Private Aggregation with Multiparty HE}
\label{sec:mhe}
As described in Section~\ref{sec:background2}, a key advantage of HE-based aggregation is that it preserves the natural communication flow of FL while enabling privacy against an untrusted aggregator. However, achieving Aggregator Obliviousness (Definition~\ref{def:aggobliviousness}) requires clients to avoid encrypting updates under a common key. \cite{MOJC23} shows that secure linear aggregation calls for multi-user HE variants~\cite{SCRCS11,AHSL22}, which allow ciphertexts encrypted under distinct keys to be homomorphically aggregated. In the context of our \textbf{multiparty aggregation problem in FL} (Definition~\ref{def:mpcaggfl}), these variants enable the construction of private aggregation protocols that remain secure against strong semi-honest adversaries controlling a majority of participants. This guarantee holds even when FL training runs over a public channel where all exchanged messages are observable by the remaining parties.

\textbf{Multiparty HE:} We briefly review in Subsection~\ref{sec:protocol-mhe} the MHE (Multiparty HE) construction introduced in~\cite{MTBH21}. This construction is based on a threshold variant of single-key HE, where the secret key associated with a collective public key is split into multiple additive shares distributed among clients. These shares are generated during an offline setup phase before the FL protocol starts. MHE enables encryption of local updates under this collective public key, while the corresponding secret key remains unknown to any individual party and decryption requires client collaboration. Since this construction can be instantiated from different underlying HE schemes, such as BFV and CKKS, we later discuss the resulting parameter constraints and ciphertext expansion in Section~\ref{sec:perf-eval}, with additional details provided in Appendix~\ref{app:bfv-ckks-details}.

Despite this setup phase, MHE schemes scale well with the number of clients and allow the straightforward reuse of techniques from single-key HE. Unfortunately, recent works~\cite{CSBBC24,CCPSS24,BSBK25} show that both single-key HE and MHE become particularly vulnerable when there is a non-negligible probability of decryption errors. These works illustrate how this vulnerability can be exploited to carry out effective key-recovery attacks against several widely used HE libraries.

\textbf{Security models for HE with multiple keys:} Previous works~\cite{MTBH21} have shown that MHE schemes can be proven secure under the $\mathsf{IND}$-$\mathsf{CPA}$ (Indistinguishability under Chosen-Plaintext Attack) security model, based on the hardness of the RLWE problem. Under this passive security model, the adversary only has access to encryptions via an encryption oracle. While this model may be sufficient from the aggregator's perspective when only RLWE ciphertexts and the public key are visible, it no longer captures the more general case considered in Definition~\ref{def:mpcaggfl} when the FL protocol is executed over a public channel.

As discussed in Section~\ref{sec:background2}, $\mathsf{IND}$-$\mathsf{CPA}^\mathsf{D}$ extends $\mathsf{IND}$-$\mathsf{CPA}$ by introducing a decryption oracle, which better captures scenarios in which the aggregator receives the output. With multiple keys, this distinction becomes more relevant, as participants may gain access not only to full decryptions without knowing the secret key, but also to partial decryptions from individual clients. Notably, the same countermeasure used in the single-key $\mathsf{IND}$-$\mathsf{CPA}^\mathsf{D}$ setting applies to MHE, namely adding an appropriate amount of smudging noise to the partial decryptions~\cite{CSBBC24,CCPSS24}. We discuss the role and concrete impact of this smudging noise in Subsection~\ref{sec:impact-smudging-noise}.

\textbf{General overview:} The FL training process with single-key HE, described in Section~\ref{sec:background2}, must be adapted when using MHE. In particular, the key generation and decryption steps (steps $1$ and $5$) differ, as they are performed collaboratively by the participants.\footnote{The public key is generated during setup and is associated with a global secret key that remains unknown to any single party, while this global secret key is ``reconstructed'' during collaborative decryption (\textit{Output step}).} We now present an adapted FL protocol that securely implements the aggregation $f_\Sigma(\vec{w}_1, \ldots, \vec{w}_L)$ using an $L$-party secret sharing scheme (see Definition~\ref{def:mhe} in Subsection~\ref{sec:protocol-mhe}). This protocol ensures input privacy against a colluding semi-honest aggregator and up to $L - 1$ semi-honest clients. In this setting, local updates are encrypted using MHE before leaving each silo, following the general approach proposed in~\cite{MTBH21}. Steps $2$–$5$ are repeated until convergence. Figure~\ref{fig:priflprotocolmhe} illustrates the full training process:
\begin{enumerate}
\item \textit{Setup step}: All clients generate individual secret keys and a collective public key. An initial common ML model is agreed upon by the clients and the aggregator.
\item \textit{Local training step}: Each client locally trains a model update using its private data.
\item \textit{Input step}: Each client encrypts its local model update using the collective public key before sending it out.
\item \textit{Evaluation step}: The aggregator homomorphically evaluates $f_\Sigma$ over the encrypted updates.
\item \textit{Output step}: All clients \uline{collaboratively decrypt} the aggregated model update. The result can then be explicitly shared with the aggregator, if needed.
\end{enumerate}

\begin{figure}
    \centering
    \includegraphics[width = \columnwidth]{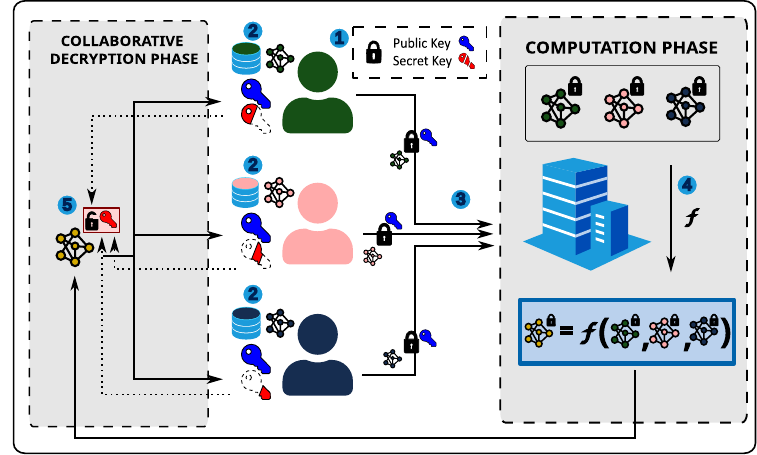}
    \caption{Protocol for private aggregation in FL using MHE.}
    \label{fig:priflprotocolmhe}
    \vspace{-0.4cm}
\end{figure}

\subsection{MHE-based protocol for private aggregation}
\label{sec:protocol-mhe}

This subsection details the baseline MHE-based protocol for securely evaluating $f_\Sigma(\cdot)$ within an FL protocol executed over a public channel. We build on the MHE scheme proposed in~\cite{MTBH21}, which provides a threshold construction on top of the two single-key HE schemes introduced in Section~\ref{sec:background2}. The resulting protocol provides a solution to the \emph{multiparty aggregation problem in FL} of Definition~\ref{def:mpcaggfl}, with the aggregator also acting as a receiver party. The scheme in~\cite{MTBH21} is proven secure in the Common Reference String (CRS) model,\footnote{This model assumes that all parties share a common random string.} and assumes that parties communicate over authenticated channels. Our tailored protocol for private aggregation is summarized in Table~\ref{tab:mhempc} and relies on the following augmented multiparty encryption scheme:
\begin{definition}[Multiparty Additive HE]\label{def:mhe} Let $\mathsf{E}_\mathsf{ahe} = (\mathsf{E.Setup}$, $\mathsf{E.SecKeyGen}$, $\mathsf{E.PubKeyGen}$, $\mathsf{E.Enc}$, $\mathsf{E.Dec}$, $\mathsf{E.Add})$ be the asymmetric and additive homomorphic encryption scheme from Definition~\ref{def:addhe}, whose security is parameterized by $\lambda$ (see Table~\ref{tab:bfv-ckks} for two concrete examples with $\mathsf{BFV}$ and $\mathsf{CKKS}$), and let $\mathsf{S} = (\mathsf{S.Share}, \mathsf{S.Combine})$ be an $L$-party secret sharing scheme. The associated \emph{multiparty homomorphic encryption scheme} ($\mathsf{E}_{\mathsf{mhe}}$) is obtained by applying the secret-sharing scheme $\mathsf{S}$ to $\mathsf{E_{ahe}}$'s secret key $\mathsf{sk}$ (\emph{ideal secret key}) and is defined as the tuple $\mathsf{E}_{\mathsf{mhe}} = (\mathsf{E_{ahe}.Enc}, \mathsf{E_{ahe}.Dec},$ $\mathsf{E_{ahe}.Add}, \mathsf{E_{ahe}^S})$, where $\mathsf{E_{ahe}^S} = (\pi_{\mathsf{SecKeyGen}}, \pi_{\mathsf{PubKeyGen}}, \pi_{\mathsf{Dec}})$ is a set of multiparty protocols executed among the clients in the set $\mathcal{C}$, and having the following \emph{private ideal functionalities} for each client $C_i$:
\begin{itemize}
    \item \emph{Ideal secret-key generation}:
    
    $f_{i,\pi_{\mathsf{SecKeyGen}}}(\lambda) = \mathsf{S}.\mathsf{Share}_i(\mathsf{E_{ahe}}.\mathsf{SecKeyGen}(\lambda)) = \mathsf{sk}_i$.
    \item \emph{Collective public-key generation}:    
    
    $f_{\pi_{\mathsf{PubKeyGen}}}(\mathsf{sk}_1, \mathsf{sk}_2, \ldots, \mathsf{sk}_L)$ \\ $= \mathsf{E_{ahe}}.\mathsf{PubKeyGen}(\mathsf{S}.\mathsf{Combine}(\mathsf{sk}_1, \mathsf{sk}_2, \ldots, \mathsf{sk}_L))$.
    \item \emph{Collective decryption}:
    $f_{\pi_{\mathsf{Dec}}}(\mathsf{sk}_1, \mathsf{sk}_2, \ldots, \mathsf{sk}_L, \mathsf{ct})$ \\ $= \mathsf{E_{ahe}}.\mathsf{Dec}(\mathsf{S}.\mathsf{Combine}(\mathsf{sk}_1, \mathsf{sk}_2, \ldots, \mathsf{sk}_L),\mathsf{ct})$.
\end{itemize}
\end{definition}

Definition~\ref{def:mhe} particularizes~\cite[Def. $2$]{MTH19} to the case where the encryption scheme $\mathsf{E_{ahe}}$ is additively homomorphic.

\begin{table}[!t] 
	\renewcommand{\arraystretch}{1.3}
	\caption{Multiparty Private Aggregation Protocol.}
	\label{tab:mhempc}
	\centering \small
	\begin{threeparttable}[b]
	\begin{tabular}{|m{0.5in}|m{2.6in}|}
		\hline
		\multicolumn{2}{|p{3.2in}|}{
		
    \textbf{Public input}: ideal aggregation functionality $f_\Sigma$
    \newline
    \textbf{Private input}: $\vec{w}_i$ for each $C_i \in \mathcal{C}$\newline
    \textbf{Output} for all clients in $\mathcal{C}$ and for the \textit{Aggregator}: the exact aggregate $\vec{w}^* = f_\Sigma(\vec{w}_1, \ldots, \vec{w}_L)$ for BFV, and an approximation for CKKS
		}\\
		\hline
		Setup & All clients $C_i$ instantiate the MHE scheme $\mathsf{E_{mhe}}$ \newline
		$\mathsf{sk}_i = \mathsf{E_{mhe}}.\pi_{i, \mathsf{SecKeyGen}}(\lambda, \kappa)$\tnote{*} \newline
		$\mathsf{cpk} = \mathsf{E_{mhe}}.\pi_{\mathsf{PubKeyGen}}(\kappa, \mathsf{sk}_1, \ldots, \mathsf{sk}_L)$ \\		
		\hline
		Input & Each $C_i$ encrypts its input $\vec{w}_i$ and provides it to the \textit{Aggregator}\tnote{+} \newline
		$c_{\vec{w}_i} = \mathsf{E_{mhe}.Enc}(\mathsf{cpk}, \vec{w}_i)$\\		
		\hline
		Evaluation & The \textit{Aggregator} computes the encrypted output for the ideal functionality $f_\Sigma$ relying on $\mathsf{E_{mhe}.Add}$\newline
		$c_{\vec{w}^*}= f_\Sigma(c_{\vec{w}_1}, \ldots, c_{\vec{w}_L})$\\
		\hline
		Output & The parties in $\mathcal{C}$ execute the decryption protocol \newline
		$\vec{w}^* = \mathsf{E_{mhe}}.\pi_{\mathsf{Dec}}(\mathsf{sk}_1, \ldots, \mathsf{sk}_L, c_{\vec{w}^*})$\\
		\hline
	\end{tabular}
	   \begin{tablenotes}
     \item[*] $\kappa$ parameterizes the homomorphic capacity of $\mathsf{E_{mhe}}$.
     \item[+] The \textit{Aggregator} is in charge of performing the aggregation.
   \end{tablenotes}
  \end{threeparttable}
    \vspace{-0.6cm}
\end{table}

\subsubsection{Concrete instantiations for Multiparty HE}
The private ideal functionalities introduced in Definition~\ref{def:mhe} can be implemented by clients via the following concrete protocols. In the decryption functionality $f_{\pi_{\mathsf{Dec}}}$, smudging noise is explicitly incorporated, as it constitutes a standard countermeasure required to achieve security in MHE schemes under the previously discussed threat model~\cite{CSBBC24}:
\begin{itemize}
    \item $f_{i,\pi_{\mathsf{SecKeyGen}}}(\lambda)$: Each client $C_i$ independently runs the procedure $\mathsf{\mathsf{E}_\mathsf{ahe}.SecKeyGen}(1^\lambda) = \mathsf{sk}_i$.
    \item $f_{\pi_{\mathsf{PubKeyGen}}}(\mathsf{sk}_1, \mathsf{sk}_2, \ldots, \mathsf{sk}_L)$: Given a common random polynomial $p_1$, each client $C_i \in \mathcal{C}$ samples $e_{\pk,i} \leftarrow \chi$ and discloses to the other clients $p_{0,i} = -p_1 \cdot \mathsf{sk_i} + e_{\pk,i}$. The collective public key is computed as: 
    \begin{equation*}\mathsf{cpk} = (\displaystyle \sum_{C_i \in \mathcal{C}}p_{0,i},p_1) = (-p_1 \cdot \underbrace{\sum_i \mathsf{sk}_i}_{\mathsf{sk}} + \underbrace{\sum_i e_{\pk,i}}_{e_{\pk}},p_1).
    \end{equation*}
    \item $f_{\pi_{\mathsf{Dec}}}(\mathsf{sk}_1, \mathsf{sk}_2, \ldots, \mathsf{sk}_L, \mathsf{ct})$: The collaborative decryption protocol can be divided into two phases.
    \begin{itemize}
        \item Given a ciphertext $\mathsf{ct} = (c_0, c_1)$ encrypted under the ideal secret key $\mathsf{sk}$, each client $C_i$ samples $e_{\mathsf{smg},i} \leftarrow \chi$ and discloses its partial decryption $h_i = \mathsf{sk}_i \cdot c_1 + e_{\mathsf{smg},i}$.
        \item Given all the decryption shares $h_i$ from the clients, the protocol outputs:
        \begin{align*}
            d & = {\left[ c_0 + \sum_{C_i \in \mathcal{C}} h_i \right]}_q \\
            &= \underbrace{\Delta m + u \cdot p_0 + e_0}_{c_0} + \sum_{C_i \in \mathcal{C}} \mathsf{sk}_i \cdot c_1 \\
            &\quad + \underbrace{\sum_{C_i \in \mathcal{C}} e_{\mathsf{smg},i}}_{e_{\mathsf{smg}}} \\
            & = \Delta m + \underbrace{e_{\pk}\cdot u + e_0 + \sk\cdot e_1}_{e_{\mathsf{ct}}} + e_{\mathsf{smg}}.
        \end{align*}
    \end{itemize}
\end{itemize}

The final step for decryption depends on whether BFV or CKKS is used as the underlying single-key HE scheme:
\begin{itemize}
    \item For BFV, the obtained decryption is $m = {[ \lfloor \frac{t}{q}d \rceil ]}_t$.
    \item For CKKS, the obtained decryption is $\tilde{m} = \frac{d}{\Delta} = m + \frac{e_\mathsf{ct} + e_{\mathsf{smg}}}{\Delta} \approx m$ for large enough $\Delta$. Therefore, if $\| e_\mathsf{ct} + e_\mathsf{smg} \|\le B_{\mathsf{ct}}^{\mathsf{MP}}$, where $B_{\mathsf{ct}}^{\mathsf{MP}}$ is defined in Section~\ref{sec:perf-q-bounds}, then $\|\tilde{m} - m\| \le B_{\mathsf{ct}}^{\mathsf{MP}}/\Delta$.
\end{itemize}

We can obtain lower bounds for $q$, similar to the single-key counterparts (see Table~\ref{tab:bfv-ckks}), by taking into account the larger noise terms present in the threshold versions (see Section~\ref{sec:perf-eval} for further details).

\subsection{Impact of the smudging noise in MHE}
\label{sec:impact-smudging-noise}
The smudging technique was introduced in~\cite{AJLTVW12} to address security vulnerabilities inherent to threshold HE schemes. In that work, the authors consider a slightly stronger class of semi-malicious adversaries who follow the protocol but may deviate from the prescribed randomness distributions. These deviations can be adversarially adapted based on the inputs provided by the honest parties. They show that sufficiently large smudging noise is required to securely execute several protocol phases, such as setup and decryption.

Under our semi-honest threat model, smudging noise is required only during collaborative decryption. Given a ciphertext $\mathsf{ct}$, once all partial decryptions $h_i = \mathsf{sk}_i \cdot c_1 + e_{\mathsf{smg},i}$ are collected, the adversary gains direct access to
$d = \Delta \cdot m + e_{\mathsf{ct}} + e_{\mathsf{smg}}$. Thus, $e_{\mathsf{smg}}$ must be chosen so that the sum $e_{\mathsf{ct}} + e_{\mathsf{smg}}$ is statistically indistinguishable from fresh noise $\tilde{e}_{\mathsf{smg}}$ (see the Smudging Lemma in~\cite{AJLTVW12}). We follow the practical guidelines provided in~\cite{CSBBC24}, which recommend Gaussian smudging noise with variance $\sigma_{\mathsf{smg}}^2 = 2^\lambda \sigma_{\mathsf{ct}}^2$. Applying an upper bound $B = \mathcal{O}(\sigma)$,\footnote{$B = 6\sigma$ is commonly used by HE libraries.} this results in
$B_{\mathsf{smg}} = 2^{\frac{\lambda}{2}} B_{\mathsf{ct}}$ (for example, $\lambda = 128$ is suggested in~\cite{MTBH21}).

\textbf{Can we do better? An ideal alternative to large smudging noise:} Assume each client has access to an oracle $\mathcal{O}_i(\cdot)$ which, given the ciphertext to be decrypted $\mathsf{ct}$ as input, outputs $\mathsf{share}_i$ such that $\sum_{i=1}^{L} \mathsf{share}_i = e_{\mathsf{ct}}$; that is, $\mathcal{O}_i(\mathsf{ct}) = \mathsf{share}_i$, where $\mathsf{share}_i$ is ideally uniformly random in $R_q$. By calling this oracle, clients could remove the ciphertext error from the collaborative decryption result and replace it with fresh independent noise. For example, this can be done by computing $h_i\!=\!\mathsf{sk}_i\mkern-1mu \cdot\mkern-1mu c_1 - \mathsf{share}_i + e_{\mathsf{smg},i}$. Once~all partial decryptions are gathered, the adversary obtains $d\!=\!\Delta\mkern-1mu\cdot\mkern-1mu m + e_{\mathsf{smg}}$. Although this expression is similar to the previous one, note that now $\mkern-1mu e_{\mathsf{smg}}\mkern-1mu$ does not need to grow exponentially with $\lambda$ because it no longer has to hide $e_{\mathsf{ct}}$ to make $e_{\mathsf{ct}}\mkern-1mu +\mkern-1mu e_{\mathsf{smg}}$ statistically indistinguishable from fresh noise $\tilde{e}_{\mathsf{smg}}$.

Despite the apparent strength of this idea, implementing such an oracle in practice is not straightforward in a public-key encryption scheme. First, the simplest approach would require each client to track the error component introduced in its input ciphertext. However, since clients do not know the underlying error of the public key, they cannot determine the share that aggregates into the final error of the output ciphertext. Second, even if clients could know this error, it would be correlated with the collective-public-key error, potentially making the scheme vulnerable to key-recovery attacks similar to those described in~\cite{CSBBC24,CCPSS24} and Section~\ref{sec:background2}.

We explore this idea in more detail in Section~\ref{sec:prop-opt} and show that it is feasible when using symmetric-key encryption schemes. The main drawback is that clients must store and track the errors of their input ciphertexts. While this may not be practical for homomorphic evaluation of general functions, it is reasonable in our FL scenario with homomorphic average aggregation, where clients only need to store the errors used during a single communication round. 

\section{Proposed protocol and optimizations}
\label{sec:prop-opt}
The MHE-based approach analyzed in the previous section provides a practical baseline for private aggregation in FL, but it also presents several limitations. In particular, ciphertexts must remain compatible with $\mathsf{cpk}$, collaborative decryption requires large smudging noise, and the noise embedded in $\mathsf{cpk}$ grows with the number of clients $L$. To overcome these issues, we introduce a construction based on \emph{multi-secret-key} HE. Instead of encrypting client updates under $\mathsf{cpk}$, each client encrypts using its own secret key, while ciphertexts remain compatible with homomorphic aggregation and collaborative decryption. This removes the dependency on $\mathsf{cpk}$ and enables several practical optimizations, including reduced ciphertext size.

This section is organized as follows: Subsection~\ref{sec:ask-he} introduces the secret-key HE building blocks used in our construction. Subsection~\ref{sec:mask-he} presents the proposed multi-secret-key aggregation protocol, including concrete instantiations based on BFV and CKKS together with implementation-oriented optimizations. Finally, Subsection~\ref{sec:proofoverview} states the main security result and provides a proof overview.

\textbf{General overview of our protocol:} The proposed protocol follows the same high-level FL workflow as the MHE-based baseline, but replaces the $\mathsf{cpk}$ encryption with a correlated symmetric-key construction across clients. The overall workflow is summarized below and illustrated in Figure~\ref{fig:multi-sk-he}.

\begin{enumerate}
\item \textit{Setup step}: All clients generate \uline{individual secret keys} and run a \uline{lightweight protocol establishing auxiliary shares} for the protocol execution. An initial ML model is agreed upon by the clients and the aggregator.
\item \textit{Local training step}: Each client locally trains a model update using its private data.
\item \textit{Input step}: Each client \uline{encrypts its local model update using its individual secret key} and sends the resulting ciphertext to the aggregator.
\item \textit{Evaluation step}: The aggregator homomorphically evaluates $f_\Sigma$ over the encrypted updates.
\item \textit{Output step}: All clients collaboratively decrypt the aggregated model update, which can then be shared with the aggregator if required.
\end{enumerate}

\begin{figure}
    \centering
    \includegraphics[width=\columnwidth]{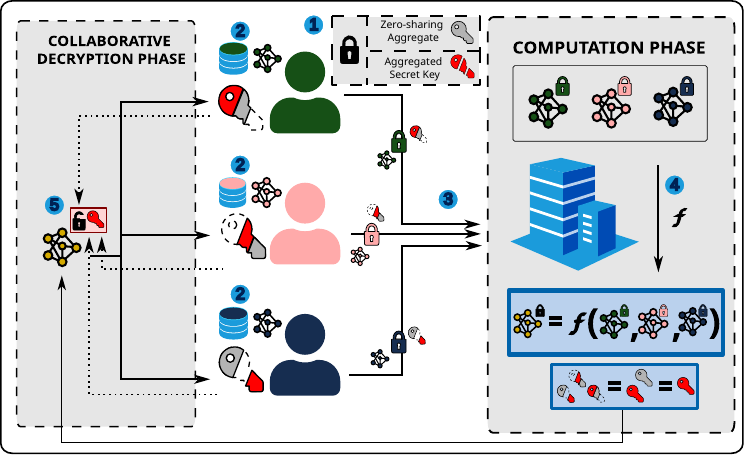}
    \caption{Protocol for private aggregation in FL using the proposed multi-secret-key HE-based construction.}
    \label{fig:multi-sk-he}
    \vspace{-0.4cm}
\end{figure}

\subsection{Secret-Key HE building blocks}
\label{sec:ask-he}

We revisit Definition~\ref{def:addhe} for additive HE schemes by adapting it to the secret-key setting and removing the public-key generation procedure $\mathsf{PubKeyGen}$. The scheme retains the same packing capacity ($n$ values) and additive homomorphism whenever ciphertexts share the same secret key.

\begin{definition}[Additive Symmetric-Key Homomorphic Encryption]\label{def:ashe}
    Let $\mathsf{E} = \{\mathsf{Setup}, \mathsf{SecKeyGen}, \mathsf{Enc}, \mathsf{Dec}\}$ be a symmetric encryption scheme, whose security is parameterized by $\lambda$. An additive symmetric-key homomorphic encryption scheme extends $\mathsf{E}$ with the $\mathsf{Add}$ procedure, having $\mathsf{E_{ashe}} = \{\mathsf{E}.\mathsf{Setup},$ $ \mathsf{E}.\mathsf{SecKeyGen}, \mathsf{E}.\mathsf{Enc}, \mathsf{E}.\mathsf{Dec}, \mathsf{Add}\}$. The correctness condition for $\mathsf{Add}$ is stated in terms of $\epsilon$ as follows:
    \begin{itemize}
        \item Let $\epsilon \geq 0$ be the target margin error: $0$ for exact HE schemes and positive for approximate HE schemes.
        \item Let $c_{\vec{v}}$ and $c_{\vec{w}}$ be two ciphertexts \uline{encrypting under the same secret key $\mathsf{sk}$}, respectively, $\vec{v}$ and $\vec{w}$. The procedure $\mathsf{Add}$ satisfies:
        \begin{equation} \label{eq:correctness-add-general}
            \left \lVert\vec{v} + \vec{w} - \mathsf{E_{ashe}}.\mathsf{Dec}(\mathsf{E_{ashe}}.\mathsf{Add}(c_{\vec{v}}, c_{\vec{w}}))\right \rVert \leq \epsilon.
        \end{equation}
    \end{itemize}
\end{definition}
As in Section~\ref{sec:background2} for the public-key HE setting, we instantiate $\mathsf{E_{ashe}}$ using both approximate and exact variants, namely secret-key versions of CKKS and BFV, respectively. Their security relies on the RLWE problem, and Table~\ref{tab:bfv-ckks-ashe} details their primitives and cryptographic parameters. Among the properties of these secret-key constructions, we next highlight two aspects that are particularly relevant for our private average aggregation protocol (see Section~\ref{sec:mask-he}).

\begin{table*}[h!] 
        \centering
	\renewcommand{\arraystretch}{2.2}
	\caption{Summary of $\mathsf{BFV}$ \& $\mathsf{CKKS}$ Additive Secret-key HE schemes.}
	\label{tab:bfv-ckks-ashe}
    \scriptsize 
	\begin{tabular}{|m{1.6in}|m{0.4in}|m{2.7in}|}
		\hline
		\multicolumn{3}{|c|}{Parameters remarks}\\
		\hline
		\multicolumn{3}{|p{4.7in}|}{
        $\mathsf{BFV}$: $R_q$ is the ciphertext ring. $\Delta = \lfloor \frac{q}{t} \rfloor$. $R_t$ is the plaintext ring, with $q > t$. 
                  \newline $\mathsf{CKKS}$: $R_q$ is the ciphertext ring. Given an arbitrary target error margin $\epsilon$ and $B$, we can set the scale $\Delta$.
		}\\
		\hline
		
		\multicolumn{3}{|c|}{Cryptographic Primitives}\\
        \hline
		\multirow{2}{*}{$\mathsf{\mathsf{E}_\mathsf{ashe}.Setup()} = pp$}
		& \textsf{BFV} & Select $pp \leftarrow (t, n, q, \sigma, B, \Delta)$ \\
		\cline{2-3}
		& \textsf{CKKS} & Select $pp \leftarrow (\epsilon, n, q, \sigma, B, \Delta)$ \\
		\hline
		$\mathsf{\mathsf{E}_\mathsf{ashe}.SecKeyGen}(1^\lambda) = \mathsf{sk}$
		& \multicolumn{2}{p{3.1in}|}{Given the security parameter $\lambda$, sample $s \leftarrow R_3$. Output $\mathsf{sk} = s$} \\
		\hline
		$\mathsf{\mathsf{E}_\mathsf{ashe}.Enc}(\mathsf{sk}, m) = \mathsf{ct}$
		& \multicolumn{2}{p{3.1in}|}{Sample $a \leftarrow R_q$ and $e \leftarrow \chi$. $\mathsf{ct} = (c_0,c_1) = (-a\cdot \mathsf{sk} + e + \Delta m, a)$} \\
		\hline
  	$\mathsf{\mathsf{E}_\mathsf{ashe}.Add}(\mathsf{ct},\mathsf{ct}') = \mathsf{ct}_{\mathsf{add}}$
		&  \multicolumn{2}{p{3.1in}|}{Given $\mathsf{ct} = (c_0, c_1)$ and $\mathsf{ct}' = (c_0', c_1')$, $\mathsf{ct}_{\mathsf{add}} = (c_0 + c_0', c_1 + c_1')$} \\	
            \hline
		\multirow{2}{*}{$\mathsf{\mathsf{E}_\mathsf{ashe}.Dec}(\mathsf{sk}, \mathsf{ct}) = m$}
		& \textsf{BFV} & $m = {[\lfloor \frac{t}{q} {[c_0 + c_1 \cdot \mathsf{sk} ]}_q \rceil]}_t$ \\
		\cline{2-3}
		& \textsf{CKKS} & $m \approx \frac{{[c_0 + c_1\cdot \mathsf{sk} ]}_q}{\Delta}$ \\
		\hline
  \multicolumn{3}{|c|}{Decryption remarks}\\
		\hline
		\multicolumn{3}{|p{4.8in}|}{
        $\mathsf{BFV}$: Decryption correctness now holds if $B < \frac{q}{2t} - \frac{t}{2}$.
        \newline $\mathsf{CKKS}$: After decryption, we obtain $\tilde{m} = m + e_{\mathsf{ct}} \approx m$, where $e_{\mathsf{ct}}$ is the internal error of the ciphertext, with $||e_{\mathsf{ct}}|| < \frac{B}{\Delta}$.
		}\\
		\hline
	\end{tabular}
    \vspace{-0.3cm}
\end{table*}

\textbf{Compressed ciphertext size:} While in public-key HE schemes (see Table~\ref{tab:bfv-ckks}) the two components of a ciphertext $\mathsf{ct}=(c_0,c_1)$ depend on $\mathsf{pk}$, in the symmetric-key setting the public component (denoted $a$ in Table~\ref{tab:bfv-ckks-ashe}, from which $c_0$ is derived) is a uniformly random polynomial in $R_q$ and therefore independent of $\mathsf{sk}$. As discussed in~\cite{PBCSSZ23}, this allows clients in an FL protocol to compress their encrypted inputs by approximately one half. Under the CRS model, this public component can be generated deterministically using a pseudorandom function $\mathsf{PRF}$ (see Definition~\ref{def:prf} below). 

\begin{definition}[Pseudorandom Function (PRF)]\label{def:prf} Let $\mathcal{K}$, $\mathcal{X}$, and $\mathcal{Y}$ be finite sets. A function $F:\mathcal{K}\times\mathcal{X}\rightarrow\mathcal{Y}$ is efficiently computable if there exists a probabilistic polynomial-time (PPT) algorithm that, on input $(\kappa,T)\in\mathcal{K}\times\mathcal{X}$, outputs $F_\kappa(T)$. We say that $F$ is a PRF (see~\cite{KL14}) if, for every PPT distinguisher $\mathcal{A}$,
\begin{equation*}
\left| \Pr\left[\mathcal{A}^{F_\kappa(\cdot)}(1^\lambda)=1\right] - \Pr\left[\mathcal{A}^{R(\cdot)}(1^\lambda)=1\right] \right| \leq \mathsf{negl}(\lambda),
\end{equation*}
with $\kappa \leftarrow \mathcal{K}$ and $R:\mathcal{X}\rightarrow\mathcal{Y}$ a uniformly random function.
\end{definition}

We derive the public RLWE component as $a \leftarrow \mathsf{PRF}_\kappa(T)$. Any party that knows $(\kappa,T)$ can deterministically reconstruct $a$, enabling ciphertext compression. Since $a$ is public in RLWE, its predictability given $(\kappa,T)$ does not affect security provided that $\kappa$ is sampled uniformly at random. For parties that do not know $\kappa$ (e.g., the aggregator), the distribution of $\mathsf{PRF}_\kappa(T)$ remains computationally indistinguishable from uniform over $R_q$.

\textbf{Ciphertext noise is known by the encryptor:} Even under ciphertext compression, and as a consequence of the RLWE assumption, the secret-key versions of both approximate and exact HE schemes satisfy $\mathsf{IND}$-$\mathsf{CPA}$ security. To achieve $\mathsf{IND}$-$\mathsf{CPA}^{\mathsf{D}}$ security, one may again rely on adding large smudging noise after decryption. However, in these schemes, encryptors not only hold $\mathsf{sk}$ but also know the internal ciphertext noise $e_{\mathsf{ct}}$, thereby effectively realizing the oracle $\mathcal{O}_1(\mathsf{ct})$ introduced in Section~\ref{sec:impact-smudging-noise} for the special case $L = 1$ (i.e., $\mathcal{O}_1(\mathsf{ct}) = e_\mathsf{ct} = \mathsf{share}_0$).
This contrasts with the public-key setting, where a party knowing $\mathsf{pk}$ can encrypt but does not know $e_{\mathsf{ct}}$. Focusing on the homomorphic evaluation of $f_\Sigma\big(\mathsf{Enc}(\mathsf{sk}, \vec{w}^1), \ldots, \mathsf{Enc}(\mathsf{sk}, \vec{w}^L)\big)$, and assuming a party holding $\mathsf{sk}$ acts as both encryptor and decryptor, this property allows the encryptor to record the noise used during encryption and track the linear combination applied to the ciphertexts. Consequently, the accumulated noise in the encrypted result can be explicitly computed and canceled during decryption, without relying on high-variance smudging noise. While this assumption is generally unrealistic in generic HE settings, it aligns well with FL deployments, where the number and lifetime of ciphertexts are naturally bounded between aggregation rounds. As discussed in Section~\ref{sec:mask-he}, this observation is instrumental for our protocol.
        
\subsection{Multi-Secret-Key HE Private Aggregation}
\label{sec:mask-he}

The $\mathsf{E_{ashe}}.\mathsf{Add}$ primitive in Definition~\ref{def:ashe} assumes that input ciphertexts are encrypted under the same unknown secret key. In our protocol, however, ciphertexts under different secret keys must be combined while preserving privacy and correctness during decryption. We address this challenge using two additional building blocks: (i) a practical implementation of the oracle $\mathcal{O}_i(\mathsf{ct})$ discussed in Section~\ref{sec:mhe}, which is feasible with secret-key HE schemes; and (ii) an extension of the $\mathsf{E_{ashe}}.\mathsf{Add}$ procedure (see Table~\ref{tab:bfv-ckks-ashe}) that, together with an $L$-party linear secret sharing scheme (Definitions~\ref{def:extadd} and~\ref{def:zero-sharing-lsss}), supports the homomorphic evaluation of linear aggregations $f_\Sigma(\mathsf{Enc}(\mathsf{sk}_1, w_1), \ldots, \mathsf{Enc}(\mathsf{sk}_L, w_L))$ over ciphertexts encrypted under distinct secret keys~\cite{PBCSSZ23}.

We first discuss the required property of secret-key HE schemes to extend $\mathsf{E_{ashe}}.\mathsf{Add}$ under multiple secret keys. Next, we detail the general construction of our proposed multi-secret-key HE-based solution to the multiparty aggregation problem in FL (see Definition~\ref{def:mpcaggfl}).

\textbf{Key homomorphisms with secret-key encryptions:} We make explicit the uniformly random polynomial term $a$ in the encryption procedure of Table~\ref{tab:bfv-ckks-ashe}. A fresh encryption of $m$ under secret key $\mathsf{sk}$ is written as $\mathsf{ct} = \mathsf{Enc}(a,\mathsf{sk}, m) = (-a \cdot \mathsf{sk} + e + \Delta m, a)$.  This formulation reconnects with the collective public-key generation functionality $f_{\pi_{\mathsf{PubKeyGen}}}(\mathsf{sk}_1,\ldots,\mathsf{sk}_L)$ in the $\mathsf{E_{mhe}}$ scheme. Under the CRS model, clients rely on a preshared polynomial $p_1 \in R_q$ to generate RLWE samples $\mathsf{ct}_i = (p_{0,i}, p_1)$ with the same second component. In particular, by renaming $p_1$ as $a$, these RLWE samples can be viewed as fresh secret-key encryptions of zero $\mathsf{Enc}(a, \mathsf{sk}_i, 0)$ for each client $C_i \in \mathcal{C}$. To formalize aggregation across ciphertexts encrypted under distinct secret keys but sharing the same $a$, we introduce the following extended addition procedure:
\begin{definition}[Extended Addition Procedure]\label{def:extadd} Let $\mathsf{ct} = (b_1,a)$ and $\mathsf{ct}' = (b_2,a)$ be two ciphertexts produced by $\mathsf{Enc}(a,\mathsf{sk}_i,m_i)$. The procedure $\mathsf{ExtAdd}$ is defined as follows:
\begin{itemize}
    \item If both ciphertexts are encrypted under the same secret key, then $\mathsf{ExtAdd}(\mathsf{ct},\mathsf{ct}') = \mathsf{Add}(\mathsf{ct},\mathsf{ct}')$.
    \item If the ciphertexts share the same polynomial component $a$ but are encrypted under distinct secret keys, then $\mathsf{ExtAdd}(\mathsf{ct},\mathsf{ct}') = \mathsf{ct}_{\mathsf{add}} = (b_1 + b_2,a)$, where $b_i = -a \cdot \mathsf{sk}_i + e_i + \Delta m_i$ (see Table~\ref{tab:bfv-ckks-ashe}).
\end{itemize}
\end{definition}
By linearity, the result is $(-a(\mathsf{sk}_1+\mathsf{sk}_2) + (e_1+e_2) + \Delta(m_1+m_2), a)$,
i.e., $\mathsf{ct}_{\mathsf{add}} = \mathsf{Enc}(a, \mathsf{sk}_1+\mathsf{sk}_2, m_1+m_2)$. More generally, aggregating $L$ ciphertexts with the same $a$ yields
$\mathsf{Enc}\big(a, f_\Sigma(\mathsf{sk}_1,\ldots,\mathsf{sk}_L), f_\Sigma(m_1,\ldots,$ $m_L)\big)$ for any linear aggregation function $f_\Sigma$ defined by public coefficients $\lambda_i$. Thus, secret-key encryptions with a shared random component satisfy a \emph{key homomorphism}: ciphertext aggregation induces aggregation of both messages and secret keys.

One could then attempt to reuse the collective decryption procedure of $\mathsf{E_{mhe}}$. However, this would introduce a security flaw: the partial decryptions would reuse the input ciphertexts' random component $a$, potentially enabling attacks that leak information about the underlying secret keys. The complete protocol is presented below, where this issue is addressed by preventing the unsafe reuse of shared polynomial components, in line with the approach introduced in~\cite{PBCSSZ23}.

\textbf{Multiparty secret-key private aggregation:} While the baseline MHE-based construction $\mathsf{E_{mhe}}$ from Section~\ref{sec:protocol-mhe} applies a generic $L$-party secret sharing scheme to an unknown ideal secret key, our approach uses a linear secret-sharing layer whose reconstruction coefficients $\lambda_i$ are aligned with the aggregation function $f_\Sigma$. In particular, we consider the following $L$-party secret sharing scheme:
\begin{definition}[Zero-Sharing Linear Secret Sharing Scheme]\label{def:zero-sharing-lsss} Let $\mathsf{S} = (\mathsf{S}.\mathsf{Share}, \mathsf{S}.\mathsf{Combine})$ be a linear secret sharing scheme (LSSS) over the ring $R_q$. The algorithm $\mathsf{S}.\mathsf{Share}$ outputs $L$ shares $r_1,\dots,r_L \in R_q$, where $r_1,\dots,r_{L-1}$ are sampled uniformly at random and $r_L$ is computed such that $\sum_{i=1}^{L} \lambda_i r_i = 0$, with $\lambda_1,\dots,\lambda_L \in \mathbb{Z}_q$ being public reconstruction coefficients. The reconstruction algorithm satisfies $\mathsf{S}.\mathsf{Combine}(r_1,\dots,r_L) = \sum_{i=1}^{L} \lambda_i r_i = 0$. In particular, the reconstruction coefficients correspond to the linear aggregation function $f_\Sigma(r_1, \ldots, r_L)$ used in the FL protocol. We write $\mathsf{S.Share}_i$ to denote the $i$-th share output by $\mathsf{S.Share}$.
\end{definition}

We assume a secure realization of $\mathsf{S}$ against static semi-honest adversaries; the concrete realization in Appendix~\ref{app:aux-zero-sharing} uses pairwise private authenticated channels between clients. Our protocol is summarized in Table~\ref{tab:multi-sk-he} and relies on the following multi-secret-key HE construction:

\begin{definition}[Multiparty Additive Symmetric-Key HE]\label{def:multi-sk-he} Let $\mathsf{E}_{\mathsf{ashe}} = (\mathsf{E.Setup}, \mathsf{E.SecKeyGen}, \mathsf{E.Enc}, \mathsf{E.Dec}, \mathsf{E.Add})$ be the symmetric additive homomorphic encryption scheme from Definition~\ref{def:ashe}, whose security is parameterized by $\lambda$ (see Table~\ref{tab:bfv-ckks-ashe} for concrete $\mathsf{BFV}$ and $\mathsf{CKKS}$ instantiations), and let $\mathsf{S} = (\mathsf{S.Share}, \mathsf{S.Combine})$ be the Zero-Sharing LSSS from Definition~\ref{def:zero-sharing-lsss}. The associated \emph{multiparty symmetric homomorphic encryption scheme} ($\mathsf{E}_{\mathsf{mshe}}$) is defined as follows:\\

(i) During setup, each client $C_i$ samples $(r_i,\mathsf{sk}_i)$ where $r_i \leftarrow \mathsf{S.Share}_i$ and $\mathsf{sk}_i \leftarrow \mathsf{E}_{\mathsf{ashe}}.\mathsf{SecKeyGen}(\lambda)$. The effective encryption key of $C_i$ is $\mathsf{sk}_i + r_i$. By Def.~\ref{def:zero-sharing-lsss}, the shares satisfy $\mathsf{S}.\mathsf{Combine}(r_1,\ldots,r_L)=f_\Sigma(r_1,\ldots,r_L)=0$. Hence, by linearity, $\mathsf{S.Combine}(\mathsf{sk}_1+r_1,\ldots,\mathsf{sk}_L+r_L) = f_\Sigma(\mathsf{sk}_1,\ldots,\mathsf{sk}_L)$, which corresponds to the aggregated (ideal) secret key.\\

(ii) The homomorphic addition procedure $\mathsf{Add}$ is replaced by $\mathsf{ExtAdd}$ from Definition~\ref{def:extadd}.\\

(iii) The scheme is finally defined as $\mathsf{E}_{\mathsf{mshe}} =
(\mathsf{E}_{\mathsf{ashe}}.\mathsf{Enc},$ $\mathsf{E}_{\mathsf{ashe}}.\mathsf{Dec}, \mathsf{ExtAdd}, \mathsf{E}_{\mathsf{ashe}}^{\mathsf{S}})$, where $\mathsf{E}_{\mathsf{ashe}}^{\mathsf{S}} = (\pi_{\mathsf{SecKeyGen}}, \pi_{\mathsf{Enc}},$ $\pi_{\mathsf{Dec}})$ denotes a set of multiparty protocols run by the clients $C_i\in\mathcal{C}$ with the following \emph{private ideal functionalities}: 
\begin{itemize}
    \item \emph{Ideal secret-key generation}: $f_{i,\pi_{\mathsf{SecKeyGen}}}(\lambda) $ \\$= \left(\mathsf{S}.\mathsf{Share}_i, \mathsf{E}_{\mathsf{ashe}}.\mathsf{SecKeyGen}(\lambda)\right) = (r_i, \mathsf{sk}_i)$.
    \item \emph{Correlated encryption}:  
    $f_{i,\pi_{\mathsf{Enc}}}((r_i,\mathsf{sk}_i),m_i) $ \\$= \mathsf{E}_{\mathsf{ashe}}.\mathsf{Enc}(a,\mathsf{sk}_i + r_i,m_i)$, where $a \in R_q$ is sampled internally once per aggregation round (or once per ciphertext block when multiple blocks are needed to encode each client update) and used consistently across all clients $C_i \in \mathcal{C}$.
    \item \emph{Collective decryption}: $f_{\pi_{\mathsf{Dec}}}((\mathsf{sk}_1, r_1), (\mathsf{sk}_2, r_2), \ldots,(\mathsf{sk}_L,$ $r_L), \mathsf{ct}) = \mathsf{E}_{\mathsf{ashe}}.\mathsf{Dec}(\mathsf{S.Combine}(\mathsf{sk}_1 + r_1, \mathsf{sk}_2 + r_2, \ldots,$ $\mathsf{sk}_L + r_L),\mathsf{ct})$.
\end{itemize}
\end{definition}

\begin{table}[!t] 
	\renewcommand{\arraystretch}{1.3}
	\caption{Multiparty Secret-Key Private Aggregation Protocol.}
	\label{tab:multi-sk-he}
	\centering \small
	\begin{threeparttable}[b]
	\begin{tabular}{|m{0.5in}|m{2.6in}|}
		\hline
		\multicolumn{2}{|p{3.2in}|}{
		
    \textbf{Public input}: ideal aggregation functionality $f_\Sigma$
    \newline
    \textbf{Private input}: $\vec{w}^i$ for each $C_i \in \mathcal{C}$\newline
    \textbf{Output} for all clients in $\mathcal{C}$, and for the \textit{Aggregator} if required: the exact aggregate $\vec{w}^* = f_\Sigma(\vec{w}_1, \ldots, \vec{w}_L)$ for BFV, and an approximation for CKKS
		}\\
		\hline
		Setup & All clients $C_i$ instantiate the Multiparty secret-key scheme $\mathsf{E_{mshe}}$ \newline
		$(r_i, \mathsf{sk}_i) = \mathsf{E_{mshe}}.\pi_{i, \mathsf{SecKeyGen}}(\lambda, \kappa)$\tnote{*} \\		
		\hline
		Input & Each $C_i$ encrypts\tnote{$\dagger$}\, its input $\vec{w}_i$ and provides it to the \textit{Aggregator}\tnote{+} \newline
		$a \leftarrow R_q$ shared across all $C_i$ \newline
		$c_{\vec{w}^i} = \mathsf{E_{mshe}.Enc}(a, \mathsf{sk}_i + r_i, \vec{w}_i)$\\		
		\hline
		Evaluation & The \textit{Aggregator} computes\tnote{$\dagger$}\, the encrypted output for the ideal functionality $f_\Sigma$ relying on $\mathsf{E_{mshe}.ExtAdd}$\newline
		$c_{\vec{w}^*}= f_\Sigma(c_{\vec{w}_1}, c_{\vec{w}_2}, \ldots, c_{\vec{w}_L})$\\
		\hline
		Output & The parties in $\mathcal{C}$ execute\tnote{$\dagger$}\, the decryption protocol \newline
		$\vec{w}^* = \mathsf{E_{mshe}}.\pi_{\mathsf{Dec}}((\mathsf{sk}_1, r_1), \ldots, (\mathsf{sk}_L, r_L), c_{\vec{w}^*})$\\
		\hline
	\end{tabular}
	   \begin{tablenotes}
     \item[*] $\kappa$ parameterizes the additive homomorphic capacity of $\mathsf{E_{ashe}}$.
     \item[$\dagger$] Remember that, for readability, $c_{\vec{w}^i}$ and $c_{\vec{w}^*}$ denote the single-ciphertext case. In the general case, each input $\vec{w}^{\,i}$ is encoded into $N_{\mathsf{ctx}}$ ciphertexts, each with a different $a_j$, and all operations are performed blockwise.
     \item[+] The \textit{Aggregator} is in charge of performing the aggregation.
   \end{tablenotes}
  \end{threeparttable}
    \vspace{-0.6cm}
\end{table}

\textbf{Remark on ciphertext encoding:} For simplicity of exposition, we describe the protocol in Table~\ref{tab:multi-sk-he} as if each client update $\vec{w}_i$ were encrypted into a single ciphertext. When needed, $\vec{w}_i$ is partitioned into $N_{\mathsf{ctx}}$ blocks and encrypted blockwise by each client as $\{\mathsf{ct}_{i,j}=(b_{i,j},a_j)\}_{j=1}^{N_{\mathsf{ctx}}}$. Aggregation with $\mathsf{ExtAdd}$ and collective decryption are applied blockwise. Throughout this section, this blockwise interpretation is assumed where appropriate. Observe that when the same secret key is reused across rounds, the public components $a_j$ \textbf{must be pairwise distinct across all ciphertexts encrypted under that key, including future aggregation rounds under the same setup.} See Section~\ref{sec:mask-he-concrete-instantiations} for CKKS- and BFV-based instantiations, and Section~\ref{sec:proofoverview} and Appendix~\ref{app:proof-theorem1} for details on the protocol's security guarantees.

\subsubsection{Concrete Multiparty HE instantiation:}
\label{sec:mask-he-concrete-instantiations}

In Section~\ref{sec:impact-smudging-noise}, we introduced an oracle $\mathcal{O}_i(\cdot)$ allowing client $C_i$ to obtain a share of the aggregated ciphertext noise. We now formalize this abstraction as a client-local programmable functionality. In our construction, the oracle can be indexed by the relevant ciphertext component (e.g., the $a$ component of $\mathsf{ct}$ in Table~\ref{tab:bfv-ckks-ashe}), instead of the full ciphertext.
\begin{definition}[Client-Local Programmable Oracle $\mathcal{O}_i(\cdot)$]\label{def:oracle}
For each client $i$, we define an abstract stateful functionality $\mathcal{O}_i : R_q \rightarrow R_q \cup \{\bot\}$ with the following interfaces:
\begin{itemize}
    \item $\mathsf{Setup}$: The functionality is initialized as an empty table. Formally, for all $a \in R_q$, we set $\mathcal{O}_i(a) = \bot$. It is local to client $i$ and not accessible to other parties.
    \item $\mathsf{Write}(a,e)$: Upon input $(a,e)$ with $a,e \in R_q$, the client updates the internal state by setting $\mathcal{O}_i(a) := e$. Optionally, this assignment may be restricted to previously undefined entries, i.e., only if $\mathcal{O}_i(a) = \bot$.
    \item $\mathsf{Read}(a)$: Upon input $a \in R_q$, the functionality returns $\mathcal{O}_i(a)$, i.e., it outputs $e$ if $\mathcal{O}_i(a)=e$, and $\bot$ otherwise.
\end{itemize}
Thus, $\mathcal{O}_i$ behaves as a deterministic client-local lookup table whose entries are defined during the protocol execution.
\end{definition}

The oracle $\mathcal{O}_i$ admits a natural realization within our scheme $\mathsf{E_{mshe}}$. By construction, each client $C_i$ already knows its local encryption noise $e_i$, and hence its weighted contribution $\lambda_i \cdot e_i$ to the aggregated ciphertext noise $e$, where $e = f_\Sigma(e_1,\ldots,e_L)$. Consequently, $\mathcal{O}_i$ can be implemented by locally programming the value associated with the relevant ciphertext component (e.g., $a$). We now describe how this mechanism is used in the protocols implementing the ideal functionalities from Definition~\ref{def:multi-sk-he}:
\begin{itemize}
\item $f_{i,\pi_{\mathsf{SecKeyGen}}}(\lambda)$: Each client $C_i$ independently runs $\mathsf{E_{ashe}}.\mathsf{SecKeyGen}(1^\lambda)=\mathsf{sk}_i$ and samples $r_i \leftarrow \mathsf{S.Share}_i$ through an interaction with the other clients implementing $\mathsf{S}$ (see Definition~\ref{def:zero-sharing-lsss}). Finally, $C_i \in \mathcal{C}$ invokes the $\mathsf{Setup}$ interface of its local oracle $\mathcal{O}_i(\cdot)$. See Appendix~\ref{app:aux-zero-sharing} for an example implementation of $\mathsf{S}$.
\item $f_{i,\pi_{\mathsf{Enc}}}((r_i,\mathsf{sk}_i),m_i)$: All clients jointly sample a common random polynomial $a \leftarrow R_q$. Each client $C_i \in \mathcal{C}$ samples $e_i \leftarrow \chi$ and sends $b_i = -a\cdot(\mathsf{sk}_i + r_i) + e_i + \Delta m_i$ to the aggregator. Finally, each client invokes $\mathsf{Write}(a,e_i)$ to update its local oracle $\mathcal{O}_i(\cdot)$.
\item $f_{\pi_{\mathsf{Dec}}}((\mathsf{sk}_1, r_1), (\mathsf{sk}_2, r_2), \ldots, (\mathsf{sk}_L, r_L), \mathsf{ct})$: The collaborative decryption protocol proceeds in two phases.
	\begin{itemize}
	        \item Given a ciphertext $\mathsf{ct} = (b,a)$ encrypted under the ideal secret key $\mathsf{sk} = f_\Sigma(\mathsf{sk}_1 + r_1, \ldots, \mathsf{sk}_L + r_L)$,
each client retrieves $e_i = \mathsf{Read}(a)$, samples $\tilde{e}_i \leftarrow \chi$, and computes its partial decryption share as $h_i = \lambda_i a \cdot \mathsf{sk}_i - \lambda_i e_i + \tilde{e}_i$.
        	\item The decryption shares $h_i$ are sent to a designated output party,\footnote{Any client may act as the designated output party; the aggregator may do so only when it is also a receiver in that execution.} which computes:
{
\setlength{\mathindent}{0pt}
\begin{align*}
             d & \mbox{ } = \mbox{ } {\left[ b + \sum_{C_i \in \mathcal{C}}h_{i}\right]}_q \\ 
               & \mbox{ } = \mbox{ } f_\Sigma(b_1, \ldots, b_L) + \sum_{C_i \in \mathcal{C}} \left( \lambda_i a \cdot \mathsf{sk}_i - \lambda_i e_i + \tilde{e}_i \right) \mbox{ } \\ 
               & \mbox{ } = \mbox{ } -a\cdot \mathsf{sk} + e + \Delta \cdot m + (a\cdot \mathsf{sk}  - e + \tilde{e}) \mbox{ } \\ 
               & \mbox{ } = \mbox{ } \Delta \cdot \underbrace{m}_{f_\Sigma(m_1, \ldots, m_L)} + \underbrace{\tilde{e}}_{\sum_{i=1}^L{\tilde{e}_i}}.
\end{align*}
}
    	\end{itemize}
\end{itemize}

The final step for decryption depends on whether BFV or CKKS is used as the underlying secret-key HE scheme:
\begin{itemize}
    \item For BFV, the obtained decryption is $m = {[ \lfloor \frac{t}{q}d \rceil ]}_t$.
    \item For CKKS, the obtained decryption is $\tilde{m} =  \frac{d}{\Delta} = m + \frac{\tilde{e}}{\Delta} \approx m$ for large enough $\Delta$. Therefore, if $\|\tilde{e}_i\|\le B$ for every client, then $\|\tilde{m} - m\| \le LB/\Delta$ (see Section~\ref{sec:perf-eval}).
\end{itemize}

\textbf{Further practical considerations:} Our protocol presents several advantages over the baseline MHE-based protocol from Section~\ref{sec:mhe}. Since the second ciphertext component is shared across clients and remains unchanged under $\mathsf{ExtAdd}$, only one ciphertext component must be transmitted in each direction: clients send a component of their encrypted inputs and the aggregator returns a component of the aggregated ciphertext. As a result, the computational cost of aggregation is reduced by at least a factor of two. These estimates do not account for the overhead of public-key encryption compared to secret-key encryption, such as extra polynomial operations and larger ciphertext noise. Moreover, the setup phase does not require $\mathsf{cpk}$, but only a lightweight execution of the secret-sharing scheme $\mathsf{S}$ once during setup (see Appendix~\ref{app:aux-zero-sharing}).  Our construction also removes the dependence of ciphertext size on $\lambda$. This can further reduce the size of exchanged encrypted data, with concrete gains depending on the target precision, as quantified in Section~\ref{sec:perf-eval}. This contrasts with other secret-key approaches such as~\cite{PBCSSZ23}, where additional redundancy is required to compensate for the unknown noise in the aggregated ciphertext during decryption. The oracle $\mathcal{O}_i(\cdot)$ from Definition~\ref{def:oracle} avoids this overhead by allowing exact removal of the ciphertext noise, assuming that clients can trace the encryption errors associated with their inputs during a single aggregation round.

Finally, the protocol enables precomputation. The shared polynomial $a$ can be generated during setup, and partial decryptions can be precomputed since, unlike in MHE-based schemes, they do not depend on the aggregated result. Using a $\mathsf{PRF}$ to derive the $a$ values per round (see Section~\ref{sec:ask-he}) further reduces interaction, requiring only agreement on a random seed $\kappa$ and an update rule for $T$. Lower bounds for the modulus $q$ and a quantitative comparison with the MHE-based protocol are provided in Section~\ref{sec:perf-eval}. 

\subsection{Simulation-based security of our protocol}
\label{sec:proofoverview}

We now state the main security result for our aggregation protocol. Recall that Definition~\ref{def:mpcaggfl} formalizes the \emph{secure multiparty aggregation problem in federated learning}, where a set of clients $\mathcal{C}=\{C_1,\ldots,C_L\}$ holds private model updates $\vec{w}_i$ and jointly computes $\vec{w}^* = f_\Sigma(\vec{w}_1,\ldots,\vec{w}_L)$, or a randomized approximation of this aggregate, with the assistance of an aggregator. 
The ideal aggregation function $f_\Sigma$ corresponds to a linear combination $\sum_{i=1}^{L}\lambda_i \vec{w}_i$ with public coefficients
$\lambda_i \in \mathbb{Z}_q$. We assume $\gcd(q,\lambda_i)=1$ for all $i$, ensuring that the
coefficients admit multiplicative inverses in $\mathbb{Z}_q$ and avoiding non-invertible reconstruction coefficients in the homomorphic aggregation procedure.

Our goal is to ensure \emph{input privacy}: an adversary should learn nothing about the inputs of honest clients beyond what can be inferred from corrupted inputs and the final aggregated output. We consider a static \emph{semi-honest} adversary that may corrupt the aggregator and up to $L-1$ clients. If exactly one client remains honest, the adversary inevitably learns that client's contribution from the aggregate; this leakage is inherent to the functionality and falls outside the scope of secure computation. In practical FL deployments, protecting the aggregated model against inference attacks typically requires complementary mechanisms. Differential Privacy can be incorporated into the aggregation by injecting noise locally at the client side before encryption~\cite{WLDMYFJQP20,SCSSG23}. In particular, approximate HE schemes such as CKKS introduce numerical noise in the decrypted result, which may provide output perturbation and potentially be leveraged to obtain formal differential privacy guarantees~\cite{Ogilvie24}.

Our proof follows the approach of the threshold fully homomorphic encryption construction for secure multiparty computation from~\cite{AJLTVW12}, specialized to our aggregation functionality and simplified to the semi-honest setting. The argument relies on the homomorphic properties of RLWE-based encryption and its key-homomorphic structure, together with the smudging technique introduced in that work to prevent information leakage from partial decryptions. A key difference in our construction is the introduction of the client-local oracle $\mathcal{O}_i(\cdot)$ from Definition~\ref{def:oracle}. This functionality allows each client to locally recover the noise component associated with its ciphertext contribution during collaborative decryption. Consequently, the protocol can deterministically cancel the ciphertext noise without relying on large smudging distributions, simplifying the noise analysis while preserving the indistinguishability guarantees required by the simulation-based proof. We remark that, throughout the security argument, $\vec{ w}^* = f_\Sigma(\vec w_1,\ldots,\vec w_L)$ denotes the exact aggregate, which is the value made available to the simulator in the ideal execution. The BFV and CKKS instantiations differ only in output correctness: BFV returns $\vec{w}^*$ under its correctness condition, whereas CKKS returns a bounded randomized approximation of $\vec{w}^*$.

\begin{theorem}[Security of the aggregation protocol]\label{th:main-simulation-security} Assume the RLWE assumption holds for the parameters of the underlying single-key HE schemes. Then the protocol from Definition~\ref{def:multi-sk-he}, instantiated with either the exact or approximate HE schemes described in Section~\ref{sec:mask-he-concrete-instantiations}, provides simulation-based input privacy for the aggregation functionality $f_\Sigma$ from Definition~\ref{def:mpcaggfl} in the presence of a static semi-honest adversary that may corrupt the aggregator and a subset of at most $L-1$ clients. The BFV-based instantiation realizes $f_\Sigma$ exactly under its correctness condition, whereas the CKKS-based instantiation provides a randomized approximation of $f_\Sigma$ with bounded error.
\end{theorem}

\textbf{Proof overview:} The proof follows a standard simulation-based argument. We construct a simulator that interacts with the adversary in the ideal world using only the inputs of corrupted parties and the output of the ideal aggregation functionality, and produces messages computationally indistinguishable from those of a real execution.

By semantic security of the underlying RLWE encryption, ciphertexts of honest clients can be replaced by encryptions of zero without being detected by the adversary. During collaborative decryption, the simulator programs the partial decryption share of one honest client so that the reconstructed plaintext matches the output of the ideal functionality for BFV, while for CKKS it yields a bounded randomized approximation of that ideal output. In this step, the client-local oracle $\mathcal{O}_i(\cdot)$ plays a key role: since each client records the encryption noise, the protocol can explicitly cancel the aggregated ciphertext noise during decryption. This avoids the need for large smudging noise distributions while preserving indistinguishability of the decryption shares.

A full proof, including the sequence of hybrid arguments and the parameter constraints, is given in Appendix~\ref{app:proof-theorem1}.

\section{Performance evaluation}
\label{sec:perf-eval}
In the previous sections, we presented two approaches for extending single-key HE to the multiparty aggregation problem in FL of Def.~\ref{def:mpcaggfl}, under strong semi-honest adversaries that may control the aggregator and a majority of clients. The first is the MHE-based baseline of Section~\ref{sec:mhe}, where clients encrypt under a collective public key and jointly decrypt the aggregate. The second is the proposed construction of Section~\ref{sec:prop-opt}, which avoids the collective public key by letting each client encrypt under its own secret key.

This section compares the two approaches and evaluates the practical improvement obtained by the proposed construction. Subsections~\ref{sec:perf-q-bounds} and~\ref{sec:perf-noise-bounds} recall the correctness conditions and instantiate the effective noise bounds~used~to select the ciphertext modulus $q$. These bounds are then~used in Subsection~\ref{sec:perf-cipher-expansion} to quantify the reduction in ciphertext expansion. Finally, Subsections~\ref{sec:perf-bfv-ckks} and~\ref{sec:perf-runtimes} compare BFV/CKKS instantiations and report implementation runtimes.

Throughout Sections~\ref{sec:perf-noise-bounds}--\ref{sec:perf-runtimes}, we consider the unweighted sum, i.e., $\lambda_i=1$ for all $i$, with averaging performed after decryption. This is the most favorable setting for the MHE-based protocol, since evaluating $f_\Sigma$ requires only ciphertext additions. For nontrivial coefficients $\lambda_i$, each ciphertext component would require a scalar multiplication; as MHE ciphertexts contain twice as many polynomial components as ours, this would further increase their relative cost.

\subsection{Correctness bounds and modulus choice}
\label{sec:perf-q-bounds}

We first recall how the decryption correctness conditions from Tables~\ref{tab:bfv-ckks} and~\ref{tab:bfv-ckks-ashe} determine the minimum admissible ciphertext modulus $q$ for the aggregation protocols of Sections~\ref{sec:mhe} and~\ref{sec:prop-opt}. In particular, we use the multiparty derivations given in Proposition~\ref{prop:qmhe} of Appendix~\ref{app:bounds-q-modulus}.

Let $B_{\mathsf{ct}}^{\mathsf{MP}}$ denote an upper bound on the infinity norm of the effective noise tolerated at the end of the aggregation protocol. This notation is generic: depending on the protocol, it is instantiated either with the MHE-based bound or with the smaller bound of the proposed protocol (see Subsection~\ref{sec:perf-noise-bounds}). With this convention, the BFV-based and CKKS-based correctness conditions can be written as follows.

For BFV-based constructions, the correctness condition has the form $B_{\mathsf{BFV}} < \frac{q_{\mathsf{BFV}}}{2t}-\frac{t}{2}$, where $t$ is the plaintext modulus. Equivalently, it is sufficient to choose
\begin{equation}
q_{\mathsf{BFV}} > 2tB_{\mathsf{BFV}}+t^2.
\label{eq:perf-bfv-q}
\end{equation}

For CKKS-based constructions, assume that the aggregated plaintext $m = f_\Sigma(m_1, \ldots, m_L)$ satisfies $\lVert m\rVert < B_m$, where $m_i$ denotes the input plaintext of client $i$. The correctness condition is $\Delta_{\mathsf{CKKS}} B_m+B_{\mathsf{CKKS}} < \frac{q_{\mathsf{CKKS}}}{2}$, where $\Delta_{\mathsf{CKKS}}$ is the CKKS scaling factor. Writing the approximation error as $\epsilon=B_{\mathsf{CKKS}}/\Delta_{\mathsf{CKKS}}$ (see Section~\ref{sec:heblocks}), it suffices to choose
\begin{equation}
q_{\mathsf{CKKS}} > 2B_{\mathsf{CKKS}}\left(\epsilon^{-1}B_m+1\right).
\label{eq:perf-ckks-q}
\end{equation}
When inputs are normalized so that $B_m\leq 1$, the parameter $\epsilon^{-1}$ can be directly interpreted as the target precision of the approximate aggregation. For a fixed precision target, both BFV and CKKS therefore require a modulus essentially linear in the corresponding decryption noise bound. More precisely, in the parameter regimes where the precision term dominates the additive constants, we can write
\begin{equation}
q \simeq C \cdot B_{\mathsf{ct}},
\label{eq:perf-q-cb}
\end{equation}
with $C\simeq 2t$ for BFV and $C\simeq 2\epsilon^{-1}$ for CKKS, after absorbing $B_m$ into the CKKS precision parameter.

We keep the exact expressions~\eqref{eq:perf-bfv-q} and~\eqref{eq:perf-ckks-q} for the BFV/CKKS comparison in Subsection~\ref{sec:perf-bfv-ckks}. The simplified form~\eqref{eq:perf-q-cb} is used in Subsection~\ref{sec:perf-cipher-expansion} and Appendix~\ref{app:ciphertext-expansion-regimes} to compare the asymptotic ciphertext expansion of the proposed protocol and the MHE-based baseline.

\subsection{Effective noise bounds for HE-based protocols}
\label{sec:perf-noise-bounds}

We now instantiate $B_{\mathsf{ct}}^{\mathsf{MP}}$ for the two protocols. In both cases, $\lVert e_{\mathsf{ct}}\rVert < B_{\mathsf{ct}}$ denotes an upper bound on the infinity norm of the ciphertext noise before final decryption.

\textbf{MHE-based protocol:} For a fresh single-key ciphertext, Lemma~\ref{lemma:correct-decryption-single-key} in Appendix~\ref{app:bounds-q-modulus} gives the bound $\lVert e_{\mathsf{ct}}\rVert<(2n+1)B$, where $\|e\|\leq B$ for any $e\leftarrow\chi$ (see Table~\ref{tab:bfv-ckks}). In the MHE setting of Section~\ref{sec:mhe}, three additional effects must be accounted for. First, the collective public key is generated from $L$ secret-key shares, leading to a fresh-ciphertext noise bound of $B(2nL+1)$. Second, the homomorphic evaluation of $f_\Sigma$ adds $L$ ciphertexts, multiplying this term by $L$. Third, each client adds a smudging noise with $\lVert e_{\mathsf{smg},i}\rVert\mkern-2mu\leq\mkern-2mu B_{\mathsf{smg}}$ during collective decryption.

Therefore, the effective decryption noise in the MHE-based protocol is upper-bounded by
\begin{equation}
B_{\mathsf{ct}}^{\mathsf{MP}} = B_{\mathsf{mhe}} = B_{\mathsf{ct}}+L B_{\mathsf{smg}}
= \left(1+L2^{\lambda/2}\right)B_{\mathsf{ct}},
\label{eq:bound-b-mhe-perf}
\end{equation}
where $B_{\mathsf{ct}}=LB(2nL+1)$ and, as described in Section~\ref{sec:mhe}, $B_{\mathsf{smg}}=2^{\lambda/2}B_{\mathsf{ct}}$. Since $q=\mathcal{O}(B_{\mathsf{mhe}})$ for both BFV and CKKS, the high-variance smudging term increases the required modulus and, consequently, the communication and computational cost.

\textbf{Proposed protocol:} For the protocol from Section~\ref{sec:prop-opt}, Lemma~\ref{lemma:correct-decryption-single-secret-key} gives a fresh single-secret-key ciphertext noise bound $B_{\mathsf{ct}}=B$. Homomorphic evaluation of $f_\Sigma$ adds $L$ ciphertexts, so the noise grows by a factor $L$. During collaborative decryption, each client cancels the noise introduced in its own input ciphertext and replaces it with fresh noise from the same distribution. Thus, no smudging noise depending on the statistical security parameter $\lambda$ is required.

The corresponding effective decryption noise is therefore bounded by
\begin{equation}
B_{\mathsf{ct}}^{\mathsf{MP}} = B_{\mathsf{prop}} = L B.
\label{eq:bound-b-prop-perf}
\end{equation}
This is the main quantitative advantage of the proposed protocol: it removes the high-variance smudging term and the extra noise introduced by the collective public key.

\subsection{Ciphertext expansion comparison}
\label{sec:perf-cipher-expansion}

We compare the ciphertext expansion induced by the MHE-based and proposed protocols. For a fixed HE family, either BFV or CKKS, let $q_{\mathsf{mhe}}$ and $q_{\mathsf{prop}}$ denote the ciphertext moduli required by the two constructions, respectively. Ciphertext sizes are measured in bits, and both protocols are assumed to encrypt model updates of the same length.

From~\eqref{eq:bound-b-mhe-perf} and~\eqref{eq:bound-b-prop-perf}, the effective decryption-noise bounds are $B_{\mathsf{mhe}}=\left(1+L2^{\lambda/2}\right)LB(2nL+1)$ and $B_{\mathsf{prop}}=LB$. Using the approximation in~\eqref{eq:perf-q-cb}, $\log_2{q} \simeq \log_2{B_{\mathsf{ct}}} + \log_2{C}$, where $C$ depends on the selected HE family and precision target. For a model update of length $M$, both protocols require $\lceil M/n\rceil$ encrypted blocks when using the same ring degree $n$. The resulting ciphertext-size ratio between the MHE-based and proposed protocols can be approximated as
\begin{equation}
	R(C) = 2\frac{\log_2{B_{\mathsf{mhe}}} + \log_2{C}}{\log_2{B_{\mathsf{prop}}} + \log_2{C}}.
	\label{eq:perf-ratio-general}
\end{equation}
Here, the factor $2$ captures the reduction from two polynomial components to one, while the remaining factor captures the reduction in modulus size. Since $B_{\mathsf{mhe}}>B_{\mathsf{prop}}$, the ratio is largest when decryption noise dominates the moduli and approaches the structural factor $2$ as $\log_2{C}$ dominates.

Table~\ref{tab:perf-ratio-regimes} summarizes the behavior of~\eqref{eq:perf-ratio-general} in representative precision regimes, using parameters satisfying at least 128 bits of security following~\cite{BCCCCDGHKKLLMPPLSYY24}. The reported values are obtained directly from this expression and rounded over $L\in[5,100]$. Across these regimes, the dominant trend is driven by $\log_2{C}$, while the dependence on $L$ is mild; Appendix~\ref{app:ciphertext-expansion-regimes} provides a detailed derivation, simplifications, and regime-wise approximations. At low-to-moderate precision, the gain mainly comes from removing the smudging factor $2^{\lambda/2}$ and the collective-public-key noise term $(2nL+1)$. At intermediate precision, $\log_2{C}$ affects the proposed protocol, but $q_\mathsf{mhe}$ remains dominated by its larger decryption noise. In the very high precision limit, $\log_2{C}$ dominates both constructions and the relative gain tends to the structural factor $2$.

\begin{table}[!t]
    \centering
    \renewcommand{\arraystretch}{1.25}
    \caption{Approximate ciphertext expansion of the MHE-based construction with respect to the proposed construction.}
    \label{tab:perf-ratio-regimes}
    \scriptsize
    \setlength{\tabcolsep}{4pt}
    \begin{tabular}{p{0.2\columnwidth}|p{0.18\columnwidth}|p{0.16\columnwidth}|p{0.31\columnwidth}}
        \hline \hline
        \centering\arraybackslash\textbf{Precision regime} &
        \centering\arraybackslash\textbf{$\log_2 C$} &
        \centering\arraybackslash\textbf{Approx. $R(C)$}  &
        \centering\arraybackslash\textbf{Representative parameters} \\
        \hline \hline
        Low &
        $5$ &
        $\approx 13$--$16$ &
        $n=2^{12}$, $L \in [5,100]$, $B=19.2$, $\lambda=128$ \\
        Low/moderate &
        $10$ &
        $\approx 10$--$12$ &
        $n=2^{12}$, $L \in [5,100]$, $B=19.2$, $\lambda=128$ \\
        Moderate &
        $20$ &
        $\approx 8$ &
        $n=2^{13}$, $L \in [5,100]$, $B=19.2$, $\lambda=128$ \\
        Intermediate &
        $32$--$64$ &
        $\approx 4.3$--$6.3$ &
        $n=2^{13}$, $L \in [5,100]$, $B=19.2$, $\lambda=128$ \\
        Very high &
        $128$ &
        $\approx 3.3$ &
        $n=2^{14}$, $L \in [5,100]$, $B=19.2$, $\lambda=128$ \\
        \hline \hline
    \end{tabular}
    \vspace{-0.4cm}
\end{table}

\subsection{BFV versus CKKS for private aggregation}
\label{sec:perf-bfv-ckks}

The previous subsections compared the MHE-based and proposed protocols independently of the underlying HE family. We now compare BFV- and CKKS-based instantiations of both approaches. In~\cite{CSBBC24}, the authors argue that threshold variants of CKKS become less practical under the $\mathsf{IND}$-$\mathsf{CPA}^{\mathsf{D}}$ model because the high-variance smudging noise required for security may compromise the precision of the decrypted result. However, in~\cite{MPP24} we showed that this conclusion does not hold for all parameter regimes: CKKS-based private aggregation can be comparable to, and sometimes more compact than, BFV-based aggregation for practical precision targets. We extend that comparison here by considering our proposed protocol, where the variance of the noise used at decryption does not depend on $\lambda$.

The comparison relies on the correctness expressions from Subsection~\ref{sec:perf-q-bounds}, which determine the ciphertext modulus required by each instantiation. Consider a fixed aggregation protocol, either MHE-based or proposed, and let $B_{\mathsf{ct}}^{\mathsf{MP}}$ be the corresponding final effective noise bound, instantiated as $B_{\mathsf{ct}}^{\mathsf{MP}}\!=\!B_{\mathsf{mhe}}$ for the MHE-based protocol and $\mkern-1muB_{\mathsf{ct}}^{\mathsf{MP}}\!=\!B_{\mathsf{prop}}$ for the proposed protocol. Combining Eqs.~\eqref{eq:perf-bfv-q} and~\eqref{eq:perf-ckks-q}, CKKS requires a smaller ciphertext modulus than BFV whenever
\begin{equation}
B_{\mathsf{ct}}^{\mathsf{MP}} >
\frac{2\Delta_{\mathsf{CKKS}}B_m-t^2}{2(t-1)} .
\label{eq:perf-ckks-bfv-first}
\end{equation}
Equivalently, if the aggregated plaintext is normalized so that $B_m<1$ and the CKKS scale is expressed through $\epsilon=B_{\mathsf{CKKS}}/\Delta_{\mathsf{CKKS}}$, the condition becomes
\begin{equation}
\frac{t^2}{2B_{\mathsf{ct}}^{\mathsf{MP}}}+t-1 > \epsilon^{-1}.
\label{eq:perf-he-comparison}
\end{equation}
We refer to Proposition~\ref{prop:qmhe} in Appendix~\ref{app:bounds-q-modulus} for the derivation.

Equation~\eqref{eq:perf-he-comparison} provides a direct criterion for identifying the CKKS-favorable region. For small and medium values of $t$, the linear term $t-1$ is the main contribution, and BFV and CKKS behave similarly when compared at the same bit precision. For larger values of $t$, the quadratic term $t^2/(2B_{\mathsf{ct}}^{\mathsf{MP}})$ becomes relevant, and CKKS may require a smaller modulus than BFV for the same target precision. The position of this transition depends on the final effective noise bound. Since $B_{\mathsf{prop}}\ll B_{\mathsf{mhe}}$, the transition occurs earlier in the proposed protocol than in the MHE-based construction.

Figure~\ref{fig:varyingL_asym-sym} illustrates this effect by showing the CKKS-favorable regions of both protocols. We set $n\!=\!8192$ and $B\!=\!19.2$, and vary $L$ and $\lambda$. The axes represent the bit precisions $\log_2 t$ for BFV and $\log_2\epsilon^{-1}$ for CKKS. A point $(\log_2 t,\log_2\epsilon^{-1})$ is shaded when CKKS requires a smaller ciphertext modulus than BFV.

Let $\mathsf{R}_{\mathsf{mhe}}$ denote the CKKS-favorable region for the MHE-based protocol, and $\mathsf{R}_{\mathsf{prop}}$ the corresponding region for the proposed protocol. Since the proposed protocol has a smaller final effective noise bound, $\mathsf{R}_{\mathsf{mhe}}\subset \mathsf{R}_{\mathsf{prop}}$. As $L$ and $\lambda$ grow, the gap between both regions increases. This follows from the asymptotic behavior $B_\mathsf{mhe} = \mathcal{O}(L^32^{\lambda/2})$ and $B_\mathsf{prop} = \mathcal{O}(L)$. Therefore, the MHE-based protocol is affected by both the number of clients and the high-variance smudging term, whereas the proposed protocol only exhibits linear dependence on $L$ and avoids dependence on $\lambda$.

\begin{figure}
\centering
\subfloat[$L = 8$, $\lambda=32$.]{\includegraphics[width=0.46\columnwidth]{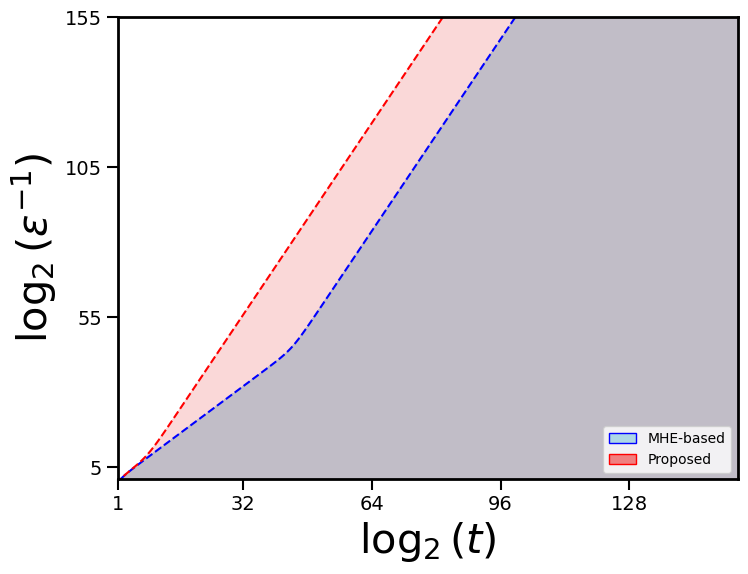}} \hspace{.2in}
\subfloat[$L=\lambda=128$.]{\includegraphics[width=0.46\columnwidth]{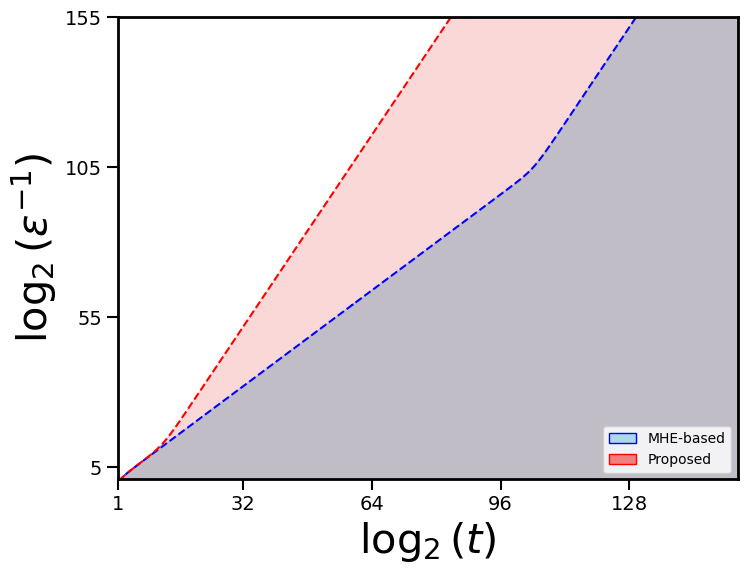}}
\caption{CKKS-favorable regions for the MHE-based protocol and for the proposed protocol.}
\label{fig:varyingL_asym-sym}
    \vspace{-0.4cm}
\end{figure}

As already discussed in Subsection~\ref{sec:perf-cipher-expansion}, the proposed protocol reduces the ciphertext size for both BFV and CKKS. The comparison in this subsection provides an additional conclusion: it also shifts the BFV/CKKS transition point to lower precision values, enlarging the practical parameter region in which CKKS is preferable to BFV. Additional parameter-region plots, heat maps quantifying the modulus gap between BFV and CKKS, and a discussion of approximate aggregation are provided in Appendix~\ref{app:bfv-ckks-details}.

\subsection{Implementation runtimes}
\label{sec:perf-runtimes}

We implemented the two main protocols in Go using Lattigo~\cite{MBTH20}. The MHE-based protocols~\cite{MTBH21} use Lattigo's multiparty package, whereas our proposed protocol builds on the low-level ring package and on the implementation of~\cite{PBCSSZ23}.\footnote{Code~\cite{PBCSSZ23} is available~at~\url{https://github.com/apedrouzoulloa/mkagg}.} We extended the latter to support BFV-based~and CKKS-based variants, and removed the rounding and rescaling steps required in~\cite{MPP24}; this is necessary to achieve the smaller ciphertext expansion predicted by the analysis above.

All experiments were run single-threaded on a system with 2× Intel Xeon Gold 6342 CPUs at 2.80~GHz and 1~TB of RAM. Implementation code is available~at~\url{https://github.com/miguelmorona/MSKHE-priavgagg}. We evaluate four aggregation families: the BFV-based and CKKS-based variants of both the MHE construction and our proposed construction. Precision is measured as $\log_2 t$ for BFV and as $\log_2\epsilon^{-1}$ for CKKS. Each parameter set of the proposed protocol was evaluated over $50$ independent runs, whereas each MHE parameter set was evaluated over $25$ independent runs. The reported runtimes are the corresponding averages.

We report the offline setup separately, since it is executed only once. The total runtime includes the encryption and collaborative-decryption costs of one client, together with the homomorphic aggregation cost. Unless otherwise stated, we set $L=32$ as a representative number of clients for the cross-silo setting considered here. All online runtimes correspond to synthetic model updates with $16{,}384{,}000$ randomly generated parameters per client. More generally, for a model update with $M$ parameters, runtimes scale approximately linearly with the number of ciphertexts, $\lceil M/n\rceil$. Thus, if $\mathsf{Runtime}(n)$ denotes the runtime for one ciphertext with ring degree $n$, the runtime for size $M$ is approximately $\left\lceil \frac{M}{n}\right\rceil \mathsf{Runtime}(n)$. For a fixed ring degree, this scaling factor is architecture-independent and depends only on the number of aggregated parameters. This scaling applies to both the proposed and MHE-based protocols.

The number of clients $L$ affects the protocol phases differently. For fixed cryptographic and model parameters, aggregation and the combination of partial decryption shares scale approximately linearly with $L$, whereas client-local encryption and the computation of each partial decryption share do not directly depend on $L$. Additional dependence arises indirectly through parameter selection and explicitly in the setup phase. Further details are provided in Appendix~\ref{app:setup-comparison}.

\textbf{Runtime of the proposed protocol:} We first report runtimes for the proposed protocol under different precision targets. Table~\ref{tab:exampleparams-parametersets-proposed} lists four parameter sets, with the corresponding runtimes in Table~\ref{tab:runtimes-proposed}. The BFV-based variant is slightly faster than the CKKS-based one for comparable modulus sizes. This difference is mainly due to implementation-level details in decryption. Although CKKS avoids the rounding and reduction modulo $t$ required by BFV, our CKKS implementation must reconstruct each decrypted value modulo $q$, represented over several limbs, before converting it into a floating-point value. When $q$ fits within a machine word, this reconstruction can be performed efficiently using native integer arithmetic. However, the overhead becomes significantly more pronounced when $q$ exceeds the machine-word size, as in parameter set $4$, since the reconstruction then requires multiprecision integer arithmetic. In contrast, the optimized BFV rounding available in Lattigo's ring package first reduces the number of limbs through a resize operation before recovering the plaintext, so that reconstruction only needs to operate over the modulus required to represent the plaintext. The magnitude of this effect should therefore be interpreted as implementation- and parameter-dependent.

\begin{table}[!htbp]
	\centering
	\renewcommand{\arraystretch}{1.3}
	\caption{Example parameter sets for the proposed private aggregation protocol ($\sigma=3.2$ and $128$ bits of security).}
    \label{tab:exampleparams-parametersets-proposed}
	\scriptsize
    \setlength{\tabcolsep}{4pt}
	\begin{tabular}{c|c|c|c|c}
        \hline \hline
		\textbf{Param.} &
        \textbf{Set $1$} $\mathsf{BFV}$ &
        \textbf{Set $2$} $\mathsf{CKKS}$ &
        \textbf{Set $3$} $\mathsf{BFV}$ &
        \textbf{Set $4$} $\mathsf{CKKS}$ \\
		\hline \hline 
    	$\{n,L\}$  & $\{2048,32\}$  & $\{2048,32\}$ & $\{4096,32\}$ & $\{4096,32\}$ \\
		$\{t,\epsilon^{-1},\Delta\}$ [bits]  & $\{21,-,21\}$ & $\{-,30,40\}$ & $\{46,-,46\}$ & $\{-,80,90\}$  \\
        $\{\#\mathrm{Limbs},q\}$ [bits]  & $\{2,42\}$ & $\{2,42\}$ & $\{2,92\}$ & $\{2,92\}$ \\
       
		\hline \hline 
	\end{tabular}
    \vspace{-0.1cm}
\end{table}

\begin{table}[!htbp]
	\centering
	\renewcommand{\arraystretch}{1.3}
	\caption{Implementation runtimes for the proposed private aggregation protocol. The total runtime excludes the offline setup.}
    \label{tab:runtimes-proposed}
	\scriptsize
    \setlength{\tabcolsep}{4pt}
	\begin{tabular}{c|c|c|c|c}
        \hline \hline 
		\textbf{Protocol step} & \textbf{Par. set $1$} & \textbf{Par. set $2$} & \textbf{Par. set $3$} & \textbf{Par. set $4$} \\
		\hline \hline
        Setup ($\mathsf{Client}$) &
        $14.12\,\mathrm{ms}$ &
        $14.84\,\mathrm{ms}$ &
        $36.04\,\mathrm{ms}$ &
        $41.03\,\mathrm{ms}$ \\	
		Encryption ($\mathsf{Client}$) &
        $2.56\,\mathrm{s}$ &
        $2.53\,\mathrm{s}$ &
        $2.53\,\mathrm{s}$ &
        $2.47\,\mathrm{s}$  \\
		Aggregation ($\mathsf{Agg}$) &
        $1.42\,\mathrm{s}$ &
        $1.43\,\mathrm{s}$ &
        $1.37\,\mathrm{s}$ &
        $1.38\,\mathrm{s}$  \\
        Coll. dec. ($\mathsf{Client}$) & 
        $2.13\,\mathrm{s}$ & 
        $2.34\,\mathrm{s}$ &
        $2.06\,\mathrm{s}$ & 
        $13.15\,\mathrm{s}$ \\
        \hline
        Total runtime & 
        $6.12\,\mathrm{s}$ &
        $6.35\,\mathrm{s}$ & 
        $5.97\,\mathrm{s}$ &
        $16.97\,\mathrm{s}$\\
		\hline \hline
	\end{tabular}
    \vspace{-0.3cm}
\end{table}

\textbf{Comparison with the MHE-based protocol:} We next compare the proposed protocol with the MHE-based construction following the blueprint of~\cite{MTBH21}. Table~\ref{tab:parametersets-comparison-mhe-prop} lists two parameter sets for each approach, and Table~\ref{tab:runtimes-comparison-mhe-prop} reports their corresponding runtimes. The higher runtimes of the MHE-based protocols are consistent with the larger modulus and ring degree required to accommodate the high-variance smudging countermeasure. This overhead is already visible during aggregation, whose cost also depends on the number of polynomial components in each ciphertext, and becomes more pronounced during collaborative decryption because of the cost of sampling and adding the large-variance smudging noise. The overhead is particularly significant for the CKKS-based MHE variant: preserving the target precision in the presence of such noise requires encoding the original floating-point plaintexts with a scale $\Delta$ whose bit length exceeds that of a single limb, thereby introducing additional encoding overhead. This additional cost does not arise in the BFV-based MHE variant implemented in Lattigo, where plaintext encoding can be performed efficiently using modular arithmetic. In contrast, our proposed construction benefits from a smaller modulus, single-component ciphertexts, and smudging noise whose variance is independent of $\lambda$. For the proposed protocol, we do not include optimizations based on precomputing encryption and partial-decryption material. Such optimizations would reduce its absolute online runtimes but would not affect the comparison in terms of ciphertext size or required modulus.

Regarding setup, although the MHE-based construction is faster, the additional cost of the proposed protocol remains small in practice. Table~\ref{tab:runtimes-comparison-setup} shows that its client-side setup remains below one second for all evaluated parameter sets even for $L=512$. Thus, even if repeated in every aggregation round, it would not become a dominant component of the per-round runtime for the evaluated model size. Further details on its scaling with $L$ are provided in Appendix~\ref{app:setup-comparison}.

\begin{table}[!htbp]
	\centering
	\renewcommand{\arraystretch}{1.3}
	\caption{Example parameter sets for the MHE-based and for the proposed protocols ($\sigma=3.2$, $\lambda=128$, and $128$ bits of security).}
    \label{tab:parametersets-comparison-mhe-prop}
	\scriptsize
    \setlength{\tabcolsep}{4pt}
	\begin{tabular}{c|c|c|c|c}
        \hline \hline
		\textbf{Param.} &
        \textbf{Set $1$} Prop. &
        \textbf{Set $2$} Prop. &
        \textbf{Set $3$} MHE &
        \textbf{Set $4$} MHE \\
		\hline \hline 
    	$\{n,L\}$  & $\{2048,32\}$  & $\{2048,32\}$ & $\{8192,32\}$ & $\{8192,32\}$ \\
		$\{t,\epsilon^{-1},\Delta\}$ [bits]  & $\{21,-,21\}$ & $\{-,30,40\}$ & $\{20,-,100\}$ & $\{-,20,118\}$  \\
        $\{\#\mathrm{Limbs},q\}$ [bits]  & $\{2,42\}$ & $\{2,42\}$ & $\{2,120\}$ & $\{2,120\}$ \\
		\hline \hline 
	\end{tabular}
    \vspace{-0.1cm}
\end{table}

\begin{table}[!htbp]
	\centering
	\renewcommand{\arraystretch}{1.3}
	\caption{Implementation runtimes for the proposed and MHE-based protocols. The total runtime excludes the offline setup.}
    \label{tab:runtimes-comparison-mhe-prop}
	\scriptsize
    \setlength{\tabcolsep}{4pt}
	\begin{tabular}{c|c|c|c|c}
        \hline \hline 
		\textbf{Protocol step} & \textbf{Par. set $1$} & \textbf{Par. set $2$} & \textbf{Par. set $3$} & \textbf{Par. set $4$} \\
		\hline \hline
        Setup ($\mathsf{Client}$) &
        $14.12\,\mathrm{ms}$ &
        $14.84\,\mathrm{ms}$ & 
        $2.20\,\mathrm{ms}$ &
        $1.75\,\mathrm{ms}$ \\	
		Encryption ($\mathsf{Client}$) &
        $2.56\,\mathrm{s}$ &
        $2.53\,\mathrm{s}$ & 
        $6.45\,\mathrm{s}$ & 
        $12.61\,\mathrm{s}$  \\
		Aggregation ($\mathsf{Agg}$) &
        $1.42\,\mathrm{s}$ & 
        $1.43\,\mathrm{s}$ &
        $2.44\,\mathrm{s}$ &
        $3.13\,\mathrm{s}$  \\
        Coll. dec. ($\mathsf{Client}$) &
        $2.13\,\mathrm{s}$ &
        $2.34\,\mathrm{s}$ & 
        $28.76\,\mathrm{s}$ &
        $37.29\,\mathrm{s}$ \\
        \hline
        Total runtime &
        $6.12\,\mathrm{s}$ &
        $6.35\,\mathrm{s}$ & 
        $37.65\,\mathrm{s}$ &
        $53.03\,\mathrm{s}$\\
		\hline \hline
	\end{tabular}
    \vspace{-0.3cm}
\end{table}

\textbf{Discussion of the performance results:} Overall, the proposed protocol reduces the communication overhead of MHE-based aggregation by removing the smudging term from the decryption noise bound and by transmitting a single polynomial component per encrypted update. The BFV/CKKS comparison remains parameter-dependent: BFV is preferable for exact or low-precision arithmetic, whereas CKKS becomes preferable as precision grows.

\section{Conclusions and future work}
\label{sec:conc-fut}
This work revisited RLWE-based Homomorphic Encryption for private average aggregation in Federated Learning, paying particular attention to the cross-silo setting and the security implications of granting access to aggregated decryptions. We analyzed single-key and multiparty HE approaches, highlighting the limitations imposed by restricted decryption access and the need, in conventional MHE constructions, for large-variance smudging noise.

To overcome this limitation, we proposed a lightweight multi-secret-key HE protocol tailored to average aggregation. Unlike MHE, our construction does not require a collective public key and lets each client encrypt under its own secret key. By compensating the ciphertext noise during decryption, the protocol avoids $\lambda$-dependent smudging noise while preserving security in the semi-honest model. We instantiated the approach with BFV and CKKS, showing that the proposed construction reduces communication overhead and improves runtimes over MHE-based alternatives.

Future work includes extending our construction scope, studying stronger adversarial models, and assessing integration into complete FL training pipelines, including its combination with Differential Privacy and output privacy.


\bibliographystyle{cas-model2-names}

\bibliography{biblio}


\appendix
\section{Bounds for the ciphertext modulus $q$}
\label{app:bounds-q-modulus}
This appendix collects technical bounds for determining the ciphertext modulus $q$ used in Section~\ref{sec:perf-eval}. We first derive bounds ensuring correct decryption for fresh public-key and secret-key BFV ciphertexts, then use these bounds to compare the ciphertext modulus required by the BFV and CKKS-based multi-client protocol constructions.

\begin{lemma}[Lemma $1$ in~\cite{MPP24}]\label{lemma:correct-decryption-single-key} We follow the notation for $\mathsf{BFV}$ indicated in Table~\ref{tab:bfv-ckks}, and assume that $e \leftarrow \chi$ satisfies $\|e\|\le B$. For a fresh ciphertext $\mathsf{ct} = (c_0, c_1)$, we have $[c_0+c_1s]_q=\Delta m+e_{\mathsf{ct}}$ with $\|e_{\mathsf{ct}}\|\le(2n+1)B$. This~implies that whenever $(2n+1)B<\tfrac{q}{2t}-\tfrac{t}{2}$, decryption works~correctly.
\end{lemma}
\begin{proof}
We start by computing $[c_0 + c_1 \cdot s]_q = \Delta m + \underbrace{u\cdot e + e_0 + e_1\cdot s}_{e_{\mathsf{ct}}}$, from which we can directly upper-bound the error polynomial $e_{\mathsf{ct}}$ as:
\begin{align}
\lVert e_{\mathsf{ct}} \rVert &= \lVert u\cdot e + e_0 + e_1\cdot s \rVert \\
&\leq \lVert u\cdot e \rVert + \lVert e_0 \rVert + \lVert e_1\cdot s \rVert \\
&\leq \delta_R B + B + \delta_R B \\
&\leq n B + B + n B \\
&= (2n+1)B. \label{eq:bound-error-bfv}
\end{align}
Now, the decryption process requires to multiply by $\tfrac{t}{q}$ and apply a final coefficient-wise rounding:
\begin{equation*}
    \left \lfloor\tfrac{t}{q} ([c_0 + c_1 \cdot s]_q) \right \rceil = \left \lfloor\tfrac{t}{q}(\Delta m + e_{\mathsf{ct}}) \right \rceil,
\end{equation*}
which must output $m$ for correct decryption. If we define $\Delta = \tfrac{q}{t} - r$ with $0\leq r < 1$,  the condition for correct decryption can be expressed equivalently as:
\begin{equation*}
    \left\Vert\tfrac{t}{q}(-rm + e_{\mathsf{ct}})\right \Vert < \tfrac 1 2.
\end{equation*}
By upper-bounding the expression in the left, we obtain:
\begin{equation}
\label{eq:bfv-cond-correctness}
\left\lVert \tfrac{t}{q}(-rm + e_{\mathsf{ct}}) \right\rVert \leq \tfrac{t}{q}(\lVert -rm \rVert + \lVert e_{\mathsf{ct}} \rVert) < \tfrac{t}{q}\left(\tfrac{t}{2}+\lVert e_{\mathsf{ct}} \rVert\right) < \frac{1}{2}.
\end{equation}
Therefore, combining the expressions~\eqref{eq:bound-error-bfv} and~\eqref{eq:bfv-cond-correctness}, we see that decryption correctness holds if $(2n + 1) B < \tfrac q {2t} -\tfrac{t}{2}.$
\end{proof}

\begin{lemma}[Secret-key analogue of the BFV correctness bound]\label{lemma:correct-decryption-single-secret-key} We follow the notation for secret-key $\mathsf{BFV}$ indicated in Table~\ref{tab:bfv-ckks-ashe}, and assume that $e \leftarrow \chi$ satisfies $\|e\| \leq B$. For a fresh ciphertext $\mathsf{ct} = (c_0,c_1) = \mathsf{E}_{\mathsf{ashe}}.\mathsf{Enc}(\mathsf{sk},m)$~under $\mathsf{sk} = s$, with $m\in R_t$ represented coefficient-wise in $(-t/2,t/2]$, we have $[c_0+c_1s]_q = \Delta m + e_{\mathsf{ct}}$ with $\|e_{\mathsf{ct}}\| \leq B$. This implies that whenever $B < \tfrac{q}{2t}-\tfrac{t}{2}$, decryption works correctly.
\end{lemma}
\begin{proof}
We start by computing the decryption phase for a fresh secret-key ciphertext. By the definition of $\mathsf{E}_{\mathsf{ashe}}.\mathsf{Enc}$ in Table~\ref{tab:bfv-ckks-ashe}, we have $(c_0,c_1)=(-a\cdot s+e+\Delta m, a)$ for some $a\leftarrow R_q$ and $e\leftarrow\chi$. Therefore $[c_0+c_1\cdot s]_q = [-a\cdot s+e+\Delta m+a\cdot s]_q = \Delta m+e = \Delta m + e_{\mathsf{ct}}$, and the ciphertext error satisfies
\begin{equation}
	\|e_{\mathsf{ct}}\| = \|e\| \leq B.
	\label{eq:bound-error-bfv-sk}
\end{equation}

The remaining correctness argument is identical to the one in Lemma~\ref{lemma:correct-decryption-single-key}. Namely, after multiplying by $\tfrac{t}{q}$ and applying coefficient-wise rounding, decryption is correct whenever $\left\| \tfrac{t}{q}(-rm+e_{\mathsf{ct}}) \right\| < \tfrac{1}{2}$, where $\Delta=\tfrac{q}{t}-r$ with $0\leq r<1$. Using $\|m\|\leq \tfrac{t}{2}$ and the bound in Eq.~\eqref{eq:bound-error-bfv-sk}, this condition is satisfied if $\tfrac{t}{q}\left(\tfrac{t}{2}+B\right)<\tfrac{1}{2}$. Equivalently, decryption correctness holds whenever $B < \tfrac{q}{2t}-\tfrac{t}{2}$.
\end{proof}

\begin{proposition}[Generalized version of Proposition 1 in~\cite{MPP24}]
\label{prop:qmhe}
Consider the descriptions of the primitives for either the MHE-based protocol (from Tables~\ref{tab:bfv-ckks} and~\ref{tab:mhempc}) or the proposed protocol (from Tables~\ref{tab:bfv-ckks-ashe} and~\ref{tab:multi-sk-he}), instantiated with their $\mathsf{BFV}$ and $\mathsf{CKKS}$ variants. Let $B_{\mathsf{ct}}^{\mathsf{MP}}$ denote the final effective noise bound entering the decryption correctness condition of the considered multi-client protocol. Then, the ciphertext modulus $q$ for $\mathsf{BFV}$ constructions is larger than the one for $\mathsf{CKKS}$ constructions if the following condition is satisfied:
\begin{equation*}
	B^{\mathsf{MP}}_{\mathsf{ct}}>\frac{2\Delta_\mathsf{CKKS} B_m - t^2}{2(t-1)}.
\end{equation*}
This condition applies to both protocols with the substitution $B_{\mathsf{ct}}^{\mathsf{MP}}=B_{\mathsf{mhe}}$ for the MHE-based protocol and $B_{\mathsf{ct}}^{\mathsf{MP}}=B_{\mathsf{prop}}$ for the proposed protocol. Section~\ref{sec:perf-bfv-ckks} uses this expression as Eq.~\eqref{eq:perf-ckks-bfv-first} to compare the cipher expansion of both schemes.
\end{proposition}
\begin{proof}
The derivation relies only on the BFV and CKKS correctness conditions, hence is protocol-independent once the final effective noise bound $B_{\mathsf{ct}}^{\mathsf{MP}}$ is fixed. From the correctness bounds in Subsection~\ref{sec:perf-q-bounds}, both schemes must satisfy:
\begin{itemize}
	\item $q_{\mathsf{BFV}}:  B^{\mathsf{MP}}_{\mathsf{ct}} < \frac{q}{2t} - \frac{t}{2}$.
	\item $q_{\mathsf{CKKS}}: \Delta_{\mathsf{CKKS}}B_m + B^{\mathsf{MP}}_{\mathsf{ct}} < \frac{q}{2}$.
\end{itemize}
Comparing both expressions, we have that $q_{\mathsf{BFV}} > q_{\mathsf{CKKS}}$ if:
\begin{equation*}
	2tB_{\mathsf{ct}}^{\mathsf{MP}} + t^2 > 2\Delta_{\mathsf{CKKS}}B_m + 2B_{\mathsf{ct}}^{\mathsf{MP}}.
\end{equation*}
Rearranging terms gives
\begin{equation*}
	2(t-1)B_{\mathsf{ct}}^{\mathsf{MP}} > 2\Delta_{\mathsf{CKKS}}B_m - t^2,
\end{equation*}
or equivalently,
\begin{equation*}
	B_{\mathsf{ct}}^{\mathsf{MP}} > \frac{2\Delta_{\mathsf{CKKS}}B_m - t^2}{2(t-1)}.
\end{equation*}
\end{proof}

\section{Ciphertext expansion regimes}
\label{app:ciphertext-expansion-regimes}
This appendix derives the approximations summarized in Table~\ref{tab:perf-ratio-regimes}. Starting from~\eqref{eq:perf-ratio-general},
\begin{equation*}
    R(C) =
    2\frac{\log_2 B_{\mathsf{mhe}}+\log_2 C}
    {\log_2 B_{\mathsf{prop}}+\log_2 C},
\end{equation*}
we distinguish three regimes based on the precision term $C$. 

Since Eq.~\eqref{eq:perf-ratio-general} depends only logarithmically on $B_{\mathsf{mhe}}$, $B_{\mathsf{prop}}$, and $C$, the regime approximations are determined by the relative sizes of $\log_2 B_{\mathsf{mhe}}$, $\log_2 B_{\mathsf{prop}}$, and $\log_2 C$.

\textbf{Low-to-moderate precision regime:} Assume first that $\log_2 C \ll \log_2 B_{\mathsf{prop}} < \log_2 B_{\mathsf{mhe}}$. Then
\begin{equation}
    R(C)\simeq
    2\frac{\log_2 B_{\mathsf{mhe}}}{\log_2 B_{\mathsf{prop}}}
    =
    2+
    2\frac{\log_2\left((1+L2^{\lambda/2})(2nL+1)\right)}
    {\log_2(LB)}.
    \label{eq:perf-ratio-low-c}
\end{equation}
In the range where $L2^{\lambda/2}\gg 1$ and $2nL\gg 1$, this becomes
\begin{equation}
    R(C)\simeq
    2+
    2\frac{\lambda/2+\log_2(2nL^2)}
    {\log_2(LB)}.
    \label{eq:perf-ratio-low-c-expanded}
\end{equation}
The gain is therefore driven by removing the smudging factor $2^{\lambda/2}$ and the collective-public-key term $(2nL+1)$.

Up to different constants, the same reasoning and conclusions also apply in the moderate-precision regime, where $C \approx B_\mathsf{prop} < B_\mathsf{mhe}$. In this case, the term $\log_2 C$ in Eq.~\eqref{eq:perf-ratio-general} can no longer be neglected. In the numerator, it contributes approximately $\log_2 B_\mathsf{prop}$, while in the denominator it gives
$\log_2 B_\mathsf{prop}+\log_2 C \approx 2\log_2 B_\mathsf{prop}$. This yields
\begin{equation}
    R(C)\simeq 2 + \frac{\log_2\left((1+L2^{\lambda/2})(2nL+1)\right)}{\log_2(LB)}.
    \label{eq:perf-ratio-moderate-c}
\end{equation}
Using again $L2^{\lambda/2}\gg 1$ and $2nL\gg 1$, we obtain
\begin{equation}
    R(C)\simeq 2 + \frac{\lambda/2+\log_2(2nL^2)}{\log_2(LB)}.
    \label{eq:perf-ratio-moderate-c-expanded}
\end{equation}

\textbf{Intermediate precision regime:} Assume now that we have $\log_2 B_{\mathsf{prop}} < \log_2 C < \log_2 B_{\mathsf{mhe}}$. Here, the denominator of~\eqref{eq:perf-ratio-general} is affected by the precision term, whereas the numerator still contains the large MHE noise bound. If $\log_2 C$ dominates $\log_2 B_{\mathsf{prop}}$, i.e., $\log_2 B_{\mathsf{prop}}\ll\log_2 C$, then
\begin{equation}
    R(C)\simeq 2 + 2\frac{\log_2 B_{\mathsf{mhe}}}{\log_2 C} \approx 2 + \frac{\lambda + 2\log_2 (2nL^3B)}{\log_2 C}.
    \label{eq:perf-ratio-intermediate}
\end{equation}

For precision targets in the range of roughly $40$--$64$ bits and $\lambda = 128$, this second term remains significant for common cross-silo values of $L$ and standard HE ring dimensions.

\textbf{Very high precision regime:} If $\log_2{C}\gg \log_2{B_{\mathsf{mhe}}}$, the precision term dominates both numerator and denominator in~\eqref{eq:perf-ratio-general}, so
\begin{equation}
    R(C)\to 2.
\end{equation}
This is the worst asymptotic regime for the proposed construction in terms of relative ciphertext expansion. Even in this case, the protocol preserves the structural advantage of transmitting one polynomial component instead of two.

\section{Additional BFV/CKKS comparison details}
\label{app:bfv-ckks-details}
This appendix complements the BFV/CKKS comparison of Subsection~\ref{sec:perf-bfv-ckks}. We include separate parameter-region plots for the MHE-based and proposed protocols, heat maps quantifying the modulus $\mathsf{gap}$ in bits, and a brief clarification on the role of approximate HE in private aggregation.

\subsection{Parameter regions for each protocol}

Figures~\ref{fig:perf-bfv-ckks-mhe} and~\ref{fig:perf-bfv-ckks-prop} plot Eq.~\eqref{eq:perf-he-comparison} for the MHE-based construction and for the proposed construction, respectively. In Figure~\ref{fig:perf-bfv-ckks-mhe}, we set $n=8192$ and $B=19.2$. The first row fixes $L=10$, with $\lambda\in\{32,128\}$, and the second fixes $\lambda=128$, with $L\in\{8,128\}$. In Figure~\ref{fig:perf-bfv-ckks-prop}, we set $n=8192$, $B=19.2$, and $\lambda=128$, with $L\in\{8,128\}$. These parameters are representative of cross-silo settings and compatible with at least $128$ bits of security for the represented precision ranges.

\begin{figure}
    \centering
	\subfloat[$L = 10, \lambda = 32$.]{\includegraphics[width=0.46\columnwidth]{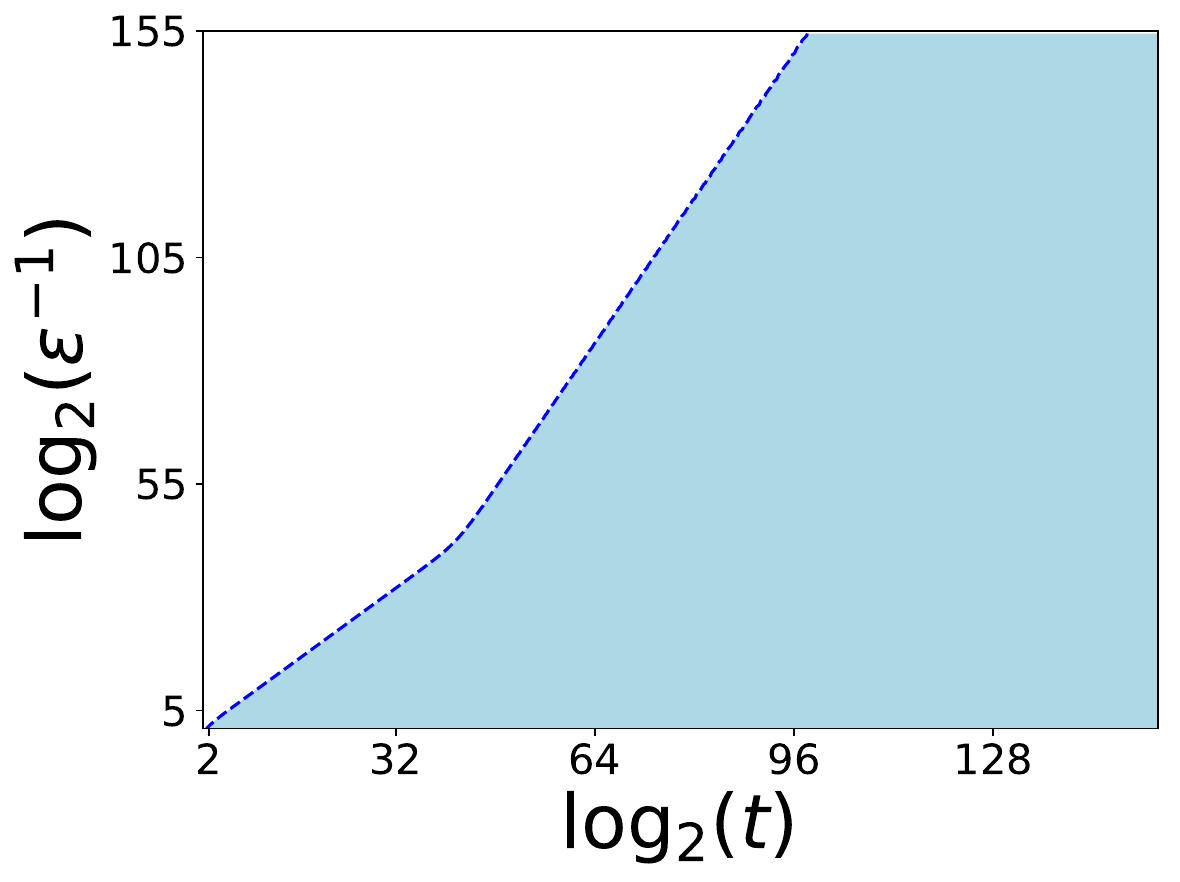}} \hspace{.2in}
	\subfloat[$L = 10, \lambda = 128$.]{\includegraphics[width=0.46\columnwidth]{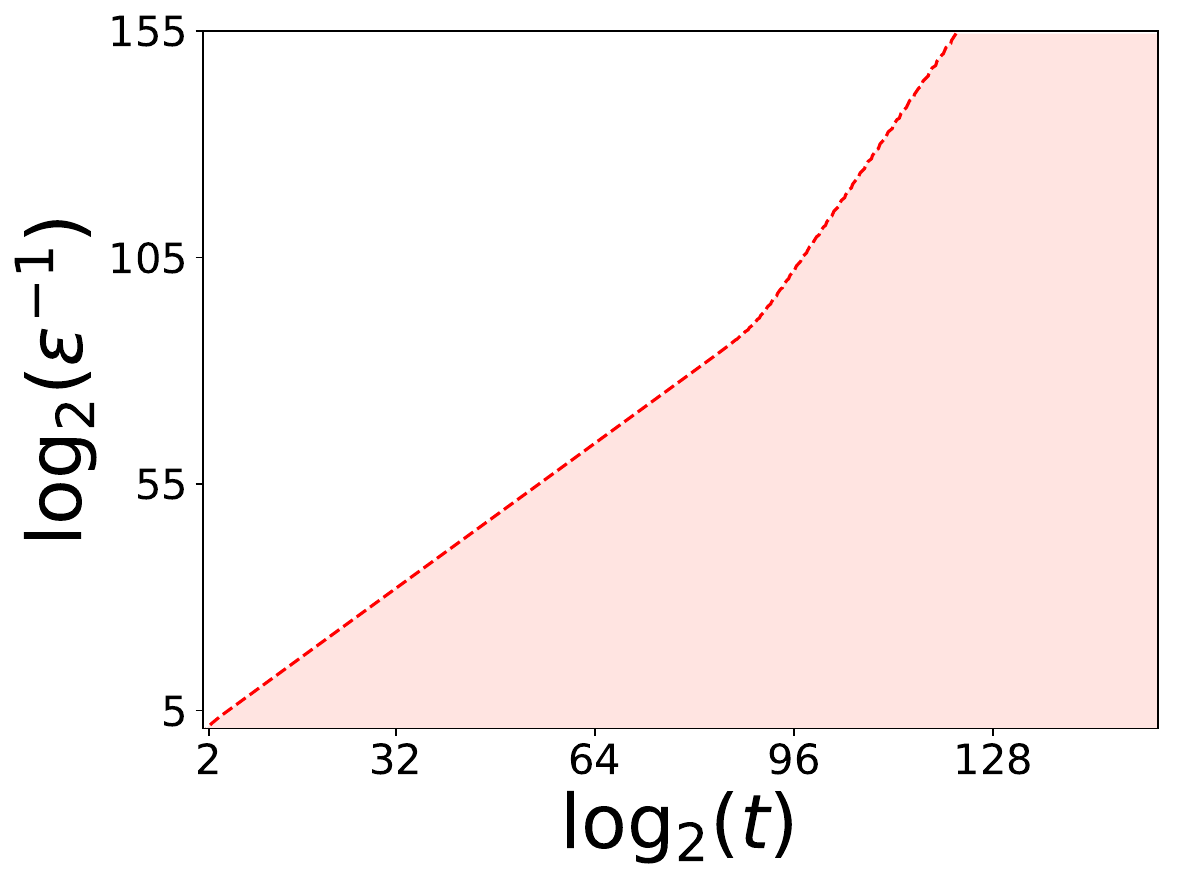}} \\
    \subfloat[$L = 8, \lambda = 128$.]{\includegraphics[width=0.46\columnwidth]{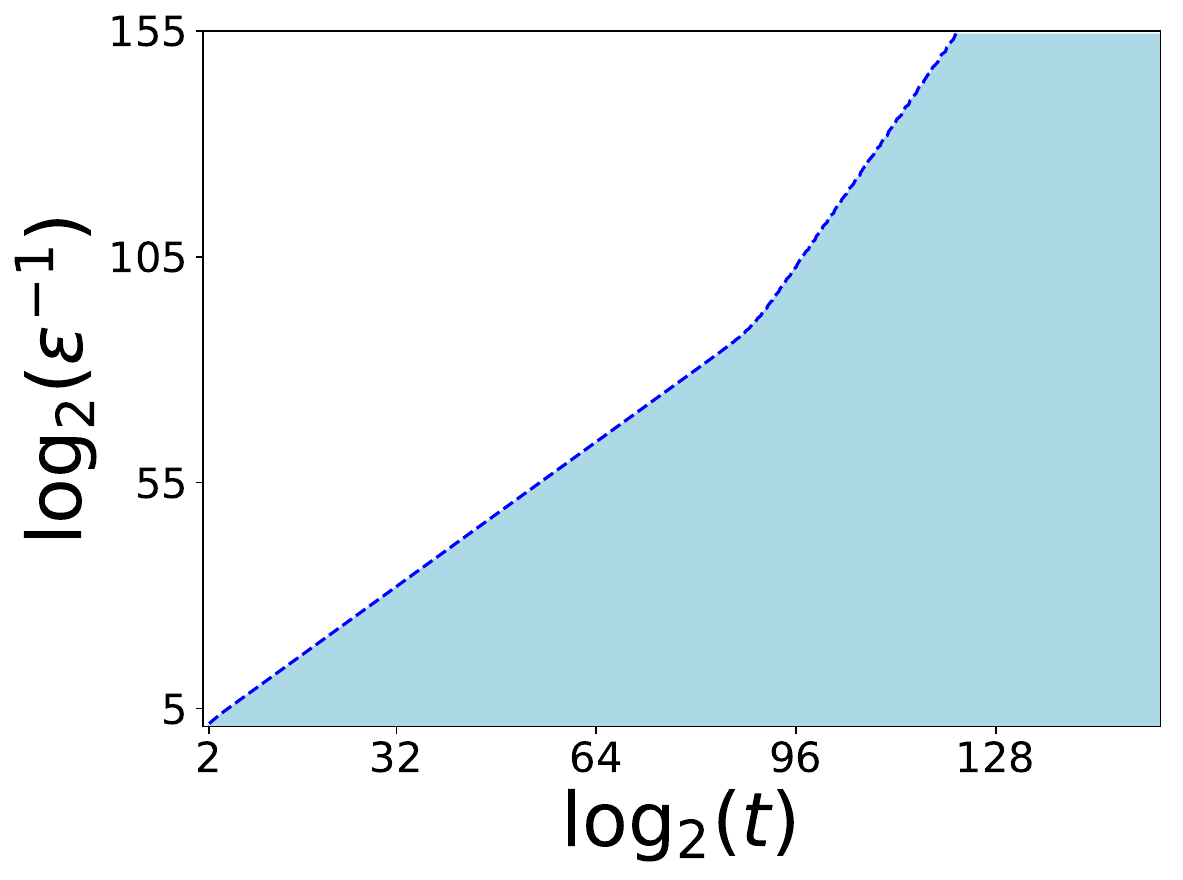}} \hspace{.2in}
	\subfloat[$L = 128, \lambda = 128$.]{\includegraphics[width=0.46\columnwidth]{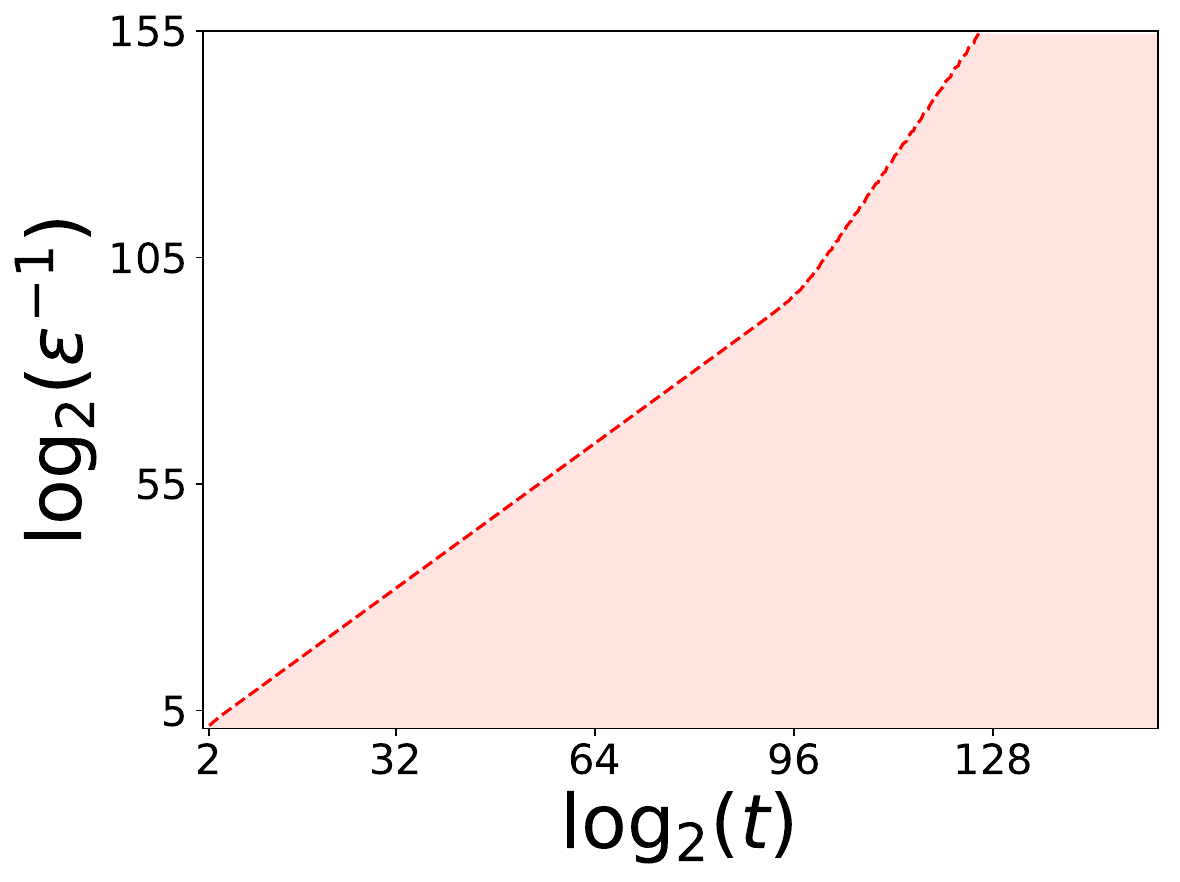}}
    \caption{CKKS-favorable regions for the MHE-based protocol.}
    \label{fig:perf-bfv-ckks-mhe}
\end{figure}

\begin{figure}
	\centering
	\subfloat[$L = 8$.]{\includegraphics[width=0.46\columnwidth]{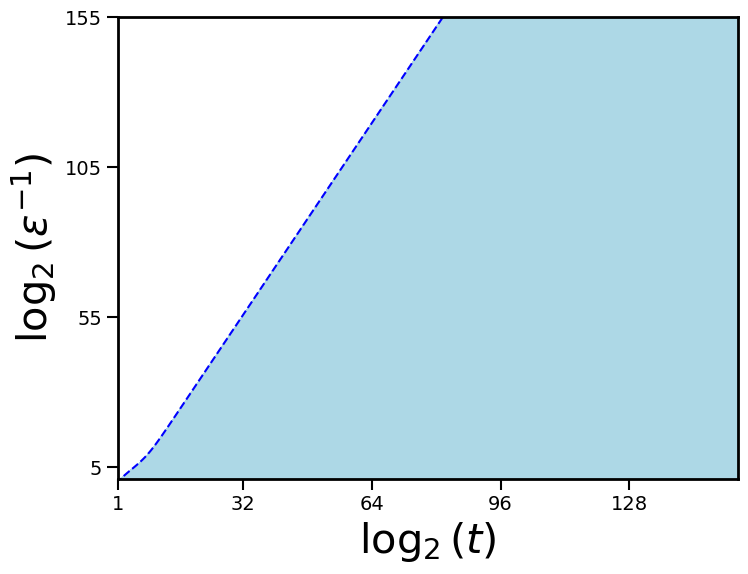}} \hspace{.2in}
	\subfloat[$L = 128$.]{\includegraphics[width=0.46\columnwidth]{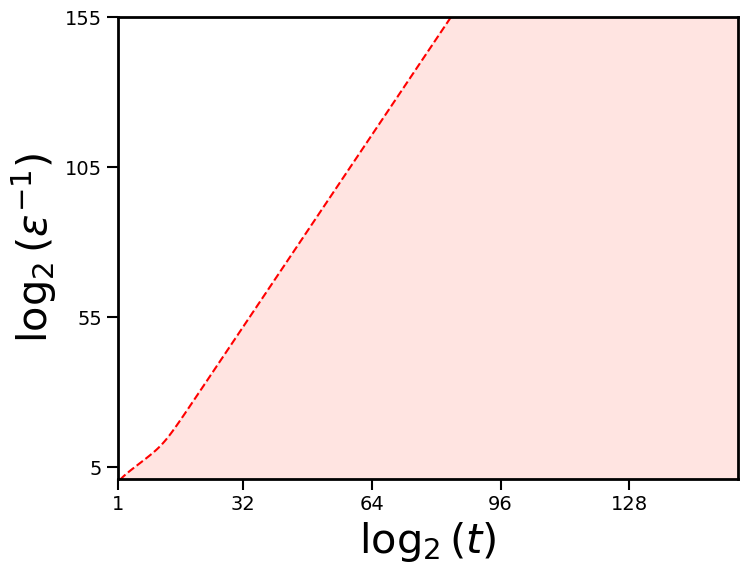}}
     \caption{CKKS-favorable regions for the proposed protocol.}
    \label{fig:perf-bfv-ckks-prop}
\end{figure}

The behavior in the figures follows directly from Eq.~\eqref{eq:perf-he-comparison}:
\begin{itemize}
\item When the term $t-1$ dominates, the boundary is approximately $
\log_2(t-1) > \log_2\epsilon^{-1}$.
\item When the quadratic term dominates, the boundary is approximately $2\log_2 t-\log_2 B_{\mathsf{ct}}^{\mathsf{MP}}-1 > \log_2\epsilon^{-1}$.
\end{itemize}
Thus, once the quadratic term becomes relevant, CKKS becomes increasingly favorable as the target precision grows. Increasing $\lambda$ affects only the MHE-based protocol through the smudging term. Increasing $L$ affects both protocols, but its impact is stronger for MHE because $B_{\mathsf{mhe}}$ contains the additional factor $(1+L2^{\lambda/2})\cdot(2nL+1)$.

\subsection{Modulus $\mathsf{gap}$ between BFV and CKKS}

We next quantify the modulus gap between the BFV-based and CKKS-based instantiations, for which we define $\mathsf{gap} = \log_2 q_\mathsf{BFV} - \log_2 q_\mathsf{CKKS}$. For a fixed public parametrization and for each point $(\log_2 t,\log_2\epsilon^{-1})$, we compute the minimum modulus required by both schemes using Eqs.~\eqref{eq:perf-bfv-q} and~\eqref{eq:perf-ckks-q}. Repeating this computation over the parameter grid gives the heat maps of $\mathsf{gap}$ shown in Figure~\ref{fig:varyingL_asym-sym-heat}.

\begin{figure}
\centering
\subfloat[Proposed protocol with $L = 128$.]{\includegraphics[width=0.46\columnwidth]{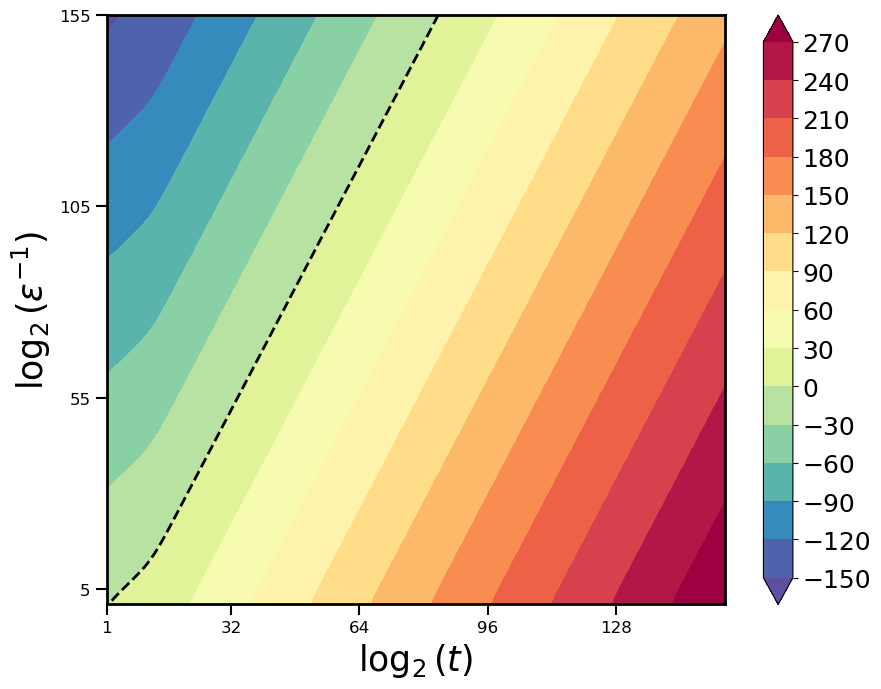}} \hspace{.2in}
\subfloat[MHE-based protocol with $L=\lambda=128$.]{\includegraphics[width=0.46\columnwidth]{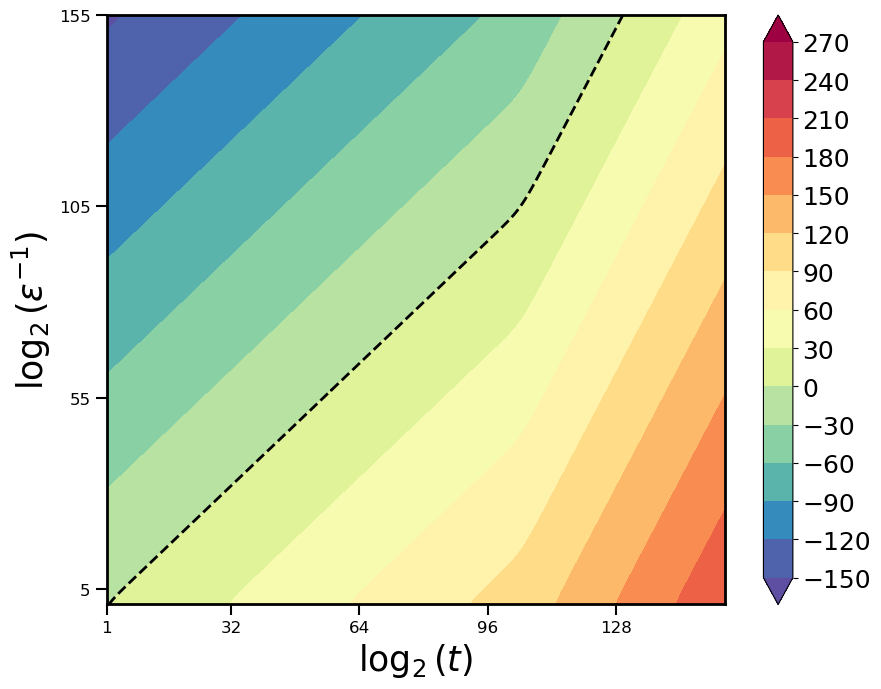}}
\caption{Heat map comparison between the BFV-based and CKKS-based instantiations of the proposed and MHE-based protocols. The color bands represent the value of $\mathsf{gap}$ in bits.}
\label{fig:varyingL_asym-sym-heat}
\end{figure}

The black dashed curve marks the separation between the BFV-favorable and CKKS-favorable regions. To interpret the behavior of the heat maps, we incorporate an explicit modulus $\mathsf{gap}$ into Eqs.~\eqref{eq:perf-ckks-bfv-first} and~\eqref{eq:perf-he-comparison}. This gives $\frac{t^2}{2B_{\mathsf{ct}}^{\mathsf{MP}}}+t-2^{\mathsf{gap}} \approx \epsilon^{-1}2^{\mathsf{gap}}$, and therefore
\begin{equation}
\epsilon^{-1}\left(t,\mathsf{gap}\right) \approx \frac{t^2}{2B_{\mathsf{ct}}^\mathsf{MP}2^\mathsf{gap}} + \frac{t}{2^\mathsf{gap}} -1.
\label{eq:epsilon-gap}
\end{equation}

The transition between the two main regimes occurs when the linear and quadratic terms in $t$ have comparable size, that is, $\frac{t}{2^{\mathsf{gap}}}\!\approx\!\frac{t^2}{2B_{\mathsf{ct}}^{\mathsf{MP}}2^{\mathsf{gap}}}$. At this transition point, the factor $2^{\mathsf{gap}}$ cancels out, yielding $t \approx 2B_{\mathsf{ct}}^{\mathsf{MP}}$. Hence, in logarithmic scale, the transition is located around $\log_2 t \approx 1+\log_2 B_{\mathsf{ct}}^{\mathsf{MP}}$, independently of $\mathsf{gap}$. Thus, the modulus gap does not shift the transition horizontally, but mainly produces a vertical displacement of the level curves. In the regime where the constant term in Eq.~\eqref{eq:epsilon-gap} is negligible, we have
\begin{equation*}
\log_2 \epsilon^{-1} \approx \log_2{\left(t+\frac{t^2}{2B_{\mathsf{ct}}^{\mathsf{MP}}}\right)} -\mathsf{gap}.
\end{equation*}
Equivalently, increasing $\mathsf{gap}$ shifts the curves downward by approximately $\mathsf{gap}$ bits in the $\log_2\epsilon^{-1}$ axis.

\subsection{Some remarks on approximate aggregation}

We close with a clarification on the BFV/CKKS comparison. Our analysis considers BFV decryption with multiplication by $t/q$, following~\cite{FV12}. When $t$ does not divide $q$, the rounding error causes the extra quadratic term in Eq.~\eqref{eq:perf-he-comparison}. If BFV decryption were instead modified to divide by a scaling factor $\Delta$, then Eq.~\eqref{eq:perf-he-comparison} would resemble the first approximately linear interval of Figures~\ref{fig:perf-bfv-ckks-mhe},~\ref{fig:perf-bfv-ckks-prop},~\ref{fig:varyingL_asym-sym}, and~\ref{fig:varyingL_asym-sym-heat}.

This does not remove the practical appeal of CKKS for private aggregation. Its approximate nature accounts for controlled error in the decrypted result, which can be useful when combined with output-privacy mechanisms applied to the released aggregate~\cite{MOJC23,MPP24}. BFV could emulate this by adding encoded noise homomorphically, but this would consume part of the plaintext-modulus budget. In CKKS, the target precision $\epsilon^{-1}$ already captures this trade-off.

\section{Indistinguishability Security Definitions}
\label{app:security-definitions}
The security of encryption schemes is commonly formalized through indistinguishability-based definitions. A fundamental notion for public-key encryption under passive attacks is indistinguishability under chosen-plaintext attack ($\mathsf{IND}$-$\mathsf{CPA}$), which is equivalent to semantic security \cite{S16}. Informally, it requires that an adversary cannot distinguish the encryptions of two chosen plaintexts with probability significantly better than random guessing.

\procedureblock{IND-CPA security game}{
\textbf{Challenger} \> \>  \textbf{Adversary} \\
(\mathsf{pk,sk}) \sample \mathsf{\Pi.KeyGen}(1^{\lambda}) \> \> \\
\> \xrightarrow{\hspace{5mm} \pk \hspace{5mm}}\\
\> \> m_0,m_1 \in \mathcal{M}\\
\>\xleftarrow{\hspace{2mm} m_0,m_1 \hspace{2mm}}\> \\
b \sample \bin \> \> \\
c_b = \mathsf{\Pi.Enc}_{\mathsf{pk}}(m_b) \> \> \\
\> \xrightarrow{\hspace{5mm} c_b \hspace{5mm} } \> \\
\> \> b' \in \bin \\
\> \> \text{Wins if } b = b'\\
}

\begin{definition}[IND-CPA Security]
\label{def:INDCPA} Let $\Pi$ be an encryption scheme. We say that $\Pi$ is $\mathsf{IND}$-$\mathsf{CPA}$ secure if for every probabilistic polynomial-time adversary $\mathcal{A}$,
\begin{equation*}
    \mathsf{Adv}^{\mathsf{IND\text{-}CPA}}_{\Pi}(\mathcal{A}) = \left|\Pr[\mathcal{A}\ \text{wins}] - \tfrac{1}{2}\right|
\end{equation*}
is negligible in the security parameter $\lambda$.
\end{definition}

\textbf{$\mathsf{IND}$-$\mathsf{CPA}^{\mathsf{D}}$ security:} For HE schemes, the $\mathsf{IND}$-$\mathsf{CPA}$ notion can be extended to capture the evaluation functionality. This leads to the $\mathsf{IND}$-$\mathsf{CPA}^{\mathsf{D}}$ security model, in which the adversary may additionally access an evaluation oracle $\mathcal{O}_{\mathsf{Eval}}$ and a restricted decryption oracle $\mathcal{O}_{\mathsf{Dec}}$. For general HE schemes, the evaluation oracle allows the adversary to evaluate arbitrary circuits over ciphertexts. In our setting, however, the supported homomorphic operation is addition, and therefore $\mathcal{O}_{\mathsf{Eval}}$ corresponds to sequences of homomorphic additions (i.e., repeated applications of $\mathsf{Add}$).

\procedureblock{$\mathsf{IND}$-$\mathsf{CPA}^{\mathsf{D}}$ security game}{
\textbf{Challenger} \< \hspace{20mm}  \textbf{Adversary} \\
(\mathsf{pk,sk,ek}) \sample \mathsf{\Pi.KeyGen}(1^{\lambda}) \< \\
b \sample \bin \< \\
\< \xrightarrow{\hspace{5mm} \pk \hspace{5mm}} \\
 \< \hspace{20mm} S = [ \hspace{1mm} ] \\
\< \hspace{20mm} \dbox {\begin{subprocedure}{
\text{Encryption Oracle}}
\end{subprocedure}} \\
\< \hspace{20mm} \dbox {\begin{subprocedure}{
\text{Evaluation Oracle}}
\end{subprocedure}} \\
\< \hspace{20mm} \dbox {\begin{subprocedure}{
\text{Decryption Oracle}}
\end{subprocedure}} \\
\< \hspace{20mm} b' \in \bin \\
\< \hspace{20mm} \text{Wins if } b' = b \< \< }

\procedureblock{Encryption Oracle $\mathcal{O}_{\mathsf{Enc}}(m_0,m_1,\pk,b,i)$ }{
\textbf{Adversary} \< \< \textbf{Oracle} \\
m_0, m_1 \in \mathcal{M} \< \< \\
 \< \sendmessageright*[1.5cm]{m_0,m_1} \< \\
 \< \< \mathsf{ct} \leftarrow \mathsf{Enc}(\pk,m_b)  \\
\< \< S[i] = (m_0,m_1,\mathsf{ct})  \\
\< \< i=i+1 \\
 \< \sendmessageleft*[1.5cm] {\mathsf{ct},i} \< 
}

\procedureblock{Evaluation Oracle $\mathcal{O}_{\mathsf{Eval}}(C,i_1,\ldots,i_k,\mathsf{ek},S)$ }{
\textbf{Adversary} \< \< \textbf{Oracle} \\
\{i_1,\dots,i_k\} \subseteq \mathcal{I} \< \< \\
C \in \mathcal{C} \< \< \\
 \< \sendmessageright*[1.5cm]{C,i_1,\dots,i_k} \< \\
 \< \< \mathsf{ct} \leftarrow \mathsf{Eval}(C,\mathsf{ek}, \\
 \< \< \hspace{18mm} S[i_1].\mathsf{ct},\ldots,S[i_k].\mathsf{ct}) \\
\< \< m_0' = C(S[i_1].m_0,\ldots, S[i_k].m_0)  \\
\< \< m_1' = C(S[i_1].m_1,\ldots, S[i_k].m_1)  \\
\< \< S[i] = (m_0',m_1',\mathsf{ct})  \\
\< \< i=i+1 \\
 \< \sendmessageleft*[1.5cm] {m_0',m_1',\mathsf{ct},i} \< 
}

\procedureblock{Decryption Oracle $\mathcal{O}_{\mathsf{Dec}}(i,\sk,S)$ }{
\textbf{Adversary} \< \< \textbf{Oracle} \\
i \in \mathcal{I} \< \< \\
 \< \sendmessageright*[1.5cm]{i} \< \\
 \< \< \textbf{if} \hspace{1mm} (S[i].m_0 = S[i].m_1) \hspace{1mm} \textbf{then} \\
\< \< \hspace{4mm} m \leftarrow \mathsf{Dec}(\sk,S[i].\mathsf{ct})  \\
\< \< \hspace{4mm} \mathsf{out} = m  \\
\< \< \textbf{else}  \\
\< \< \hspace{4mm} \textsf{out} = \bot   \\
 \< \sendmessageleft*[1.5cm] {\mathsf{out}} \< 
}

In this security game, the adversary attempts to guess the secret bit $b \in \bin$ chosen uniformly at random by the challenger. The adversary may query the encryption oracle $\mathcal{O}_{\mathsf{Enc}}$ arbitrarily many times. Each query stores the pair $(m_0,m_1)$ and the ciphertext $\mathsf{ct}=\mathsf{Enc}(\pk,m_b)$ in the list $S$. The adversary may also evaluate circuits using the evaluation oracle $\mathcal{O}_{\mathsf{Eval}}$. The resulting ciphertext and the corresponding evaluated messages are added to the list $S$. Finally, the adversary may query the restricted decryption oracle $\mathcal{O}_{\mathsf{Dec}}$. This oracle returns the decryption of $S[i].\mathsf{ct}$ only if the associated messages satisfy $S[i].m_0 = S[i].m_1$; otherwise it returns $\bot$.

\begin{definition}[$\mathsf{IND}$-$\mathsf{CPA}^{\mathsf{D}}$ Security] Let $\Pi=(\mathsf{Setup},$ $\mathsf{KeyGen},\mathsf{Enc},\mathsf{Dec},\mathsf{Eval})$ be a homomorphic encryption scheme. We say that $\Pi$ is $\mathsf{IND}$-$\mathsf{CPA}^{\mathsf{D}}$ secure if for every probabilistic polynomial-time adversary $\mathcal{A}$ with oracle access to $\mathcal{O}_{\mathsf{Enc}}$, $\mathcal{O}_{\mathsf{Eval}}$, and $\mathcal{O}_{\mathsf{Dec}}$,
\begin{equation*}
    \mathsf{Adv}^{\mathsf{IND\text{-}CPA^{D}}}_{\Pi}(\mathcal{A}) = \left| \Pr[\mathcal{A}\ \text{wins}] - \tfrac{1}{2}\right|
\end{equation*}
is negligible in $\lambda$.
\end{definition}

\section{Auxiliary Zero-Sharing Protocols}
\label{app:aux-zero-sharing}
We describe two simple protocols used to generate additive shares of zero among $L$ clients $C_i \in \mathcal{C}$. The first protocol~\cite{PBCSSZ23} produces uniformly random additive shares that sum to zero. The second protocol derives shares compatible with the zero-sharing LSSS defined in Definition~\ref{def:zero-sharing-lsss}.

The concrete protocols below assume pairwise private authenticated channels between clients; alternative secure realizations of $\mathsf{S}$ may also be used.
\newline

\textbf{Protocol 1: Additive Zero Sharing.} Given $L$ clients $C_i \in \mathcal{C}$, the parties generate additive shares $r_1,\dots,r_L \in R_q$ such that $\sum_{i=1}^{L} r_i = 0$. This is done as follows:
\begin{enumerate}
    \item Each client $C_i$ samples $(L-1)$ elements $r^{(i,j)} \xleftarrow{\$} R_q$ for all $j\in\{1,\ldots,L\}$ with $j\neq i$.
    \item Each client $C_i$ sets $r^{(i,i)} = - \sum_{j \neq i} r^{(i,j)}$, so that $\sum_{j=1}^{L} r^{(i,j)} = 0$.    
    \item Client $C_i$ sends $r^{(i,j)}$ to client $C_j$ over their private authenticated channel for all $j\neq i$.
    \item Each client $C_i$ computes its share $r_i = \sum_{j=1}^{L} r^{(j,i)}$. 
\end{enumerate}

It follows that the resulting shares satisfy $\sum_{i=1}^{L} r_i = 0$.
\newline

\textbf{Protocol 2: Zero Sharing with Reconstruction Coefficients.} Let $\lambda_1,\dots,\lambda_L \in \mathbb{Z}_q$ be the public reconstruction coefficients from Definition~\ref{def:zero-sharing-lsss}. The protocol is as follows:
\begin{enumerate}
    \item Run Protocol~1 to obtain additive shares $r_1,\dots, r_L$ such that $\sum_{i=1}^{L} r_i = 0$.
    \item Each client $C_i$ locally computes $\tilde r_i = \lambda_i^{-1} \cdot r_i$.
\end{enumerate}
The resulting shares $\tilde r_1,\dots,\tilde r_L$ satisfy $\sum_{i=1}^{L} \lambda_i \tilde r_i = \sum_{i=1}^{L} r_i = 0$, thus forming a valid zero-sharing by Definition~\ref{def:zero-sharing-lsss}.

\section{Scalability with the number of clients}
\label{app:setup-comparison}
For fixed cryptographic parameters and model size, the direct dependence on the number of clients $L$ only appears in operations that combine client contributions. Each client independently encrypts its input and computes its partial decryption share, so the local cost of both operations is independent of $L$ per ciphertext block. In contrast, aggregating the encrypted contributions and combining the partial decryption shares require processing $L$ client contributions, leading to a cost that grows linearly with $L$.

There is also an indirect dependence through parameter selection. For the proposed protocol, the effective decryption-noise bound is $B_{\mathsf{prop}}=LB$, as given in Eq.~\eqref{eq:bound-b-prop-perf}. Together with the BFV and CKKS correctness conditions in Eqs.~\eqref{eq:perf-bfv-q} and~\eqref{eq:perf-ckks-q}, this implies that, for a fixed precision target, the required ciphertext modulus $q$ grows linearly with $L$, while its bit length grows only logarithmically. In practice, this affects computation mainly when a larger $q$ requires an additional modulus limb or a larger ring degree. This dependence is substantially milder than in the MHE-based construction, whose effective noise bound in Eq.~\eqref{eq:bound-b-mhe-perf} grows as $\mathcal{O}(L^3 2^{\lambda/2})$.

The setup phase has the strongest explicit dependence on $L$. In the zero-sharing protocol of Appendix~\ref{app:aux-zero-sharing}, each client samples $L-1$ elements of $R_q$, exchanges one contribution with every other client, and combines the received shares. The resulting local computation and communication scale linearly with $L$, while the total pairwise communication scales quadratically. Table~\ref{tab:runtimes-comparison-setup} shows the practical impact of this scaling. Even for $L=512$, the setup of the proposed protocol remains below one second for all evaluated parameter sets. Although the MHE-based setup is faster in these experiments, setup is performed only once in the offline phase. Moreover, its cost remains small relative to the online computation for the evaluated model size, as shown in Table~\ref{tab:runtimes-comparison-mhe-prop}. Thus, even if setup were repeated in every aggregation round instead of being performed only once offline, it would not become the dominant runtime cost.

\begin{table}[!htbp]
	\centering
	\renewcommand{\arraystretch}{1.3}
    \begin{threeparttable}[b]
	\caption{Client-side setup runtime as a function of the number $L$ of participating clients for the proposed and MHE-based protocols. The first four columns of the proposed protocol use the four parameter sets of Table~\ref{tab:exampleparams-parametersets-proposed}, whereas the last two MHE-based columns use the two MHE parameter sets of Table~\ref{tab:parametersets-comparison-mhe-prop}. Reported values are average single-threaded runtimes under the experimental setting of Section~\ref{sec:perf-runtimes}.}
    \label{tab:runtimes-comparison-setup}
	\scriptsize
    \setlength{\tabcolsep}{1.5pt}
	\begin{tabular}{c|c|c|c|c|c|c|}
        \hline \hline 
		$L$ & \textbf{Set $1$}\tnote{*} & \textbf{Set $2$} & \textbf{Set $3$ Prop.} & \textbf{Set $4$ Prop.} & \textbf{Set $3$ MHE} & \textbf{Set $4$ MHE}\\
		\hline \hline
        $32$ &
        $14.12\,\mathrm{ms}$ &
        $14.84\,\mathrm{ms}$ & 
        $36.04\,\mathrm{ms}$ &
        $41.03\,\mathrm{ms}$ &
        $2.20\,\mathrm{ms}$ &
        $1.75\,\mathrm{ms}$ \\	
		$64$ &
        $39.52\,\mathrm{ms}$ &
        $48.60\,\mathrm{ms}$ & 
        $100.51\,\mathrm{ms}$ &
        $139.60\,\mathrm{ms}$ &
        $5.10\,\mathrm{ms}$ &
        $3.98\,\mathrm{ms}$ \\
        $128$ &
        $79.90\,\mathrm{ms}$ &
        $77.73\,\mathrm{ms}$ & 
        $184.21\,\mathrm{ms}$ &
        $193.44\,\mathrm{ms}$ &
        $8.10\,\mathrm{ms}$ &
        $5.13\,\mathrm{ms}$ \\	
		$256$ &
        $156.69\,\mathrm{ms}$ &
        $145.25\,\mathrm{ms}$ & 
        $375.37\,\mathrm{ms}$ &
        $354.61\,\mathrm{ms}$ &
        $12.29\,\mathrm{ms}$ &
        $9.27\,\mathrm{ms}$ \\
		$512$ &
        $277.68\,\mathrm{ms}$ &
        $293.02\,\mathrm{ms}$ & 
        $733.94\,\mathrm{ms}$ &
        $742.92\,\mathrm{ms}$ &
        $20.43\,\mathrm{ms}$ &
        $15.40\,\mathrm{ms}$ \\
        \hline \hline
	\end{tabular}
\begin{tablenotes}
\item[*] The parameter sets corresponding to the columns are, from left to right, Set 1 BFV, Set 2 CKKS, Set 3 BFV (Prop.), and Set 4 CKKS (Prop.) from Table~\ref{tab:exampleparams-parametersets-proposed}, and Set 3 BFV (MHE) and Set 4 CKKS (MHE) from Table~\ref{tab:parametersets-comparison-mhe-prop}.
\end{tablenotes}
\end{threeparttable}
\end{table}

\section{Proof of Theorem~\ref{th:main-simulation-security}}
\label{app:proof-theorem1}
\textbf{Proof strategy:} We prove Theorem~\ref{th:main-simulation-security} by showing that \emph{our protocol from Table~\ref{tab:multi-sk-he} provides simulation-based input privacy for the aggregation functionality $f_\Sigma$} (see Definition~\ref{def:secure-realize-functionality} below). This is done by running a simulation-based argument in the standard real/ideal paradigm, following~\cite{AJLTVW12} but tailored to our simpler semi-honest setting. Our proof is substantially simpler than the general framework of~\cite{AJLTVW12} for three reasons: (i) we only consider static semi-honest adversaries, (ii) the functionality is the fixed linear aggregation function $f_\Sigma$, and (iii) the use of symmetric-key encryption together with the client-local oracle $\mathcal O_i(\cdot)$ allows the real protocol to inject a fresh decryption noise locally per client, so that no $\lambda$-dependent large-variance smudging noise is needed in the simulation.

\begin{definition}[Simulation-based input privacy for $f_\Sigma$]
\label{def:secure-realize-functionality}
Let $\pi_{f_\Sigma}$ be a protocol for private aggregation between the clients $\mathcal{C}$ and the \textit{Aggregator} (see Table~\ref{tab:multi-sk-he} and Definition~\ref{def:multi-sk-he}). We say that $\pi_{f_\Sigma}$ \emph{securely realizes} the aggregation functionality $f_\Sigma$ in the semi-honest model if, for every static semi-honest adversary $\mathcal A$ that may corrupt the \textit{Aggregator} and an arbitrary subset of up to $L - 1$ clients, there exists a probabilistic polynomial-time simulator $\mathcal S$ such that, for every input tuple $\{\vec{w}_i\}_{i=1}^L$, the joint view of the adversary in the real execution of $\pi_{f_\Sigma}$, i.e., the collection of all inputs, randomness, internal states, and messages held by corrupted parties, is computationally indistinguishable from its view in an ideal execution where:
\begin{itemize}
    \item honest clients provide their inputs only to the ideal functionality $f_\Sigma$;
    \item corrupted parties provide their inputs through the adversary;
    \item the functionality computes the exact aggregate $\vec{w}^*=f_\Sigma(\vec{w}_1, \dots, \vec{w}_L)$, then output parties receive this value under the exact-output convention, or a randomized approximation of it under the CKKS output convention described below; and
    \item the simulator $\mathcal S$ interacts with $\mathcal A$ using only the inputs of the corrupted parties and the output $\vec{w}^*$, and produces a view for a distinguisher $\mathcal{Z}$ computationally indistinguishable from the real one.
\end{itemize}
\end{definition}
This notion captures \emph{input privacy} in the sense of Definition~\ref{def:mpcaggfl}: the adversary learns no additional information about the inputs of honest clients beyond what follows from its own inputs and the aggregation result.

\textbf{Output convention:} For the exact BFV instantiation, the protocol output coincides with $\vec{w}^*$ whenever the corresponding correctness condition holds. For the approximate CKKS instantiation, the protocol output is instead $\vec{w}^* + \Delta^{-1}\sum_{i=1}^L\bm{\tilde{e}}_i$ (see Eq.~\eqref{eq:programmed-final-output}). This distinction affects only correctness and the form of the protocol output. In both cases, input privacy is defined with respect to the exact aggregate $\vec{w}^*$, which remains the value made available to the simulator in the ideal execution.

\textbf{Blockwise encoding of client updates:} Following the blockwise convention introduced in Section~\ref{sec:mask-he}, we make explicit how the proof handles vector updates spanning several ciphertexts. Let $N_{\mathsf{ctx}} := \left\lceil \frac{\#\mathsf{ModelParams}}{n(\lambda)} \right\rceil$ denote the number of secret-key ciphertexts required to encode one client update, where $\#\mathsf{ModelParams}$ is the total number of model parameters and $n(\lambda)$ is the number of plaintext slots supported by the underlying RLWE scheme. Accordingly, whenever a client update $\vec w_i$ does not fit into a single ciphertext, it is represented blockwise by $N_{\mathsf{ctx}}$ ciphertexts. In that case, all encryptions, aggregations, and partial decryptions appearing below are understood componentwise over these $N_{\mathsf{ctx}}$ ciphertext blocks. For the $j$-th ciphertext block, the public component is sampled independently as $a_j \leftarrow R_q$, and, whenever the same secret key is reused, these public components are required to be pairwise distinct; more generally, \emph{ciphertexts generated under a common secret key always use fresh, non-repeated public components.}

\subsection{Indistinguishability of encryptions}
\label{app:ind-enc-dist}

We begin by formalizing the distinguishing advantage~\cite{KL14,AJLTVW12}, the central quantity in the hybrid arguments. For experiment distributions $H,H'$ over a common sample space,~i.e., producing outputs of the same form, and any PPT distinguisher $\mathcal{Z}$ taking samples from that space as input, define:
\begin{equation*}
    \Adv_{\mathcal Z}(H,H') := \left|\Pr[\mathcal Z(H)=1]-\Pr[\mathcal Z(H')=1]\right|.
\end{equation*}
Here, $\mathcal{Z}(H) = 1$ denotes the event that $\mathcal{Z}$ outputs $1$~when~given a sample drawn from $H$; distinguishability is measured by the gap between the corresponding output probabilities.

This notion satisfies the triangle inequality: for any three distributions $H_0,H_1,H_2$, it holds that
\begin{equation*}
    \Adv_{\mathcal Z}(H_0,H_2) \le \Adv_{\mathcal Z}(H_0,H_1) + \Adv_{\mathcal Z}(H_1,H_2).
\end{equation*}

This property, together with Proposition~\ref{prop:ind-enc-dist}, will be used repeatedly in our proof of Theorem~\ref{th:main-simulation-security} to justify the replacements of honest ciphertexts and partial decryptions across the sequence of hybrids of the simulator construction:
\begin{proposition}[Indistinguishability of hybrid encryption distributions]
\label{prop:ind-enc-dist}
Following Table~\ref{tab:bfv-ckks-ashe}, let $R_q$ be the ciphertext ring, let $\chi$ be the error distribution, and let $\mathsf{Enc}(a,\mathsf{sk},m)=(-a\cdot \mathsf{sk}+e+\Delta m,\;a)$ denote a fresh encryption of $m\in R_q$, where $\mathsf{sk} \in R_q$ is a valid secret key, $a \leftarrow R_q$ is uniform, and $e \leftarrow \chi$. Let $\epsilon_{\mathsf{RLWE}}(\lambda)$ denote the single-sample RLWE distinguishing advantage. Then the following results on $\Adv_{\mathcal{D}}(\cdot, \cdot)$ hold:
\begin{enumerate}
    \item For every $m\in R_q$ and every PPT distinguisher $\mathcal D$,
    \begin{equation*}
    \big|
    \Pr\!\left[\mathcal D(\mathsf{Enc}(a,\mathsf{sk},m))=1\right]
    -
    \Pr\!\left[\mathcal D(u,a)=1\right]
    \big|
    \le \epsilon_{\mathsf{RLWE}}(\lambda),
    \end{equation*}
    where $a,u \leftarrow R_q$ are uniformly random.

    \item For every fixed $m,m'\in R_q$, every pair of valid secret keys $\mathsf{sk},\mathsf{sk}' \in R_q$, and every PPT distinguisher $\mathcal D$,
    \begin{equation*}
        \begin{aligned}
        \big| \Pr\!\left[\mathcal D(\mathsf{Enc}(a,\mathsf{sk},m))=1\right]
        &{}-\Pr\!\left[\mathcal D(\mathsf{Enc}(a,\mathsf{sk}',m'))=1\right] \big| \\ 
        &\le 2\,\epsilon_{\mathsf{RLWE}}(\lambda).
        \end{aligned}
    \end{equation*}

    \item Let $\{\mathsf{ct}_i\}_{i=1}^M$ be a collection of $M$ valid encryptions of $m_i \in R_q$ under the same or different valid secret keys $\mathsf{sk}_{i} \in R_q$, where $\mathsf{ct}_i = (b_i,a_i)$ with $b_i = -a_i\cdot \mathsf{sk}_{i} + e_i + \Delta m_i$, and, for all $i \neq j$, $a_i \neq a_j$ if $\mathsf{sk}_{i} = \mathsf{sk}_{j}$. Then, for every PPT distinguisher $\mathcal D$,
    \begin{equation*}
        \begin{aligned}
        \bigg| \Pr\!\left[\mathcal D\big(\{(b_i, a_i)\}_{i=1}^M\big)=1\right]
    &{}-\Pr\!\left[\mathcal D\big(\{(u_i, a_i)\}_{i=1}^M\big)=1\right]
    \bigg|\\
    &\le M\,\epsilon_{\mathsf{RLWE}}(\lambda),
    \end{aligned}
    \end{equation*}
    where each $u_i \leftarrow R_q$ is uniformly random. Note also that, by applying result $\#2$, any two such collections with the same public components $\{a_i\}_{i=1}^M$ but possibly different plaintexts, noises, and/or secret keys are computationally indistinguishable with distinguishing gap at most $2M\epsilon_{\mathsf{RLWE}}(\lambda)$.
\end{enumerate}
\end{proposition}

\begin{proof} For Item 1, for any fixed $m$, the map $x\mapsto x+\Delta m$ is a bijection over $R_q$. Hence, the distribution $(-a\cdot \mathsf{sk}+e+\Delta m,\;a)$ is obtained from the RLWE distribution $(-a\cdot \mathsf{sk}+e,\;a)$ by adding the public offset $\Delta m$ to the first component. Likewise, if $u\leftarrow R_q$ is uniform, then $u+\Delta m$ is also uniform in $R_q$. Therefore, distinguishing $\mathsf{Enc}(a,\mathsf{sk},m)$ from $(u,a)$ is equivalent to distinguishing one RLWE sample from one uniform sample~\cite{LPR13}, which is bounded by $\epsilon_{\mathsf{RLWE}}(\lambda)$.

Item 2 follows from Item 1 and the triangle inequality, using the uniform distribution as an intermediate hybrid.

For Item 3, use a standard hybrid argument replacing~the $M$ samples one by one with uniform samples. Each replacement incurs loss at most $\epsilon_{\mathsf{RLWE}}(\lambda)$, and summing over positions gives $M\epsilon_{\mathsf{RLWE}}(\lambda)$. Applying the argument twice yields the bound $2M\epsilon_{\mathsf{RLWE}}(\lambda)$ between two such collections.
\end{proof}

\subsection{An overview of the simulator construction}
\label{app:sim-construction}

Let $\mathcal{A}$ be a static semi-honest adversary corrupting a subset of \textit{clients} $\mathcal{C}_A \subseteq \mathcal C$ and possibly also the \textit{Aggregator}. Let $\mathcal{C}_H=\mathcal{C}\setminus \mathcal{C}_A$ denote the subset of honest clients, with $|\mathcal{C}_H|\ge 1$. During several steps of the simulation, we fix one distinguished honest client $C_h\in\mathcal{C}_H$, whose simulated messages will be handled separately whenever required. Recall also that $f_\Sigma(\vec{w}_1,\ldots,\vec{w}_L)=\sum_{i}\lambda_i\vec{w}_i=\vec{w}^*$. We define $\vec{w}_A^*:=\sum_{i: C_i\in\mathcal{C}_A}\lambda_i\vec{w}_i$ and $\vec{w}_H^*:=\vec{w}^*-\vec{w}_A^*$. Observe that $\vec{w}_H^*$ can be computed by the simulator solely from the prescribed output $\vec{w}^*$ and the corrupted parties' inputs. We next describe the main building blocks used to construct a simulator $\mathcal{S}$ for the ideal execution.

\textbf{Access of the simulator:} Since corrupted parties are semi-honest, their internal state is part of the adversarial view. Thus, the simulator has access to: (i) the inputs of corrupted parties; (ii) their secret keys; and (iii) their randomness, including encryption and decryption noises.

\textbf{Client-local oracle:} Each honest client $C_i\in \mathcal C_H$ maintains a local oracle $\mathcal{O}_i(\cdot)$ storing the encryption noise used in the ciphertexts generated by $C_i$ (see Definition~\ref{def:oracle}). During decryption, this oracle is used to cancel the ciphertext-dependent noise and replace it with fresh noise sampled from the error distribution $\chi$. This feature is what allows our protocol, and therefore the simulator, to work with a fresh decryption noise independent of the ciphertext, without requiring a large $\lambda$-dependent smudging noise.
\newline

\textbf{\uline{Main simulation steps:}}
\begin{itemize}
\item \uline{Setup:} (i) For corrupted parties, the simulator uses their prescribed randomness and generates the corresponding keys honestly, preserving their real internal state. (ii) For honest parties, whose internal state is not visible to the adversary, the simulator generates fresh keys and randomness consistent with the protocol.
\item \uline{Input phase:} (i) For each honest client $C_i\in\mathcal{C}_H\setminus\{C_h\}$, the simulator first replaces the encryptions of $\vec{w}_i$ by encryptions of the same-length zero-vector $\vec{\mathbf{0}}$ under the corresponding secret key, while maintaining the associated local state needed by the oracle $\mathcal{O}_i(\cdot)$. For the distinguished honest client $C_h$, the simulator instead first replaces its encryptions by encryptions of $\lambda_h^{-1}\vec{w}_H^*$, which are subsequently replaced by encryptions of $\vec{\mathbf{0}}$. Later, all these encryptions are further replaced by uniformly random samples in $R_q^2$. (ii) Corrupted parties behave honestly on their actual inputs.
\item \uline{Aggregation and decryption:} Let $\vec{w}^*\!=\!f_\Sigma(\vec{w}_1,\ldots,\vec{w}_L)$ be the exact output returned by the ideal functionality. Then, during decryption:
\begin{itemize}
    \item the simulator computes the partial decryptions of corrupted parties honestly using their known keys and randomness;
    \item for all but one honest client, the simulator replaces their partial decryptions by several fresh uniform polynomials from $R_q$;
    \item for one distinguished honest client $C_h\in \mathcal C_H$,~the simulator programs its partial decryptions so~that the final reconstructed vector equals the addition of (i) the prescribed functionality output $\vec{w}^*$, scaled by $\Delta$, and (ii) an adequate aggregated decryption noise $\vec{\tilde{e}}$ which is independent of the honest inputs.
\end{itemize}
\end{itemize}

\subsection{Hybrid simulator experiments}
\label{app:hybrid-sim-exp}

We now define a series of hybrid games~\cite{KL14} used to prove the indistinguishability against $\mathcal{Z}$ of the real and ideal executions; i.e., $\mathsf{IDEAL}_{f_\Sigma,\mathcal{S},\mathcal{Z}} \overset{c}{\equiv} \mathsf{REAL}_{\pi,\mathcal{A},\mathcal{Z}}$, where $\overset{c}{\equiv}$ denotes computational indistinguishability. The output of each game is just that of the environment/distinguisher $\mathcal{Z}$. In particular, we define the following four hybrid experiments.
\newline

\noindent$\triangleright$\ \textbf{$\mathsf{HYB}_0$:} This is the experiment $\mathsf{REAL}_{\pi,\mathcal{A},\mathcal{Z}}$; i.e., the real execution of the protocol of Table~\ref{tab:multi-sk-he} with real honest inputs.
\newline

\noindent$\triangleright$\ \textbf{$\mathsf{HYB}_1$:} Same as $\mathsf{HYB}_0$, except that for every honest client $C_i\in \mathcal C_H$, its ciphertexts\footnote{Updates $\vec{w}_i$ are encoded blockwise as $\{\mathsf{block}_{\vec{w}_i, j}\}_{j=1}^{N_{\mathsf{ctx}}}$ and~encrypted across $\mkern-1mu N_{\mathsf{ctx}}\mkern-1mu$ ciphertexts per client; see the blockwise encoding remark above.}
\begin{equation*}
\mathbf{ct}_{\vec{w}_i} := \{c_{i,j}\}_{j=1}^{N_{\mathsf{ctx}}}
= \{\mathsf{E_{mshe}.Enc}(a_j,\mathsf{sk}_i+r_i, \mathsf{block}_{\vec{w}_i, j})\}_{j=1}^{N_{\mathsf{ctx}}}
\end{equation*}
are modified as follows:
\begin{itemize}
    \item For each $C_i\in\mathcal{C}_H\setminus\{C_h\}$, the ciphertexts $\mathbf{ct}_{\vec{w}_i}$ are replaced by $\{\mathsf{E_{mshe}.Enc}(a_j,\mathsf{sk}_i+r_i,\vec{\mathbf{0}})\}_{j=1}^{N_{\mathsf{ctx}}}$.
    \item For the distinguished honest client $C_h$, the ciphertexts $\mathbf{ct}_{\vec{w}_h}$ are instead replaced by encryptions of $\vec{v}_h:=\lambda_h^{-1}\vec{w}_H^*$, namely
        \begin{equation*}
            \mathbf{ct}_{\vec{v}_h}:=\{c_{h,j}\}_{j=1}^{N_{\mathsf{ctx}}}=\{\mathsf{E_{mshe}.Enc}(a_j,\mathsf{sk}_h+r_h,\mathsf{block}_{\vec{v}_h,j})\}_{j=1}^{N_{\mathsf{ctx}}}.
        \end{equation*}
\end{itemize}

\noindent$\triangleright$\ \textbf{$\mathsf{HYB}_2$:} Same as $\mathsf{HYB}_1$, except that the ciphertexts $\mathbf{ct}_{\vec{v}_h}$ of the distinguished honest client $C_h$ are first replaced by encryptions $\{\mathsf{E_{mshe}.Enc}(a_j,\mathsf{sk}_h+r_h,\vec{\mathbf{0}})\}_{j=1}^{N_{\mathsf{ctx}}}$. The partial decryptions of honest parties are then modified as follows:
\begin{itemize}
    \item for every honest client except one distinguished party $C_h\in \mathcal C_H$, its partial decryptions are replaced by an independent and uniformly random vector of polynomials $\bm{u} \leftarrow R_q^{N_{\mathsf{ctx}}}$;
    \item the remaining honest client's partial decryptions are programmed so the final reconstructed value equals
    \begin{equation} \label{eq:programmed-final-output}
    \Delta \cdot \vec{w}^*+\sum_{C_i\in \mathcal C_A}\bm{\tilde{e}}_i+\sum_{C_i\in \mathcal C_H}\bm{\tilde{e}}_i,
    \end{equation}
    where each polynomial component of $\bm{\tilde{e}}_i$ for all $C_i\in \mathcal C_H$ is sampled independently from $\chi$. For corrupted clients, the terms $\bm{\tilde{e}}_i$ are those determined by their prescribed randomness, whereas the honest-client terms are unknown to the adversary. Hence, BFV decoding yields $\vec{w}^*$ under its correctness condition, while CKKS rescaling yields $\vec{w}^* + \Delta^{-1}\sum_{i=1}^L\bm{\tilde{e}}_i$. Moreover, since $|\mathcal{C}_H| \ge 1$, the aggregated approximation error contains at least one contribution that is fresh and independent from the adversary's point of view.
\end{itemize}

\noindent$\triangleright$\ \textbf{$\mathsf{HYB}_3$:} Same as $\mathsf{HYB}_2$, except that the honest encryptions of $\vec{\mathbf{0}}$ are replaced by uniformly random polynomial samples in $R_q^2$. By construction, $\mathsf{HYB}_3$ coincides with $\mathsf{IDEAL}_{f_\Sigma,\mathcal{S},\mathcal{Z}}$; i.e., the ideal execution generated by the simulator $\mathcal{S}$ under the corresponding exact or approximate output convention defined above.

\subsection{Indistinguishability arguments}

Let $N_{\mathsf{enc0}}$ denote the number of honest ciphertexts replaced in the transition $\mathsf{HYB}_0\to\mathsf{HYB}_1$, let $N_{\mathsf{enc1}}$ denote the number of distinguished honest ciphertexts replaced in $\mathsf{HYB}_1\to\mathsf{HYB}_2$, let $N_{\mathsf{pd}}$ denote the number of honest partial decryptions replaced in $\mathsf{HYB}_1\to\mathsf{HYB}_2$, and let $N_{\mathsf{rand}}$ denote the number of honest ciphertexts replaced in $\mathsf{HYB}_2\to\mathsf{HYB}_3$. By repeated application of Proposition~\ref{prop:ind-enc-dist}, each of the above terms admits a bound linear in the number of replaced samples and in $\epsilon_{\mathsf{RLWE}}(\lambda)$, which is negligible in $\lambda$ under the RLWE assumption~\cite{LPR13}. We now bound each term.

\begin{claim}
$\mathsf{HYB}_0 \overset{c}{\equiv} \mathsf{HYB}_1$.
\end{claim}

\begin{proof}
The only difference between the two hybrids is that the honest encryptions of $\vec{w}_i$ are replaced by encryptions of $\vec{\mathbf{0}}$ for $C_i\in\mathcal{C}_H\setminus\{C_h\}$, and by encryptions of $\vec{v}_h$ for $C_h$, under the same secret keys $\mathsf{sk}_i+r_i$ and with the same public polynomials $a_j$. By Item 2 of Proposition~\ref{prop:ind-enc-dist}, each such replacement changes the adversarial view by at most $2\,\epsilon_{\mathsf{RLWE}}(\lambda)$. Therefore, by a standard hybrid argument~\cite{KL14} over the $N_{\mathsf{enc0}} = N_{\mathsf{ctx}}|\mathcal{C}_H|$ replaced honest ciphertexts,
\begin{equation*}
\Adv_{\mathcal Z}(\mathsf{HYB}_0,\mathsf{HYB}_1)
\le
2N_{\mathsf{enc0}} \cdot \epsilon_{\mathsf{RLWE}}(\lambda) = \ngl(\lambda).
\end{equation*}
\end{proof}

\begin{claim}
$\mathsf{HYB}_1 \overset{c}{\equiv} \mathsf{HYB}_2$.
\end{claim}

\begin{proof}
In $\mathsf{HYB}_1$, the ciphertexts of all honest clients except $C_h$ already encrypt $\vec{\mathbf{0}}$, while those of the distinguished honest client encrypt $\vec{v}_h$. In $\mathsf{HYB}_2$, the latter are also replaced by encryptions of $\vec{\mathbf{0}}$ under the same secret key $\mathsf{sk}_h+r_h$ and with the same public polynomials $a_j$.

By construction of the client-local oracle, the decryption shares contributed by honest parties contain fresh noise independent of the aggregated ciphertext. However, since each ciphertext and its corresponding partial decryption share secret-key and error terms, replacing the corresponding partial decryption by a uniform polynomial must be justified with respect to their joint distribution.

For this purpose, for every $C_i\in\mathcal{C}_H\setminus\{C_h\}$ and $j\in\{1,\ldots,N_{\mathsf{ctx}}\}$, write the $j$-th zero-encryption as $\mathsf{ct}_{i,j}=(b_{i,j},a_j)$, where $b_{i,j}=-a_j(\mathsf{sk}_i+r_i)+e_{i,j}$, and let $h_{i,j}=\lambda_i a_j\mathsf{sk}_i-\lambda_i e_{i,j}+\tilde e_{i,j}$ denote its corresponding real partial decryption. Define $y_{i,j}:=-a_j\mathsf{sk}_i+e_{i,j}$ and consider the public bijection $T_i(b_{i,j},h_{i,j})=(z_{i,j},h_{i,j})$, with $z_{i,j}:=\lambda_i b_{i,j}+h_{i,j}=-\lambda_i a_jr_i+\tilde e_{i,j}$ and inverse $b_{i,j}=\lambda_i^{-1}(z_{i,j}-h_{i,j})$. Let $u_{i,j},v_{i,j}\leftarrow R_q$ be independent and uniform. Then, three consecutive RLWE replacements, respectively for the prescribed distributions of $\mathsf{sk}_i$, $r_i$, and the effective key $\mathsf{sk}_i+r_i$, give

\begin{equation*}
\begin{aligned}
(z_{i,j},-\lambda_i y_{i,j}+\tilde e_{i,j})
& \overset{c}{\approx} (z_{i,j},u_{i,j}) \\
& \overset{c}{\approx} (v_{i,j},u_{i,j}) \\
& \overset{c}{\approx} (\lambda_i b_{i,j}+u_{i,j},u_{i,j}).
\end{aligned}
\end{equation*}

In the first transition, when $y_{i,j}$ is replaced by a uniform polynomial independent of $a_j$, the polynomial component $-\lambda_i y_{i,j}+\tilde e_{i,j}$ remains uniform even conditioned on $(a_j,r_i,\tilde e_{i,j})$, and is therefore independent of $z_{i,j}$; the second transition replaces the secret-dependent component $z_{i,j}$ of the RLWE sample associated with $r_i$ by a uniform polynomial; and the third restores exactly the original ciphertext distribution using RLWE under $\mathsf{sk}_i+r_i$. Applying $T_i^{-1}$ therefore yields $(b_{i,j},h_{i,j})\overset{c}{\approx}(b_{i,j},u_{i,j})$.

Importantly, the above transformations are applied to the joint honest view rather than marginally. By the zero-sharing reconstruction equation, $r_h$ is determined by the remaining zero-shares $\{r_i\}_{i\neq h}$; conditioned on the corrupted parties' setup view, their zero-shares are fixed, and any remaining dependence on the non-distinguished honest zero-shares is preserved in the transformed coordinates. For each block $j$, let $d_j$ denote the prescribed reconstructed value in that block. Since $d_j=\sum_k z_{k,j}$, this is enforced by setting $z_{h,j}=d_j-\sum_{k\neq h}z_{k,j}$ and then programming $h_{h,j}=z_{h,j}-\lambda_h b_{h,j}$. Since the coordinate transformations are bijective, this programming requires no additional hybrid and preserves the prescribed reconstructed value in Eq.~\eqref{eq:programmed-final-output}.

The programmed partial decryptions of the distinguished honest client are therefore distributed consistently with the target simulated execution. As in $\mathsf{HYB}_2$, the simulator samples the fresh polynomial vectors $\bm{\tilde e}_i$ independently from $\chi$ for all $C_i\in\mathcal{C}_H$; consequently, the programmed final share reveals no additional information beyond the allowed output.

By Item 2 of Proposition~\ref{prop:ind-enc-dist}, replacing the $N_{\mathsf{enc1}}=N_{\mathsf{ctx}}$ ciphertexts of $C_h$ changes the adversarial view by at most $2N_{\mathsf{enc1}}\cdot\epsilon_{\mathsf{RLWE}}(\lambda)$. Moreover, by repeated application of Item 1 of Proposition~\ref{prop:ind-enc-dist}, and Item 3 for the corresponding collection-wise hybrid, each replacement of a non-distinguished honest partial decryption requires the three RLWE transitions above and therefore changes the joint adversarial view by at most $3\epsilon_{\mathsf{RLWE}}(\lambda)$. Therefore, accounting for $N_{\mathsf{pd}}=N_{\mathsf{ctx}}(|\mathcal{C}_H|-1)$ replaced honest partial decryptions yields
\begin{equation*}
\Adv_{\mathcal Z}(\mathsf{HYB}_1,\mathsf{HYB}_2) \le \left(2N_{\mathsf{enc1}}+3N_{\mathsf{pd}}\right) \cdot\epsilon_{\mathsf{RLWE}}(\lambda) = \ngl(\lambda).
\end{equation*}
\end{proof}

\begin{claim}
$\mathsf{HYB}_2 \overset{c}{\equiv} \mathsf{HYB}_3$.
\end{claim}

\begin{proof}
The only difference is that the honest encryptions of $\vec{\mathbf{0}}$ are replaced by uniformly random ciphertexts in $R_q^2$. By Item 1 of Proposition~\ref{prop:ind-enc-dist}, each replacement of an honest encryption of $\vec{\mathbf{0}}$ by a uniform ciphertext changes the adversarial view by at most $\epsilon_{\mathsf{RLWE}}(\lambda)$. Therefore, by a standard hybrid argument over the $N_{\mathsf{rand}} = N_\mathsf{ctx}|\mathcal{C}_H|$ replaced ciphertexts,
\begin{equation*}
\Adv_{\mathcal Z}(\mathsf{HYB}_2,\mathsf{HYB}_3)
\le
N_{\mathsf{rand}}\epsilon_{\mathsf{RLWE}}(\lambda) = \ngl(\lambda).
\end{equation*}
\end{proof}

\subsection{Conclusion of the proof of Theorem~\ref{th:main-simulation-security}}
\label{app:conc-proof}

Combining the above bounds via the triangle inequality~\cite{KL14}, we obtain:
\begin{equation*}
\begin{aligned}
\Adv_{\mathcal Z}(\mathsf{HYB}_0, \mathsf{HYB}_3)
&\le \Adv_{\mathcal Z}(\mathsf{HYB}_0, \mathsf{HYB}_1) \\
&\quad + \Adv_{\mathcal Z}(\mathsf{HYB}_1, \mathsf{HYB}_2) \\
&\quad + \Adv_{\mathcal Z}(\mathsf{HYB}_2, \mathsf{HYB}_3) \\
&\le \bigl(2N_{\mathsf{enc0}} + 2N_{\mathsf{enc1}} \\
&\quad + 3N_{\mathsf{pd}} + N_{\mathsf{rand}}\bigr) \cdot \epsilon_{\mathsf{RLWE}}(\lambda) \\
&\le \ngl(\lambda).
\end{aligned}
\end{equation*}

Since $\mathsf{HYB}_0$ corresponds to the real execution and $\mathsf{HYB}_3$ to the ideal execution produced by the simulator $\mathcal S$, \emph{the protocol $\pi_{f_\Sigma}$ provides simulation-based input privacy for the aggregation functionality $f_\Sigma$ in the semi-honest model, with correctness interpreted according to the corresponding exact or approximate output convention.}

\textbf{Explicit bounds:} We now make the dependence on the number of honest clients and the RLWE distinguishing advantage explicit. If we write $L_H := |\mathcal C_H|$ for the number of honest clients, we obtain:
\begin{itemize}
\item \textbf{$\mathsf{HYB}_0\!\to\!\mathsf{HYB}_1$:} $N_{\mathsf{enc0}}\!=\!L_H N_{\mathsf{ctx}}$ ciphertexts replaced;
\item \textbf{$\mathsf{HYB}_1\!\to\!\mathsf{HYB}_2$:} $N_{\mathsf{enc1}}\!=\!N_{\mathsf{ctx}}$ ciphertexts and $N_{\mathsf{pd}}\!=\!(L_H - 1)N_{\mathsf{ctx}}$ partial decryptions replaced;
\item \textbf{$\mathsf{HYB}_2\!\to\!\mathsf{HYB}_3$:} $N_{\mathsf{rand}}\!=\!L_H N_{\mathsf{ctx}}$ ciphertexts replaced.
\end{itemize}

Substituting into the bound above yields:
\begin{equation*}
\Adv_{\mathcal Z}(\mathsf{HYB}_0, \mathsf{HYB}_3)
\le (6L_H - 1)\,N_{\mathsf{ctx}}\,
\epsilon_{\mathsf{RLWE}}(\lambda).
\end{equation*}

Using the trivial bound $L_H \le L$ and recalling that 
$N_{\mathsf{ctx}} = \left\lceil \frac{\#\mathsf{ModelParams}}{n(\lambda)} \right\rceil$, we obtain:
\begin{equation*}
\Adv_{\mathcal Z}(\mathsf{HYB}_0, \mathsf{HYB}_3)
\le
6L\left\lceil \frac{\#\mathsf{ModelParams}}{n(\lambda)} \right\rceil
\epsilon_{\mathsf{RLWE}}(\lambda).
\end{equation*}

\textbf{Security discussion:} Importantly, the distinguishing advantage of the overall simulation reduces directly to the RLWE distinguishing advantage $\epsilon_{\mathsf{RLWE}}(\lambda)$, up to a multiplicative factor independent of any additional noise parameters. In contrast to the protocol of~\cite{AJLTVW12}, our construction does not rely on $\lambda$-dependent large smudging noise to decorrelate the simulated decryption shares. Instead, the client-local oracle ensures that fresh decryption noise is injected independently of the ciphertexts, allowing all hybrid transitions to be justified solely under the RLWE assumption~\cite{LPR13}. As a result, the overall security follows directly from the hardness of RLWE, without introducing additional approximation losses stemming from smudging parameters. In particular, the CKKS approximation affects only correctness and the form of the output, while the simulation-based privacy bound remains negligible under RLWE.



\end{document}